\documentclass[
  aps,
  prb,
  reprint,
  superscriptaddress,
  nofootinbib
]{revtex4-2}

\usepackage{array,tabularx,booktabs}
\usepackage{needspace}
\usepackage{graphicx}
\usepackage{amsmath,amsfonts,amsthm,amssymb}
\usepackage{hyperref}
\hypersetup{colorlinks=true,linkcolor=blue,citecolor=blue,urlcolor=blue}

\newcommand{\mA}{\mathcal{A}}
\newcommand{\mL}{\mathcal{L}}
\newcommand{\mT}{\mathcal{T}}
\newcommand{\snorm}[1]{\left\lVert #1\right\rVert}
\newcommand{\abs}[1]{\left\lvert #1\right\rvert}
\newcommand{\diam}{\operatorname{diam}}
\newcommand{\dist}{\operatorname{dist}}
\newcommand{\supp}{\operatorname{supp}}
\newcommand{\id}{\mathrm{id}}
\newcommand{\dd}{\mathrm{d}}
\newcommand{\ii}{\mathrm{i}}
\newcommand{\locerr}{\mathcal{E}}

\newcommand{\tailmass}{\mT_{\Phi}}

\newcommand{\Tr}{\operatorname{Tr}}

\newcommand{\ev}[1]{\langle#1\rangle}

\newtheorem{definition}{Definition}[section]
\newtheorem{theorem}[definition]{Theorem}
\newtheorem{proposition}[definition]{Proposition}
\newtheorem{lemma}[definition]{Lemma}
\newtheorem{corollary}[definition]{Corollary}
\newtheorem{assumption}[definition]{Assumption}

\theoremstyle{remark}

\begin{document}

\title{Local observable errors from truncating interaction tails in gapped quantum lattice systems}
\author{Kangle Li}
\email{lklmr@nus.edu.sg}
\affiliation{Department of Physics, National University of Singapore, 117551, Singapore}

\begin{abstract}
    We bound the error in ground-state expectations of local observables caused by truncating the spatial tails of a gapped quantum lattice Hamiltonian. We show that the error is controlled by the interaction strength discarded near each site, rather than by the extensive norm of the omitted Hamiltonian.
    If the gap remains open along a path that removes the interaction tail, the resulting bounds are uniform in system size and extend, under suitable assumptions, to thermodynamic-limit ground states. The convergence rate reflects the decay of the interaction tail: it is algebraic for power-law interactions, superpolynomial for superpolynomial interactions, and exponential at any strictly smaller rate for exponentially decaying interactions. 
    For superpolynomial interactions, we also derive a direct infinite-volume estimate using automorphic equivalence. The results extend to isolated low-energy sectors and parity-even fermionic systems.
    For two-body interactions decaying as $r^{-p}$ in 
    $d$ dimensions, we prove an error bound $O(R^{-(p-d)})$ for $p>2d$, where $R$ is the truncation 
    range. We construct a gapped non-translation-invariant example that saturates this scaling, showing that the bound is optimal for
    the general class considered. Our results quantify when finite-range truncations faithfully reproduce 
    the local physics of gapped long-range systems.
\end{abstract}

\maketitle
\tableofcontents

\section{Introduction}
\label{sec:introduction}

Interactions with spatial tails occur in many quantum many-body systems \cite{Defenu_2023}.  Examples include dipolar quantum gases, Rydberg-atom arrays with dipole or van der Waals interactions, and trapped-ion simulators whose Coulomb-mediated spin couplings are long ranged \cite{Lahaye_2009,Saffman_2010,Monroe_2021}.  Depending on the microscopic mechanism, the interaction can decay algebraically, exponentially, or faster than any inverse power while remaining formally infinite range.  Such tails can modify propagation and correlations \cite{Hastings-Koma_2006,Matsuta_2016,Kuwahara_2021,Wang_Hazzard_2023}, as well as entanglement and phase constraints \cite{Kuwahara_2020,Liu_2025}, and have motivated extensive work on locality in long-range systems.

Analytical and numerical treatments commonly replace an infinite-range interaction by a finite-range or otherwise compressed approximation.  Matrix-product-operator constructions, for example, give explicit efficient representations and approximations of long-range Hamiltonians for ground-state and time-evolution calculations \cite{Crosswhite_2008,Pirvu_2010, Froewis_2010}.  A hard range cutoff likewise reduces the number of interaction terms and can simplify analytical estimates, numerical calculations, and implementations with restricted connectivity. 

This raises a natural question: when, and in what sense, does a range-truncated model accurately approximate the original long-range interaction? Truncation quality can be assessed at several levels, including changes in the spectrum, global ground-state fidelity, and errors in expectation values of local observables. A natural first choice is global fidelity. However, fidelity is generally too stringent for a volume-independent truncation estimate: small local differences can accumulate throughout an extensive system, so the ground-state overlap between the full and truncated interactions may vanish as the volume grows even when their reduced states on every fixed finite region remain close. 

\begin{figure}[t]
  \centering
  \includegraphics[width=\columnwidth]{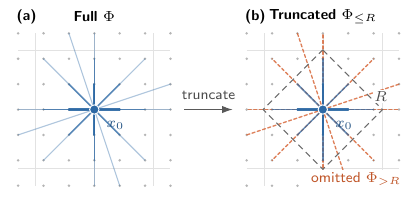}
  \caption{
    Range truncation illustrated schematically for a two-body interaction
    in two dimensions.
    (a) From an arbitrary reference site $x_0$, the full interaction
    contains couplings over all distances; bond opacity represents their
    decreasing norms.
    (b) The range-$R$ interaction $\Phi_{\le R}$ retains bonds with
    $\dist(x_0,y)\le R$, while the dashed orange bonds form the omitted
    tail $\Phi_{>R}$.
    The dashed diamond is the graph-metric cutoff boundary.
  }
  \label{fig:interaction_truncation}
\end{figure}

Here we focus on the last criterion: errors in ground-state expectations of local observables. 
Intuitively, sufficiently weak interaction tails should have only a limited effect on observables supported in a fixed finite region \cite{Hastings-Koma_2006,hastings2010localityquantumsystems, Gong_2014, Foss_Feig_2015, Matsuta_2016, Tran_2019, Tran_2020_hierarchy, Else_2020}.
The concrete question is: for an observable supported on a fixed finite region, how rapidly does the difference between its ground-state expectation values of the full and truncated interactions decrease as truncation radius grows?

In this work, we systematically study the error bounds on ground-state expectations of arbitrary local observables induced by truncating gapped interactions at range $R$, for a broad class of local interactions with tails.
For a unique ground state, or for a specified pair of matched states in isolated low-energy sectors, our bounds control the trace-norm distance between the reduced density matrices on any fixed finite region. More generally, for finite-dimensional low-energy sectors separated by a gap, they control the difference between the matrices of a local observable compressed to the two sectors after spectral-flow identification. The same estimate extends to ordered products of finitely many local operators: the corresponding moment bound depends on the total size of their supports, but not on the separations between the insertion regions. This local notion of accuracy remains physically meaningful even when the overlap between the global many-body states becomes small as the system size increases.

Our results are formulated separately in finite and infinite volume, with the latter treated using operator-algebraic methods and Gelfand–Naimark–Segal (GNS) representations \cite{bratteli2012operator,Naaijkens_2017,Nachtergaele_2019}. Under the corresponding uniform-gap assumptions, the finite-volume truncation error in a local ground-state expectation is controlled by the discarded interaction strength near each site. Consequently, power-law interaction tails yield algebraically decaying errors as the truncation radius increases, with the rate inherited from the local tail, whereas superpolynomial and exponential tails yield correspondingly rapid decay.

For infinite-volume setting, we also obtain results for superpolynomial interactions, within the spectral-flow framework of Becker, Teufel, and Wesle \cite{becker2025automorphicequivalencegappedphases}.  This result assumes
a prescribed differentiable path of locally unique GNS-gapped states and
gives errors smaller than every inverse power of truncation radius $R$. In particular, if the interaction has exponentially small tail, then the error bound will also be exponentially small in truncation radius.

These conclusions use two logically independent gap assumptions.  A
gapped direct path restores the complete discarded tail and compares the
range-$R$ and full systems in one step.  A gapped shell path instead
compares two finite cutoffs.  Concatenating compatible shells does not
improve the truncation rate, but when the successive errors are summable
it constructs a convergent infinite-volume ground-state branch.  
This shell flow limit construction assumes ground states along the path.

Our finite-volume and shell estimates use quasi-adiabatic continuation,
or spectral flow
\cite{hastings2010quasiadiabaticcontinuationdisorderedsystems,
Bachmann_2011,Nachtergaele_2019}, together with the improved long-range
Lieb--Robinson bound of Teufel and Wessel
\cite{teufel2025liebrobinsonboundsautomorphicequivalence}.  Unlike the
usual ``local perturbations perturb locally'' (LPPL) setting
\cite{De_Roeck_2015,teufel2025liebrobinsonboundsautomorphicequivalence},
a truncation tail is not localized near a single perturbation region:
discarded terms occur throughout the lattice.  We therefore control the
response to each discarded term separately and sum the contributions
through the discarded interaction mass near each site.  This yields a
general range-truncation bound for many-body interactions and extends
naturally to isolated low-energy sectors and thermodynamic-limit state
constructions.  Closely related static bounds were obtained by Wang and
Hazzard~\cite{Wang_Hazzard_2023}, who applied a power-law LPPL estimate
term by term to finite-size boundary errors; our result instead treats a
global interaction-range cutoff and organizes the error by the local
tail mass. Related approximation results have also been established for real-time
dynamics of power-law systems.  Tran et al. developed spatial
decompositions of long-range evolution and corresponding
digital-simulation bounds~\cite{Tran_2019,Tran_2020_hierarchy}.  These
results concern finite-time dynamics and spatial-volume cutoffs, whereas
we compare the ground states or isolated low-energy sectors of the full
and globally range-truncated Hamiltonians along a gapped path. 

The paper is organized as follows. Before presenting the technical framework, we summarize the main conclusions in Section~\ref{sec:summary_results}. Next, Section~\ref{sec:setup} introduces the setup and states the main theorems.
Section~\ref{sec:weighted_response} introduces spectral flow and bounds the
response to lattice-wide perturbations.  Section~\ref{sec:finite_volume}
treats finite-volume direct and shell truncations.
Section~\ref{sec:infinite_volume} develops the infinite-volume direct and
shell constructions.  Section~\ref{sec:applications} gives a practical
gap-stability criterion, two-body corollaries, the fermionic extension,
and exactly solvable Gaussian benchmarks with numerical comparisons.
Section~\ref{sec:discussions}
discusses the gap stability problem and raises some open questions.  The appendices
collect more details of the proof and examples.

\section{Summary of main results}
\label{sec:summary_results}

We summarize the main conclusions before introducing the mathematical
framework. An interaction $\Phi$ associates each finite subset $Z$ in the lattice an operator $\Phi(Z)$.  The range-$R$ interaction $\Phi_{\le R}$ discards every term whose
diameter of support exceeds $R$ (see Fig.~\ref{fig:interaction_truncation} for a schematic description).  Its \emph{local tail mass} is
\begin{equation*}
  \tailmass(R)
  :=
  \sup_{x\in\Gamma}
  \sum_{\substack{Z\ni x\\ \diam Z>R}}
  |Z|\snorm{\Phi(Z)}.
\end{equation*}
Here $|Z|$ is its size.  Thus $\tailmass(R)$
is the discarded interaction strength per site, rather than the extensive norm of the whole discarded Hamiltonian. 
The ground state of $\Phi$ is denoted by $\omega$, and its expectation value of an operator $A$ is denoted by $\omega(A) = \ev{A}_\omega$.
Section~\ref{sec:setup} gives the precise definitions
and the summable long-range locality condition used by the first two results.

\medskip
\noindent\textbf{Theorem~\ref{thm:main_finite_response} (informal, finite-volume).}
Consider a specified family for which the adiabatic path from $\Phi_{\le R}$ to
$\Phi$ (called direct path, see Eq.\ \eqref{eq:direct_path}), or the adiabatic path from $\Phi_{\le R}$ to $\Phi_{\le R'}$ (called shell path, see Eq.\ \eqref{eq:shell_path}), has a common
separating gap.  When the Hamiltonian has a unique gapped ground state throughout the path, then every local observable $A_X$ satisfies
\begin{equation*}
  \abs{\ev{A_X}_{{\Lambda}}-\ev{A_X}_{\Lambda,R}}
  \le C|X|\snorm{A_X}\tailmass(R).
\end{equation*}
The shell comparison obeys the same bound, and the constant is independent of
the volume and cutoff scales.  For an isolated low-energy sector, spectral
flow identifies the endpoint sectors and the same estimate controls the matrix
of $A_X$ within them, independently of sector rank and internal band width.

\begin{table*}[t]
\caption{Logical guide to the truncation routes.  The listed hypotheses are required by the corresponding results;
decay-dependent rates appear in Table~\ref{tab:summary}.}
\label{tab:result_routes}
\footnotesize
\begin{ruledtabular}
\begin{tabular}{lll}
Goal & Main hypotheses & Output \\
\hline
Full versus finite cutoff
& Direct path; common sector gap
& Tail-mass bound \\
Cutoffs $R<R'$
& Shell path; common sector gap
& One-shell tail-mass bound \\
Thermodynamic-limit shell flow
& Compatible gapped shells; unique ground states; summable tails
& Selected branch $\omega_{R_k}\to\omega_\infty$ \\
Infinite-volume direct-path
& Superpolynomial decay; regular locally unique GNS-gapped path
& Superpolynomial bound \\
\end{tabular}
\end{ruledtabular}
\end{table*}

\Needspace{7\baselineskip}
\medskip
\noindent\textbf{Theorem~\ref{thm:main_infinite_shell} (informal, thermodynamic-limit).}
Suppose compatible finite-volume shell paths track the unique ground state and
have a common gap, uniformly in volume and shell scale.  If
$R_k\uparrow\infty$ and $\sum_k\tailmass(R_k)<\infty$, choose $\omega_{R_0}$ as
a weak-$*$ subsequential limit of the finite-volume ground states at the base
cutoff.  Starting from this choice, the shell flows produce infinite-volume
ground states $\omega_{R_k}$ converging locally to a ground state
$\omega_\infty$ of the full interaction, with
\begin{equation*}
  \abs{\ev{A_X}_{\infty}-\ev{A_X}_{R_j} }
  \le C|X|\snorm{A_X}\sum_{k\ge j}\tailmass(R_k).
\end{equation*}
The selected branch may depend on the chosen base limit.  If the full ground
state is unique, the selected limit is that state.

\medskip
\noindent\textbf{Theorem~\ref{thm:main_infinite_direct} (informal, infinite-volume).}
For superpolynomially decaying interactions, suppose that, for every
sufficiently large $R$, a prescribed cutoff ground state $\omega_R$ is
connected to the same full ground state $\omega$ by a differentiable path of
locally unique GNS-gapped states.  If the gap is uniform and the path has the
regularity specified below, then, for every $m>0$,
\begin{equation*}
  |\ev{A_X}_\omega - \ev{A_X}_R|
  \le C_m|X|\snorm{A_X}(1+R)^{-m}.
\end{equation*}
Here $C_m$ is independent of $R$ but may not be uniformly bounded in $m$.  If the interaction tail is exponentially small, then the right hand side of this inequality is replaced by an exponential function $e^{-\kappa R}$ with $\kappa$ independent of $R$.

We emphasize that, although
Theorems~\ref{thm:main_infinite_shell} and
\ref{thm:main_infinite_direct} both compare local-observable
expectations for truncated and full interactions, they apply in
different settings. First, they rely on different adiabatic-path
assumptions. Second, Theorem~\ref{thm:main_infinite_shell} starts from
finite-volume ground states of the truncated interactions and, through
the shell construction, selects an infinite-volume ground-state branch
of the full interaction. By contrast,
Theorem~\ref{thm:main_infinite_direct} directly compares two prescribed
infinite-volume GNS states; in this setting, both the truncated and full
interactions are assumed to have locally unique gapped ground states.

The preceding results apply to general interactions. To make the resulting
rates explicit, we specialize to two-body interactions. If the interaction
between sites \(x\) and \(y\) is bounded by a constant times
\((1+\dist(x,y))^{-p}\) with \(p>d\), lattice-shell counting gives
\(\tailmass(R)=O(R^{-(p-d)})\). If $p>2d$ \footnote{This condition is required by the intermediate weighted-response estimate, Proposition~\ref{prop:weighted_response}.}, under the corresponding adiabatic-path
hypothesis, the error bound has the same rate $O(R^{-(p-d)})$.
The condition $p>2d$ is sufficient for the present method, but we do not claim that it is necessary in all cases. 

All three main results extend to parity-even fermionic interactions and
observables, with the same rates and parameter dependence.
The theorems also extends to ordered products of finitely many local
observables. For such products, the absolute-error bound depends on the total
size of their supports, but not on the separations between them; it does not,
however, guarantee small relative error when the correlator itself is very small.

In the next section, we introduce the basic definitions and state the main theorems rigorously. The logical routes and the corresponding rate estimates are summarized in Table~\ref{tab:result_routes} and~\ref{tab:summary}, respectively.

\begin{table*}[t]
\caption{Summary of the truncation estimates of local observable errors ($\sim\mathcal{T}_\Phi(R)$) under the
corresponding path and gap hypotheses.  The finite-volume column gives
the observable error for an isolated low-energy sector and,
in the rank-one ground-state case, the local state distance $d_X$.  The
symbol $\omega_\infty$ denotes the state selected by the
thermodynamic limit shell construction. The symbol $O(R^{-\infty})$ denotes decay faster than every inverse power of $R$.}
\label{tab:summary}
\begin{ruledtabular}
\begin{tabular}{lccc}
Decay & Finite volume & Infinite-volume direct &  Thermodynamic-limit shell \\
\hline
Power law, $\snorm{\Phi}_{F_\eta}<\infty$, $\eta>d$
& $O(R^{-\eta})$
& not treated
& $O(R^{-\eta})$ to $\omega_\infty$ \\
Superpolynomial
& $O(R^{-\infty})$
& $O(R^{-\infty})$
& $O(R^{-\infty})$ to $\omega_\infty$ \\
Exponential weighted norm, $\snorm{\Phi}_{f_\mu}<\infty$
& $O(e^{-\mu R})$
& $O(e^{-\kappa R})$, $\kappa<\mu$
& $O(e^{-\mu R})$ to $\omega_\infty$ \\
\end{tabular}
\end{ruledtabular}
\end{table*}

\section{Setup and main theorems}
\label{sec:setup}
We first introduce the lattice and interaction setup, including the local
tail mass, in Section~\ref{subsec:setup}; define the truncation paths and
assumptions in Section~\ref{subsec:assumptions}; and state the main finite-
and infinite-volume results in Section~\ref{subsec:main_results}.

\subsection{Lattice systems, interactions, and volume settings}\label{subsec:setup}

For definiteness, our infinite-volume setting is
$\Gamma=\mathbb{Z}^{d}$ with the graph metric \footnote{For the standard lattice, the graph metric is $\dist(x,y)=\|x-y\|_1$. Replacing it by an equivalent lattice metric, such as the Euclidean metric, changes only dimension-dependent constants and the normalization of exponential decay rates; the polynomial exponents and conclusions remain unchanged.}.  The finite-volume
estimates are first stated for finite subsets of this lattice, but their
proofs use only uniform metric-growth and locality estimates and also
apply to other finite metric lattices.  For finite sets
$X,Y\Subset\Gamma$, we write
\begin{align}
  \dist(X,Y)&:=\min_{x\in X,\,y\in Y}\dist(x,y),
  \\
  \diam X&:=\max_{x,y\in X}\dist(x,y).
\end{align}
Each site $x\in\Gamma$ carries a finite-dimensional Hilbert space $\mathcal{H}_{x}$.  For a finite set $X\Subset\Gamma$, let
\begin{equation}
  \mA_X=\bigotimes_{x\in X}\mathcal{B}(\mathcal{H}_x),
  \qquad
  \mA=\overline{\bigcup_{X\Subset\Gamma}\mA_X}^{\snorm{\cdot}}
\end{equation}
be the local and quasi-local observable algebras.
This tensor-product formulation is used for the proofs in the main
sections.  The fermionic version replaces these algebras by the local
CAR algebras and restricts interactions and physical observables to
their even parts; the exact correspondence of the proof ingredients is
given in Section~\ref{sec:fermionic_extension}.

An interaction is a collection of self-adjoint operators
$\Phi(Z)\in\mA_Z$, indexed by finite $Z\Subset\Gamma$.  For a
nonincreasing function $F\colon[0,\infty)\to(0,\infty)$, define the
diameter norm
\begin{equation}
  \snorm{\Phi}_F
  :=
  \sup_{x\in\Gamma}
  \sum_{Z\ni x}|Z|\frac{\snorm{\Phi(Z)}}{F(\diam Z)}.
\end{equation}
If $\snorm{\Phi}_F<\infty$, we say that $\Phi$ is $F$-local, or
diameter $F$-summable \footnote{Lieb--Robinson bounds are commonly formulated using a
pair-distance norm.  The diameter norm used here controls the relevant
pair-distance norm; see Appendix~\ref{app:lr_bound}.}.
For the polynomial profile with exponent $\eta>0$,
\begin{equation}
  F_{\eta}(r)=(1+r)^{-\eta},
\end{equation}
the diameter norm is
\begin{equation}
  \snorm{\Phi}_{F_{\eta}}
  :=
  \sup_{x\in\Gamma}
  \sum_{Z\ni x}
  |Z|(1+\diam Z)^{\eta}\snorm{\Phi(Z)}.
  \label{eq:diameter_norm}
\end{equation}
Throughout this work we consider interactions with $\snorm{\Phi}_{F_{\eta}}<\infty$ for some fixed value $\eta>d$; this is
the principal locality hypothesis in the finite-volume arguments.

Next define the truncations.  Let $R>1$ be a length scale.  The range-$R$
truncation, discarded tail, and shell interaction are
\begin{align}
  \Phi_{\le R}(Z)&=\mathbf{1}_{\{\diam Z\le R\}}\Phi(Z),
  \\
  \Phi_{>R}(Z)&=\mathbf{1}_{\{\diam Z>R\}}\Phi(Z),
  \\
  \Phi_{(R,R']}(Z)&=\mathbf{1}_{\{R<\diam Z\le R'\}}\Phi(Z).
\end{align}
The quantity that directly enters our estimates is the
\emph{local tail mass}:
\begin{equation}
  \tailmass(R)
  :=
  \sup_{x\in\Gamma}
  \sum_{\substack{Z\ni x\\ \diam Z>R}}
  |Z|\snorm{\Phi(Z)}.
  \label{eq:tail_mass}
\end{equation}
If $\snorm{\Phi}_{F_{\eta}}<\infty$, then
\begin{equation}
  \tailmass(R)
  \le
  (1+R)^{-\eta}\snorm{\Phi}_{F_{\eta}}.
  \label{eq:tail_mass_power}
\end{equation}
Indeed, on the summation domain $\diam Z>R$, inserting
$(1+\diam Z)^{-\eta}(1+\diam Z)^\eta=1$ and using
$(1+\diam Z)^{-\eta}\le(1+R)^{-\eta}$ proves this estimate term by term.
We call $\Phi$ \emph{superpolynomially decaying} if
$\snorm{\Phi}_{F_\eta}<\infty$ for \emph{every} $\eta\ge0$.  For such an
interaction Eq.~\eqref{eq:tail_mass_power} holds for every $\eta>0$.
For $f_\mu(r):=e^{-\mu r}$, we say that $\Phi$ is
\emph{exponentially local} and satisfies \emph{exponential diameter summability at rate $\mu$}  if
\begin{equation}
  \snorm{\Phi}_{f_{\mu}}
  :=
  \sup_{x\in\Gamma}
  \sum_{Z\ni x}|Z|e^{\mu\diam Z}\snorm{\Phi(Z)}<\infty.
  \label{eq:exponential_norm}
\end{equation}
With this notation, we have
\begin{equation}
  \tailmass(R)\le e^{-\mu R}\snorm{\Phi}_{f_{\mu}}.
  \label{eq:tail_mass_exp}
\end{equation}
Inserting $e^{-\mu\diam Z}e^{\mu\diam Z}=1$ and using
$e^{-\mu\diam Z}\le e^{-\mu R}$ for $\diam Z>R$ gives the stated bound.
Moreover, for every $\eta\ge0$,
\begin{equation}
  \snorm{\Phi}_{F_\eta}
  \le
  c_{\eta,\mu}\snorm{\Phi}_{f_\mu},
  \qquad
  c_{\eta,\mu}:=
  \sup_{r\ge0}(1+r)^\eta e^{-\mu r}<\infty.
  \label{eq:exponential_controls_polynomial}
\end{equation}
This follows term by term from
$(1+r)^\eta\le c_{\eta,\mu}e^{\mu r}$.
Thus exponential diameter summability supplies every fixed polynomial
background norm required below.

Exponential diameter summability at rate $\mu$ must not be inferred merely
from pointwise exponential decay of individual terms at the same rate; see
Section~\ref{sec:two_body_exponential} for the two-body case.

Unless stated otherwise, finite subsets $\Lambda\Subset\Gamma$ carry
open boundary conditions, and we define
\begin{equation}
  H_{\Lambda,R}=\sum_{Z\Subset\Lambda}\Phi_{\le R}(Z),
  \qquad
  H_{\Lambda}=\sum_{Z\Subset\Lambda}\Phi(Z),
\end{equation}
and, whenever their ground states are unique, denote those states by
$\omega_{\Lambda,R}$ and $\omega_{\Lambda}$.  Unique finite-volume
ground states are assumed only in statements that explicitly use this
notation.

\medskip
\noindent\emph{Other finite-volume geometries and boundary conditions.}
Our finite-volume results extend to other lattice geometries and boundary conditions under the uniformity assumptions stated below. For other lattice geometries, let
$\{(\Gamma_L,\dist_L)\}_{L\in\mathbb N}$ be a family of finite metric
lattices.  Here $L$ labels the system size, $\Gamma_L$ is the finite set
of lattice sites, and $\dist_L$ is the intrinsic distance on
$\Gamma_L$ (the length of a shortest lattice path).  We assume uniform
polynomial volume growth:
\begin{equation}
  \sup_L\sup_{x\in\Gamma_L}
  \abs{B_L(x,r)}
  \le C_{\mathrm{vol}}(1+r)^d,
  \qquad r\ge0,
  \label{eq:uniform_finite_volume_growth}
\end{equation}
where $B_L(x,r):=\{y\in\Gamma_L:\dist_L(x,y)\le r\}$ is the metric ball
of radius $r$ about $x$, $d$ is the common growth dimension, and
$C_{\mathrm{vol}}<\infty$ is independent of $L$, $x$, and $r$.
If the interaction norms, local tail masses defined using $\dist_L$,
and the Lieb--Robinson constants of Appendix~\ref{app:lr_bound} are
uniform in $L$, all finite-volume proofs hold with $\Gamma_L$ and its
tail mass.  For the discrete torus
$\Gamma_L=(\mathbb Z/L\mathbb Z)^d$, the label $L$ is its circumference
in lattice units and $\dist_L$ is the shortest-path graph distance,
including wrap-around paths.  Diameters and truncation ranges are
always measured using
$\dist_L$, so a bond crossing a periodic seam remains short.

A boundary condition may also be represented by an additional
interaction $\Psi^{\mathrm{bc}}_\Lambda$, where the superscript denotes
the chosen boundary condition.  If this uniformly local interaction is
fixed along the interpolation path, it affects only the background
constants and gap assumption, not the path derivative or tail mass.  If
it is varied or truncated, its per-site strength must be added to the response
bound.  The relevant sector must remain uniformly isolated; its rank
may depend on the boundary condition, and the projector formulation
below covers an isolated degenerate sector on a closed geometry.

For the thermodynamic limit statements we continue to use the
open-boundary cubic exhaustion
$\Lambda_n=[-L_n,L_n]^d\cap\mathbb Z^d$, where
$L_n\uparrow\infty$.

\medskip
In the infinite-volume setting, the Hamiltonian as a global sum of
interaction terms does not converge in $\mA$ in the norm topology.
Instead, the interaction defines a derivation $\mL_\Phi$ on local
observables
\begin{equation}
  \mL_{\Phi}(A)
  =
  \sum_{Z\cap\supp A\ne\varnothing}[\Phi(Z),A].
  \label{eq:liouvillian}
\end{equation}
A state $\omega$ is a ground state of $\Phi$ if
\begin{equation}
  \omega\bigl(A^*\mL_{\Phi}(A)\bigr)\ge0
  \label{eq:ground_state_condition}
\end{equation}
for every local $A$.
It is a \emph{locally unique gapped ground state} with gap $\gamma>0$
if
\begin{equation}
  \omega\bigl(A^*\mL_{\Phi}(A)\bigr)
  \ge
  \gamma\left[
    \omega(A^*A)-\abs{\omega(A)}^2
  \right]
  \label{eq:locally_unique_gap}
\end{equation}
for every $A$ in the domain of the closed derivation $\mL_\Phi$.
In the GNS representation
$(\mathcal H_\omega,\pi_\omega,\Omega_\omega)$, let $H_\omega\ge0$ be
the Hamiltonian implementing the dynamics, normalized by
$H_\omega\Omega_\omega=0$.  If
$\pi_\omega(\operatorname{Dom}\mL_\Phi)\Omega_\omega$ is a form core
for $H_\omega$, Eq.~\eqref{eq:locally_unique_gap} is equivalent to
\begin{equation}
  \ker H_\omega=\mathbb C\Omega_\omega,
  \qquad
  \sigma(H_\omega)\setminus\{0\}\subset[\gamma,\infty).
  \label{eq:gns_unique_gap}
\end{equation}
Local uniqueness is a
property of this chosen GNS branch; it does not exclude other algebraic
ground states in disjoint representations.

To pass from the finite-volume shell construction to the infinite-volume system, we regard the resulting states as states on the common quasi-local algebra $\mA$, so that weak-$*$ limits can be considered. Fix a reference product state on $\mA$. For each $n$, extend a state on $\mA_{\Lambda_n}$ to $\mA$ by tensoring it with the restriction of the reference state to $\Gamma\setminus\Lambda_n$. This extension is only an auxiliary choice: for every fixed finite $X\subset\Gamma$, one has $X\subset\Lambda_n$ for all sufficiently large $n$, and hence the expectation of every observable in $\mA_X$ is eventually independent of the state chosen outside $\Lambda_n$. Consequently, the weak-$*$ limiting behavior on local observables is determined entirely by the finite-volume shell states.

\subsection{Truncation paths and spectral assumptions}\label{subsec:assumptions}

Next we define the interpolation paths used to compare the truncated
and full interactions.  Let $R'>R>1$ be two truncation scales.  The
direct path restores the entire discarded tail in one step, whereas
the shell path adds only the interactions between the two scales.  More
concretely, they are defined by
\begin{align}
  {\text{direct path: }} \Phi_s^{(R,\infty)}
  &=\Phi_{\le R}+s\Phi_{>R},
  \label{eq:direct_path}
  \\
  {\text{shell path:  }} \Phi_s^{(R,R')}
  &=\Phi_{\le R}+s\Phi_{(R,R')},
  \label{eq:shell_path}
\end{align}
with $s\in[0,1]$.  Formally, the direct path is the case $R'=\infty$.

The direct path is used below both for finite-volume systems and,
for superpolynomial interactions under additional state-path
assumptions, intrinsically in infinite volume.  By contrast, the shell
path hypotheses (discussed below) in this work are imposed only in finite volume, not in the intrinsic infinite-volume setting.
Yet, with a shell-by-shell iteration, we can construct a sequence of states converging to an infinite-volume limit through finite-range spectral flows.
This construction requires uniform control along
every selected shell and compatibility at common truncation scales, hence the hypothesis is different from the direct-path hypothesis.
An intrinsic infinite-volume alternative and the obstacles for
genuinely polynomial long-range paths are discussed in
Appendix~\ref{app:intrinsic_infinite_volume}.

\medskip
\noindent\emph{Finite-volume spectral sector convention.}
Along either path, let $H_{\Lambda}(s)$ be the induced finite-volume
Hamiltonian and let $E_{\Lambda}(s)=\min\sigma(H_{\Lambda}(s))$ ($\sigma$ denoting the spectrum).
An \emph{isolated low-energy sector} means a family of compact
intervals $I_{\Lambda}(s)$ with continuously differentiable endpoints
such that
\begin{equation}
  \begin{aligned}
    \Sigma_{\Lambda}^{\mathrm{low}}(s)
    &:=
    \sigma(H_{\Lambda}(s))\cap I_{\Lambda}(s),
    \\
    E_{\Lambda}(s)
    &\in
    \Sigma_{\Lambda}^{\mathrm{low}}(s),
    \\
    \sup_{\lambda\in\Sigma_{\Lambda}^{\mathrm{low}}(s)}
    \bigl(\lambda-E_{\Lambda}(s)\bigr)
    &\le\delta_{\Lambda}.
  \end{aligned}
  \label{eq:low_energy_band}
\end{equation}
\begin{equation}
  \dist\left(
    I_{\Lambda}(s),
    \sigma(H_{\Lambda}(s))
      \setminus\Sigma_{\Lambda}^{\mathrm{low}}(s)
  \right)
  \ge\gamma
  \label{eq:isolated_sector_gap}
\end{equation}
with the required uniformity in $s$ and across the specified path
family stated in the hypotheses below.  Here $\delta_{\Lambda}\ge0$ is the width allowed for the
low-energy band $\Sigma_\Lambda^{\rm low}(s)$, and
\begin{equation}
  P_{\Lambda}(s)
  :=
  \mathbf 1_{I_{\Lambda}(s)}(H_{\Lambda}(s))
  \label{eq:sector_projection}
\end{equation}
is its spectral projection.  The continuously tracked interval and the
uniform separation imply that $\operatorname{rank}P_{\Lambda}(s)$ is
constant along the path, although it need not be bounded uniformly in
$\Lambda$.  The case $\delta_{\Lambda}=0$ is an exactly degenerate
ground space.  If $\delta_{\Lambda}>0$, we use ``approximately
degenerate'' only for the isolated low-energy sector and do not call every
state supported in it an exact finite-volume ground state.  No condition
$\delta_{\Lambda}\to0$ is needed for the finite-volume projector bound.
For a thermodynamic-limit statement about arbitrary states supported in
the sector, additional hypotheses are needed to ensure that their
weak-$*$ accumulation points are infinite-volume ground states.  Along
the exhaustion used here, $\delta_{\Lambda_n}\to0$ is a simple
sufficient energy-width condition for that conclusion, together with
the usual locality assumptions, but it does not by itself imply
convergence of the entire sequence of states.

At the endpoints of a direct path we write the projections as
$P_{\Lambda,R}$ and $P_{\Lambda}$, and at the endpoints of a shell path
as $P_{\Lambda,R}$ and $P_{\Lambda,R'}$. 
Once the projection at one endpoint is fixed, the uniform separating gap determines its continuation uniquely along the path. For successive shells, we choose one projection \(P_{\Lambda,R_k}\) at each common cutoff and require both adjacent paths to use this projection. This consistency is automatic for the full ground-space projection.

We now state the spectral hypotheses for the two paths.  The
finite-volume clauses refer to the sector convention above.

\begin{assumption}[Finite-volume direct-path hypothesis]
\label{assumption:adiabaticity-direct-path}
Let $\mathcal D$ be a specified family of pairs
$(\Lambda,R)$ with $\Lambda\Subset\Gamma$ and $R>1$.
For every $(\Lambda,R)\in\mathcal D$, the Hamiltonian induced by
Eq.~\eqref{eq:direct_path} carries an isolated low-energy sector
satisfying Eqs.~\eqref{eq:low_energy_band}--\eqref{eq:sector_projection}.
There exists $\gamma>0$, independent of
$(\Lambda,R)\in\mathcal D$ and $s\in[0,1]$, such that the separating
gap is at least $\gamma$.

For each finite volume $\Lambda$, fix an isolated low-energy projection
$P_\Lambda$ of the full Hamiltonian $H_\Lambda(\Phi)$. For every
$(\Lambda,R)\in\mathcal D$, require the sector along the direct path to
end at $P_\Lambda$ when $s=1$.
\end{assumption}

\begin{assumption}[Finite-volume shell-path hypothesis]
\label{assumption:adiabaticity-shell-path}
Let $\mathcal S$ be a specified family of triples
$(\Lambda,R,R')$ with $1<R<R'<\infty$.  For every
$(\Lambda,R,R')\in\mathcal S$, the Hamiltonian induced by
Eq.~\eqref{eq:shell_path} carries an isolated low-energy sector
satisfying Eqs.~\eqref{eq:low_energy_band}--\eqref{eq:sector_projection}.
There exists $\gamma>0$, independent of
$(\Lambda,R,R')\in\mathcal S$ and $s\in[0,1]$, such that the separating
gap is at least $\gamma$.  Whenever $(\Lambda,R,R')$ and
$(\Lambda,R',R'')$ belong to $\mathcal S$, their prescribed sector
projections at the common cutoff $R'$ agree.
\end{assumption}

The intrinsic infinite-volume result requires additional regularity beyond the GNS gap.

\medskip
\noindent\emph{State-derivative regularity.}
Let $\mA_\infty$ denote the superpolynomially localized observable
space of Ref.~\cite{becker2025automorphicequivalencegappedphases}, and
fix the origin $0\in\Gamma=\mathbb Z^d$.  A state path
$(\omega_s)_{s\in[0,1]}$ satisfies the state-derivative regularity
condition if there exist $\nu\in\mathbb N_0$ and $C<\infty$ such that,
for every $A\in\mA_\infty$, the map
$s\mapsto\omega_s(A)$ is differentiable and
\begin{equation}
  \abs{\partial_s\omega_s(A)}
  \le
  C\snorm{A}^{\mathrm{loc}}_{\nu,0},
  \qquad s\in[0,1].
  \label{eq:state_path_derivative_bound}
\end{equation}
Here $\snorm{\cdot}^{\mathrm{loc}}_{\nu,0}$ is the
localized-observable seminorm defined in
Eq.~\eqref{eq:localized_observable_seminorm}.

\begin{assumption}[Direct-path hypotheses, infinite-volume setting]
\label{assumption:direct-path-infinite}
Suppose that $\Phi$ is superpolynomially decaying.  There exist
$R_*<\infty$, $\gamma>0$, and a fixed infinite-volume ground state
$\omega$ of $\Phi$ such that, for every $R>R_*$, there are an
infinite-volume ground state $\omega_R$ of $\Phi_{\le R}$ and a path
of states $(\omega_s^{(R)})_{s\in[0,1]}$ satisfying the
state-derivative regularity condition above and
\[
  \omega_0^{(R)}=\omega_R,
  \qquad
  \omega_1^{(R)}=\omega.
\]
For every $R>R_*$ and $s\in[0,1]$, the state $\omega_s^{(R)}$ is a
locally unique gapped ground state of $\Phi_s^{(R,\infty)}$, in the
sense of Eq.~\eqref{eq:locally_unique_gap}, with gap at least
$\gamma$.  The regularity parameters $\nu$ and $C$ may depend on $R$;
no uniformity of these parameters in $R$ is assumed.
\end{assumption}

The state-derivative regularity condition is an independent
hypothesis: it does not follow from differentiability of the
interaction path and the locally unique gap alone.  By contrast,
because $\Phi$ is superpolynomially decaying and
\[
  \Phi_s^{(R,\infty)}
  =
  \Phi_{\le R}+s\Phi_{>R},
  \qquad
  \partial_s\Phi_s^{(R,\infty)}=\Phi_{>R},
\]
the interaction path
$(\Phi_s^{(R,\infty)})_{s\in[0,1]}$ automatically belongs to
$\mathcal P_{\infty,[0,1]}^{(1)}$ for every $R>R_*$.  Consequently,
\begin{equation}
  \bigl(
    (\omega_s^{(R)})_{s\in[0,1]},
    (\Phi_s^{(R,\infty)})_{s\in[0,1]}
  \bigr)
  \label{eq:differentiable_path}
\end{equation}
is a differentiable path of gapped systems in the sense of
Ref.~\cite{becker2025automorphicequivalencegappedphases}.

Note that the three assumptions are logically independent. Each estimate below uses only the assumption associated with its truncation route.

Under the finite-volume clauses, spectral flow is implemented by
unitaries that identify the isolated sectors at the two endpoints.
Under the intrinsic clause of
Assumption~\ref{assumption:direct-path-infinite}, it is an
automorphism cocycle that transports the prescribed infinite-volume
states.

For states $\omega$ and $\nu$, define the local distinguishability on $X$ by
\begin{equation}
  d_X(\omega,\nu)
  =
  \sup_{\substack{A=A^*\in\mA_X\\ \snorm{A}\le1}}
  \abs{\omega(A)-\nu(A)}.
  \label{eq:local_distinguishability}
\end{equation}
For spin systems, this equals the trace norm of the difference of the reduced density matrices on $X$.  The same equality holds for parity-even fermionic states when the supremum is taken over the even local algebra; see Section~\ref{sec:fermionic_extension} for details.

For two low-energy-sector projections $P_0$ and $P_1$ identified by a
unitary $U$ satisfying $P_1=UP_0U^*$, define the
\emph{compressed-observable error} by
\[
  \locerr(A_X;P_0,P_1,U)
  :=
  \snorm{U^*P_1A_XP_1U-P_0A_XP_0}.
\]
It compares the compression of the same local observable to the two
sectors after their unitary identification and therefore also controls
off-diagonal matrix elements within the sectors.  When the projections
and unitary are fixed by a truncation path, we suppress them and display
the relevant truncation scales in the subscript.

\subsection{Main results}\label{subsec:main_results}
We now summarize the main results.  Table~\ref{tab:summary} gives the
scaling of local distinguishability in each decay regime and volume setting.

The first result is finite-volume and applies to either interpolation.

\begin{theorem}[Finite-volume response]
\label{thm:main_finite_response}
Let $\eta>d$ and $\snorm{\Phi}_{F_{\eta}}<\infty$. 
The following conclusions hold independently: 
\begin{enumerate}
    \item For the direct path setting,  suppose $\Phi$ satisfies Assumption~\ref{assumption:adiabaticity-direct-path} with \(\mathcal D\) the specified family of pairs $(\Lambda, R)$ and $\gamma$ the uniform gap.
Let
$U_{\Lambda}^{(R,\infty)}$ denote
spectral-flow unitaries at $s=1$ defined in
Section~\ref{sec:finite_spectral_flow}. Then 
\begin{align}
  \locerr_{\Lambda,R}(A_X)
  &=
  \Bigl\lVert
    \bigl(U_{\Lambda}^{(R,\infty)}\bigr)^*
    P_{\Lambda}A_XP_{\Lambda}
    U_{\Lambda}^{(R,\infty)}
  \nonumber\\
  &\qquad
    -
    P_{\Lambda,R}A_XP_{\Lambda,R}
  \Bigr\rVert
  \nonumber\\
  &\le
  C|X|\snorm{A_X}\tailmass(R).
  \label{eq:main_direct_sector}
\end{align}

\item For the shell path setting, suppose $\Phi$ satisfies Assumption~\ref{assumption:adiabaticity-shell-path} with $\mathcal{S}$ the family of shells in the assumption, and $U_{\Lambda}^{(R,R')}$ denote the corresponding spectral-flow unitaries at $s=1$. Then for every $(\Lambda,R,R')\in\mathcal S$,
\begin{align}
  \locerr_{\Lambda,R,R'}(A_X)
  &=
  \Bigl\lVert
    \bigl(U_{\Lambda}^{(R,R')}\bigr)^*
    P_{\Lambda,R'}A_XP_{\Lambda,R'}
    U_{\Lambda}^{(R,R')}
  \nonumber\\
  &\qquad
    -
    P_{\Lambda,R}A_XP_{\Lambda,R}
  \Bigr\rVert
  \nonumber\\
  &\le
  C|X|\snorm{A_X}\tailmass(R).
  \label{eq:main_shell_sector}
\end{align}
\end{enumerate}

Here $C$ depends on $d$, $\eta$, $\gamma$, and an upper bound for
$\snorm{\Phi}_{F_\eta}$, but not on $\Lambda$, $R$, $R'$, the sector
rank, or the internal band width $\delta_{\Lambda}$. If the low-energy sector is the rank-one projection onto the unique ground state along the
entire path, Eqs.~\eqref{eq:main_direct_sector} and
\eqref{eq:main_shell_sector} reduce to
\begin{align}
  d_X(\omega_{\Lambda},\omega_{\Lambda,R})
  &\le C|X|\tailmass(R)
  \label{eq:main_direct_finite}
\end{align}
and
\begin{align}
  d_X(\omega_{\Lambda,R'},\omega_{\Lambda,R})
  &\le C|X|\tailmass(R)
  \label{eq:main_shell_step}
\end{align}
for the direct and shell paths, respectively, with $\omega_{\Lambda},\omega_{\Lambda,R},\omega_{\Lambda,R'}$ the corresponding ground states.
\end{theorem}

\noindent\emph{Proof location.}
The unique ground-state bounds are proved in
Section~\ref{sec:finite_response_proof}, using
Proposition~\ref{prop:weighted_response}.  The compressed-sector
identity, its proof, and the matched-state interpretation are given in
Appendix~\ref{app:finite_volume_sectors}. The proof of this main theorem is based on techniques developed in Ref.\ \onlinecite{teufel2025liebrobinsonboundsautomorphicequivalence}.
Appendix~\ref{app:finite_range_shell} gives a weaker alternative estimate that does not use Ref.\ \onlinecite{teufel2025liebrobinsonboundsautomorphicequivalence}.

In the unique-ground-state case, the shell path estimate can be iterated
along increasing truncation scales.  In finite volume this gives an
optional comparison under successive shell gap assumptions.  Its main
use below is to construct a selected infinite-volume branch when the
tail masses along the chosen scale sequence are summable.

\begin{theorem}[Thermodynamic-limit shell flow]
\label{thm:main_infinite_shell}
Let $\eta>d$, let $\snorm{\Phi}_{F_\eta}<\infty$, and let
$R_k\uparrow\infty$ be an increasing sequence with $R_0>1$ such that
\begin{equation}
  \sum_{k\ge0}\tailmass(R_k)<\infty.
  \label{eq:main_shell_tail_summability}
\end{equation}
For example, any geometrically increasing sequence
$R_k=\lambda^kR_0$ with $\lambda>1$ has this property.  Let
$\Lambda_n=[-L_n,L_n]^d\cap\mathbb Z^d$, with
$L_n\uparrow\infty$, and set
\begin{equation}
  \mathcal S
  =
  \{(\Lambda_n,R_k,R_{k+1}):n\ge1,\ k\ge0\}.
  \label{eq:infinite_shell_family}
\end{equation}
Assume that $\mathcal S$ and $\Phi$ satisfy
Assumption~\ref{assumption:adiabaticity-shell-path}, with every prescribed
sector being the rank-one projection onto the unique ground state and
with common gap $\gamma>0$.  Let $\omega_{R_0}$ be a weak-$*$ limit,
along a subsequence of the exhaustion, of the ground states of
$\Phi_{\le R_0}$.

Then every shell spectral flow has a thermodynamic-limit automorphism
$\alpha^{(R_k,R_{k+1})}_s$.  Defining recursively
\begin{equation}
  \omega_{R_{k+1}}
  =
  \omega_{R_k}\circ\alpha^{(R_k,R_{k+1})}_1,
  \label{eq:main_coherent_branch}
\end{equation}
produces a shell-transported sequence of infinite-volume ground states satisfying
\begin{equation}
  \abs{\omega_{R_{k+1}}(A_X)-\omega_{R_k}(A_X)}
  \le
  C|X|\snorm{A_X}\tailmass(R_k).
  \label{eq:main_shell_cauchy}
\end{equation}
Consequently, the sequence converges locally to a state
$\omega_{\infty}$, and $\omega_{\infty}$ is a ground state of the full
interaction $\Phi$.  Moreover,
\begin{equation}
  \abs{\omega_\infty(A_X)-\omega_{R_j}(A_X)}
  \le
  C|X|\snorm{A_X}
  \sum_{k\ge j}\tailmass(R_k).
  \label{eq:main_shell_limit_bound}
\end{equation}
If the full ground state is unique, then $\omega_{\infty}=\omega$.
\end{theorem}

\noindent\emph{Proof location.}
The thermodynamic-limit shell flow automorphisms are constructed in
Section~\ref{sec:thermodynamic_shell_flow}; convergence of the transported
states and the ground-state property are proved in
Section~\ref{sec:coherent_shell_limit}.

If $\snorm{\Phi}_{f_\mu}<\infty$, the same theorem retains the full
exponential diameter summability rate (in Eq.~\ref{eq:exponential_norm}) along a fixed shell-transported branch with
$R_j=\lambda^jR_0$, $\lambda>1$:
\begin{equation}
  d_X(\omega_\infty,\omega_{R_j})
  \le
  C_\mu |X|\snorm{\Phi}_{f_\mu}e^{-\mu R_j}.
  \label{eq:infinite_shell_exponential}
\end{equation}
Indeed, Eq.~\eqref{eq:main_shell_limit_bound} and
$\tailmass(R_k)\le e^{-\mu R_k}\snorm{\Phi}_{f_\mu}$ reduce the claim to
\begin{equation}
  \sum_{k\ge j}e^{-\mu R_k}
  =
  \sum_{\ell\ge0}e^{-\mu\lambda^\ell R_j}
  \le
  \frac{e^{-\mu R_j}}{1-e^{-\mu(\lambda-1)}},
  \qquad R_j>1.
  \label{eq:infinite_shell_exponential_sum}
\end{equation}
Here $\lambda^\ell R_j\ge R_j+\ell(\lambda-1)$ was used in the last
inequality.
Without uniqueness of the full ground state, Eq.~\eqref{eq:infinite_shell_exponential}
compares only states on this shell-transported branch \footnote{We remark that restarting
the construction at a different base scale need not select the same
$\omega_\infty$.}.

For a power-law diameter norm, Theorems~\ref{thm:main_finite_response} and \ref{thm:main_infinite_shell} give the same $R^{-\eta}$ rate.  For two-body interactions with termwise strength $r^{-p}$, the more physical form of the rate is $R^{-(p-d)}$, which is discussed in Section~\ref{sec:two_body}.

The infinite-volume direct result uses a different input.

\begin{theorem}[Infinite-volume direct superpolynomial truncation]
\label{thm:main_infinite_direct}
Suppose the intrinsic infinite-volume 
Assumption~\ref{assumption:direct-path-infinite} holds.  Thus the
direct paths are paired with prescribed differentiable state branches
having a common locally unique GNS gap, and their full endpoint is the
same state $\omega$ for every $R>R_*$.  Then, for every $m>0$,
\begin{equation}
  d_X(\omega,\omega_R)
  \le C_m|X|(1+R)^{-m},
  \label{eq:main_infinite_superpoly}
\end{equation}
where $C_m$ is independent of $R$.  
\end{theorem}
We remark that no uniformity of $C_m$ in $m$ is asserted. 
The conclusion applies only to the prescribed endpoint states $\omega_R$ and $\omega$ connected by the assumed differentiable path; it makes no statement about other ground states of $\Phi_{\le R}$ or $\Phi$ that are not on this path.

\noindent\emph{Proof location.}
The theorem is proved in Section~\ref{sec:direct_superpoly}; the
uniform, quantitative inverse-Liouvillian estimate used there is proved in
Appendix~\ref{app:inverse_liouvillian}.

The local distinguishability bounds also control the expectation of product operators $\omega(A_1\cdots A_n)$. The estimate depends on the size of the union support,
but not on the distances between the individual supports. This gives a bound for multipoint correlators.

\begin{corollary}[Distance-independent multipoint correlator error bounds]
\label{cor:multipoint_correlations}
Suppose that two states $\omega$ and $\nu$ satisfy
\begin{equation}
  d_Z(\omega,\nu)\le K|Z|\varepsilon
  \label{eq:multipoint_input}
\end{equation}
for every finite $Z\Subset\Gamma$, where $K$ and $\varepsilon$ are
independent of $Z$.  Let $A_i\in\mA_{X_i}$ for $i=1,\ldots,n$, and set
$X=\bigcup_{i=1}^nX_i$.  Then
\begin{align}
  &\abs{
    \omega(A_1\cdots A_n)-\nu(A_1\cdots A_n)
  }
  \nonumber\\
  &\qquad\le
  K\varepsilon |X|\prod_{i=1}^n\snorm{A_i}
  \le
  K\varepsilon
  \left(\sum_{i=1}^n|X_i|\right)
  \prod_{i=1}^n\snorm{A_i}.
  \label{eq:multipoint_ordinary}
\end{align}
In particular, for $A\in\mA_X$ and $B\in\mA_Y$,
\begin{align}
  &\abs{
    \omega(AB)-\nu(AB)
  }
  \nonumber\\
  &\qquad\le
  K\varepsilon |X\cup Y|\snorm{A}\snorm{B}
  \le
  K(|X|+|Y|)\snorm{A}\snorm{B} \varepsilon.
  \label{eq:two_point_ordinary}
\end{align}
Neither bound depends on the diameter of the union support or on the mutual
distances between the insertion supports.
\end{corollary}

The corollary applies to every setting of truncation comparison above. One
takes $\varepsilon=\tailmass(R)$ for the finite-volume direct and
one-shell estimates, $\varepsilon=\sum_{k\ge j}\tailmass(R_k)$ for the
thermodynamic-limit shell flow, and
$\varepsilon=(1+R)^{-m}$ for the infinite-volume direct
superpolynomial estimate, with the corresponding constant $K$.  For an
isolated finite-volume sector, substituting $A_X=A_1\cdots A_n$ in
Eqs.~\eqref{eq:main_direct_sector} and \eqref{eq:main_shell_sector}
gives the corresponding bound for the spectrally identified endpoint
compressions.  The same bound holds for spectral-flow-matched state
moments by Eq.~\eqref{eq:matched_state_estimate}.

We remark that Corollary~\ref{cor:multipoint_correlations} controls absolute error only. It does not imply small relative error when the correlator itself is very small, nor does it show that a finite-range truncation reproduces the asymptotic large-distance behavior of the full correlator. Long-distance structure may therefore be missed even when all fixed-support moments are close.

The decay classes and volume settings are summarized in Table~\ref{tab:summary}.  The infinite-volume direct-path power-law case is not claimed here \footnote{Because the presently available intrinsic infinite-volume automorphic-equivalence framework assumes superpolynomial decay, we do not obtain an intrinsic direct-path result for power-law interactions. Existing polynomial finite-volume spectral-flow bounds involve substantial decay losses, and their extension to the required infinite-volume GNS setting has not been formally developed.}
The finite-volume response estimate alone neither constructs an infinite-volume long-range spectral flow nor selects a thermodynamic-limit ground-state branch.

\section{Response to lattice-wide perturbations}
\label{sec:weighted_response}

This section first introduces the finite-volume spectral flow and then
derives the response estimate used below.  The technical input is the
locality bound for individual interaction terms proved in
Appendix~\ref{app:lr_bound}.

\subsection{Spectral flow}
\label{sec:finite_spectral_flow}

Let $s\mapsto\Phi_s$, $s\in[0,1]$, be a differentiable interaction path
on a finite volume $\Lambda$.  Define the Hamiltonian $H_s$ and its
Heisenberg dynamics on $B\in\mA_\Lambda$ by
\begin{equation}
  H_s:=\sum_{Z\Subset\Lambda}\Phi_s(Z),
  \qquad
  \tau_t^s(B):=e^{\ii tH_s}Be^{-\ii tH_s}.
  \label{eq:finite_path_physical_dynamics}
\end{equation}
Suppose that $H_s$ carries an isolated low-energy sector with spectral
projection $P_s$ and separating gap at least $\gamma>0$, uniformly in
$s$.  Fix an exact filter $W_\gamma$ as in
Refs.~\cite{Bachmann_2011,teufel2025liebrobinsonboundsautomorphicequivalence}.
It is real-valued and odd, and its Fourier transform satisfies
\begin{equation}
  \widehat W_\gamma(\omega)
  =
  -\frac{\ii}{\sqrt{2\pi}\,\omega}
  \quad\text{for }|\omega|\ge\gamma,
  \label{eq:finite_filter_fourier}
\end{equation}
and its polynomial moments are finite:
\begin{equation}
  m_N(W_\gamma)
  :=
  \int_{\mathbb R}|W_\gamma(t)|(1+|t|)^N\,\dd t
  <\infty,
  \qquad N\in\mathbb N_0.
  \label{eq:finite_filter_moments}
\end{equation}
The associated inverse-Liouvillian map $\mathcal{I}_s(\cdot)$ is defined by \cite{teufel2025liebrobinsonboundsautomorphicequivalence,becker2025automorphicequivalencegappedphases}
\begin{equation}
  \mathcal I_s(B)
  :=
  \int_{\mathbb R}W_\gamma(t)\tau_t^s(B)\,\dd t.
  \label{eq:finite_inverse_liouvillian}
\end{equation}
The integral converges absolutely in operator norm.
It inverts $B\mapsto-\ii[H_s,B]$ on matrix elements between the range of
$P_s$ and its orthogonal complement.  Define the spectral-flow generator
\begin{equation}
  G_s
  :=
  -\mathcal I_s(\dot H_s)
  =
  \sum_{Z\Subset\Lambda}K_{s,Z},
  \qquad
  K_{s,Z}:=-\mathcal I_s(\dot\Phi_s(Z)).
  \label{eq:finite_spectral_flow_generator}
\end{equation}
The inverse property then gives
\begin{equation}
  \dot P_s=-\ii[G_s,P_s].
  \label{eq:finite_spectral_projection_equation}
\end{equation}
The spectral-flow unitary is defined by
\begin{equation}
  \ii\partial_sU_s=G_sU_s,
  \qquad
  U_0=\mathbf 1.
  \label{eq:finite_spectral_flow_unitary}
\end{equation}
The corresponding flow of observables is
\begin{equation}
  \alpha_s(A):=U_s^*AU_s.
  \label{eq:finite_spectral_flow_automorphism}
\end{equation}
With these conventions, the isolated low-energy sector is transported exactly:
\begin{equation}
  P_s=U_sP_0U_s^*.
  \label{eq:finite_spectral_flow_transport}
\end{equation}
This evolution, often referred to as ``spectral flow'', is also called quasi-adiabatic continuation or automorphic equivalence \cite{Hastings_2005,Bachmann_2011}.  Although
$G_s$ is not strictly local, Lieb--Robinson bounds and the decay of
$W_\gamma$ make its action on local observables quasi-local.  The
quantitative form needed below is
Proposition~\ref{prop:appendix_termwise}, proved in
Appendix~\ref{app:lr_bound}.

\subsection{Weighted response bound}

For an interaction $\Theta$, $\beta>d$, and a local region $X$, define
the weighted perturbation strength relative to $X$ by
\begin{equation}
  \mathcal W_{\beta,X}(\Theta)
  :=
  \sum_{Z\Subset\Lambda}
  |Z|\snorm{\Theta(Z)}F_\beta\bigl(\dist(Z,X)\bigr).
  \label{eq:weighted_size}
\end{equation}

\begin{proposition}[Weighted response bound]
\label{prop:weighted_response}
Fix $\eta>d$ and $0<\varepsilon<\eta-d$, and set
$\beta=\eta-\varepsilon>d$.  Assume
\begin{equation}
  \sup_{s\in[0,1]}\snorm{\Phi_s}_{F_\eta}\le M
  \label{eq:background_bound}
\end{equation}
and suppose that $H_s$ carries an isolated low-energy sector with
separating gap at least $\gamma>0$, uniformly in $s$ and $\Lambda$.
Then, for every $A_X\in\mA_X$, the finite-volume spectral flow satisfies
\begin{equation}
  \snorm{\alpha_1(A_X)-A_X}
  \le
  C\snorm{A_X}\int_0^1
  \mathcal W_{\beta,X}(\dot\Phi_s)\,\dd s,
  \label{eq:weighted_response}
\end{equation}
where $C=C(d,\eta,\beta,\gamma,M)$ depends on $d$, $\eta$, $\varepsilon$, $\gamma$, and $M$, but not on $\Lambda$ or the spatial support of $\dot\Phi_s$.
\end{proposition}

\begin{proof}
Let $\mathbb E_{\Lambda\setminus X}$ be the conditional expectation
from $\mA_\Lambda$ onto $\mA_{\Lambda\setminus X}$.  Applied to the
decomposition in Eq.~\eqref{eq:finite_spectral_flow_generator},
Proposition~\ref{prop:appendix_termwise} gives
\begin{equation}
  \snorm{K_{s,Z}-\mathbb{E}_{\Lambda\setminus X}(K_{s,Z})}
  \le
  C|Z|\snorm{\dot\Phi_s(Z)}
  F_\beta\bigl(\dist(Z,X)\bigr).
  \label{eq:main_termwise_localization}
\end{equation}
Define the auxiliary generator
\begin{equation}
  G_{s,X^c}
  :=
  \sum_{Z\Subset\Lambda}
  \mathbb E_{\Lambda\setminus X}(K_{s,Z}).
  \label{eq:main_auxiliary_generator}
\end{equation}
It belongs to $\mA_{\Lambda\setminus X}$.  Summing
Eq.~\eqref{eq:main_termwise_localization} over $Z$ gives
\begin{equation}
  \snorm{G_s-G_{s,X^c}}
  \le
  C\mathcal W_{\beta,X}(\dot\Phi_s).
  \label{eq:main_generator_difference}
\end{equation}
Let $V_s$ be the unitary generated by $G_{s,X^c}$:
\begin{equation}
  \ii\partial_sV_s=G_{s,X^c}V_s,
  \qquad
  V_0=\mathbf 1.
\end{equation}
Since $V_s\in\mA_{\Lambda\setminus X}$, it commutes with $A_X$ and
$V_s^*A_XV_s=A_X$.  Duhamel's formula and
Eq.~\eqref{eq:main_generator_difference} imply
\begin{align}
  \snorm{U_1-V_1}
  &\le
  \int_0^1\snorm{G_s-G_{s,X^c}}\,\dd s
  \nonumber\\
  &\le
  C\int_0^1\mathcal W_{\beta,X}(\dot\Phi_s)\,\dd s.
  \label{eq:main_duhamel_bound}
\end{align}
Consequently,
\begin{align}
  \snorm{\alpha_1(A_X)-A_X}
  &=
  \snorm{U_1^*A_XU_1-V_1^*A_XV_1}
  \nonumber\\
  &\le
  2\snorm{A_X}\snorm{U_1-V_1},
\end{align}
which proves Eq.~\eqref{eq:weighted_response} after absorbing the factor
$2$ into constant $C$.
\end{proof}

This is the key distinction from Theorem~10 of
Ref.~\cite{teufel2025liebrobinsonboundsautomorphicequivalence}.  Its LPPL
conclusion is organized by distance from a common perturbation region.
Although that theorem also gives an interaction-norm variant for extensive
perturbations, when truncation terms occur throughout the lattice the common
region is lattice-wide and its distance factor yields no useful smallness.
Proposition~\ref{prop:weighted_response} instead retains the termwise
weighted sum, which Corollary~\ref{cor:local_mass_bound} converts into a
per-site tail bound.

For the truncation paths, the detailed spatial distribution in
$\mathcal W_{\beta,X}$ can be compressed into a simpler per-site
quantity.  The following corollary performs this reduction and is the
form used throughout the finite-volume argument.

For an interaction $\Theta$, define its per-site strength by
\begin{equation}
  \mT_{\Theta}
  :=
  \sup_{x\in\Gamma}
  \sum_{Z\ni x}|Z|\snorm{\Theta(Z)}.
  \label{eq:local_mass_theta}
\end{equation}

\begin{corollary}[Per-site perturbation bound]
\label{cor:local_mass_bound}
Under the background and gap hypotheses of
Proposition~\ref{prop:weighted_response}, there exists a finite constant
$C'=C'(d,\eta,\gamma,M)$ such that
\begin{equation}
  \snorm{\alpha_1(A_X)-A_X}
  \le
  C'|X|\snorm{A_X}
  \int_0^1\mT_{\dot\Phi_s}\,\dd s.
  \label{eq:local_mass_response}
\end{equation}
If further every nonzero term of $\dot\Phi_s$ has diameter larger than
$R$ and
\begin{equation}
  \sup_{s,x}
  \sum_{\substack{Z\ni x\\ \diam Z>R}}
  |Z|\snorm{\dot\Phi_s(Z)}
  \le t(R),
  \label{eq:per_site_tail_assumption}
\end{equation}
then the right-hand side of Eq.~\eqref{eq:local_mass_response} is bounded by $C'|X|\snorm{A_X}t(R)$.
\end{corollary}

\begin{proof}
The claim is trivial for $X=\varnothing$, so assume that $X$ is
nonempty.  Fix $s\in[0,1]$ and regard the finite-volume interaction
$\dot\Phi_s$ as zero outside $\Lambda$. 
Set
\[
  \varepsilon_*:=\frac{\eta-d}{2},
  \qquad
  \beta_*:=\eta-\varepsilon_*=\frac{\eta+d}{2},
\]
and apply Proposition~\ref{prop:weighted_response} with these values. For every nonempty $Z$ that
contributes to $\mathcal W_{\beta_*,X}(\dot\Phi_s)$, choose an anchor
$z_Z\in Z$ such that
\begin{equation}
  \dist(z_Z,X)=\dist(Z,X).
\end{equation}
Grouping the interaction terms by their anchors gives
\begin{align}
  \mathcal W_{\beta_*,X}(\dot\Phi_s)
  &=
  \sum_{z\in\Gamma}
  F_{\beta_*}\bigl(\dist(z,X)\bigr)
  \sum_{Z:\,z_Z=z}|Z|\snorm{\dot\Phi_s(Z)}
  \nonumber\\
  &\le
  \mT_{\dot\Phi_s}
  \sum_{z\in\Gamma}F_{\beta_*}\bigl(\dist(z,X)\bigr)
  \nonumber\\
  &\le
  |X|\mT_{\dot\Phi_s}
  \sup_{x\in\Gamma}\sum_{z\in\Gamma}
  F_{\beta_*}\bigl(\dist(z,x)\bigr).
  \label{eq:anchor_sum}
\end{align}
The first inequality follows because every set assigned to the anchor
$z$ contains $z$, so the corresponding inner sum is bounded by
$\mT_{\dot\Phi_s}$.  For the second inequality, the monotonicity of
$F_\beta$ gives
\begin{align}
  F_{\beta_*}\bigl(\dist(z,X)\bigr)
  &=
  \max_{x\in X}F_{\beta_*}\bigl(\dist(z,x)\bigr)
  \nonumber\\
  &\le
  \sum_{x\in X}F_{\beta_*}\bigl(\dist(z,x)\bigr).
\end{align}
Moreover,
\begin{equation}
  c_{\beta_*,d}
  :=
  \sup_{x\in\Gamma}\sum_{z\in\Gamma}
  F_{\beta_*}\bigl(\dist(z,x)\bigr)
  <\infty
  \label{eq:spatial_kernel_sum}
\end{equation}
because $\Gamma=\mathbb Z^d$ and $\beta_*>d$.  Substituting
Eq.~\eqref{eq:anchor_sum} into Proposition~\ref{prop:weighted_response}
therefore yields
\begin{equation}
  \snorm{\alpha_1(A_X)-A_X}
  \le
  C c_{\beta_*,d}|X|\snorm{A_X}
  \int_0^1\mT_{\dot\Phi_s}\,\dd s,
\end{equation}
which proves Eq.~\eqref{eq:local_mass_response} with
$C':=C c_{\beta_*,d}$.

For the final assertion, every term with $\diam Z\le R$ vanishes by
assumption.  Hence, for each $s\in[0,1]$ and $x\in\Gamma$,
\begin{equation}
  \sum_{Z\ni x}|Z|\snorm{\dot\Phi_s(Z)}
  =
  \sum_{\substack{Z\ni x\\ \diam Z>R}}
  |Z|\snorm{\dot\Phi_s(Z)}
  \le t(R)
\end{equation}
by Eq.~\eqref{eq:per_site_tail_assumption}.  Taking the supremum over
$x$ gives $\mT_{\dot\Phi_s}\le t(R)$.  Since the path interval has
length one,
\begin{equation}
  \int_0^1\mT_{\dot\Phi_s}\,\dd s\le t(R),
\end{equation}
which proves the stated bound.
\end{proof}

\section{Finite-volume truncation}
\label{sec:finite_volume}

We now apply Section~\ref{sec:weighted_response} to the finite-volume setting, for both
direct path and shell paths, and then evaluate the resulting bounds for interactions with
power-law, superpolynomial and exponential tails.

\subsection{Direct and shell truncation estimates}
\label{sec:finite_response_proof}
As mentioned, the same response estimate controls both paths, even though their gap assumptions are logically distinct.  
In both cases the background interaction fixes the response constant, whereas the per-site mass of the path derivative fixes the error. This section shows the unique-ground-state calculation, and we put the discussion of isolated low-energy sectors in Appendix~\ref{app:finite_volume_sectors}.

The background interaction norms are uniform along both interpolation
paths:
\begin{equation}
  \sup_{s,R,R'}
  \max\left\{
    \snorm{\Phi_s^{(R,\infty)}}_{F_\eta},
    \snorm{\Phi_s^{(R,R')}}_{F_\eta}
  \right\}
  \le
  \snorm{\Phi}_{F_\eta}.
  \label{eq:finite_uniform_background}
\end{equation}
Thus Proposition~\ref{prop:weighted_response} and
Corollary~\ref{cor:local_mass_bound} apply with constants independent of
$\Lambda$, $R$, and $R'$.

\begin{proof}[Proof of the rank-one part of
Theorem~\ref{thm:main_finite_response}]
Let $\alpha_s$ denote the spectral flow for the relevant path.  Since it
transports the unique ground-state projection, the ground states at
$s=0$ and $s=1$ obey
\begin{equation}
  \omega_1(A)=\omega_0\bigl(\alpha_1(A)\bigr).
  \label{eq:endpoint_state_transport}
\end{equation}
Consequently,
\begin{equation}
  \abs{\omega_1(A_X)-\omega_0(A_X)}
  \le
  \snorm{\alpha_1(A_X)-A_X}.
  \label{eq:state_from_automorphism}
\end{equation}

\noindent\emph{Direct interpolation.}
For this path,
\begin{equation}
  \dot\Phi_s^{(R,\infty)}=\Phi_{>R},
  \qquad
  \mT_{\Phi_{>R}}=\tailmass(R).
\end{equation}
Equation~\eqref{eq:local_mass_response} therefore gives
\begin{equation}
  \abs{\omega_\Lambda(A_X)-\omega_{\Lambda,R}(A_X)}
  \le
  C|X|\snorm{A_X}\tailmass(R).
\end{equation}
Taking the supremum over self-adjoint $A_X$ with $\snorm{A_X}\le1$
proves Eq.~\eqref{eq:main_direct_finite}.

\noindent\emph{Two-cutoff interpolation.}
For a single path from $R$ to $R'$,
\begin{equation}
  \dot\Phi_s^{(R,R')}=\Phi_{(R,R']},
  \qquad
  \mT_{\Phi_{(R,R']}}
  \le
  \tailmass(R).
\end{equation}
This step eliminates dependence on the upper shell radius $R'$.  Hence
\begin{equation}
  \abs{\omega_{\Lambda,R'}(A_X)-\omega_{\Lambda,R}(A_X)}
  \le
  C|X|\snorm{A_X}\tailmass(R),
\end{equation}
and hence Eq.~\eqref{eq:main_shell_step}.
\end{proof}

For the case of low-energy sectors, the proof uses the same spectral flow and the same
operator-norm estimate, but retains the low-energy-sector projections
instead of taking a ground-state expectation.  The resulting compressed-operator
identity, together with the matched-state consequence, is proved in
Appendix~\ref{app:finite_volume_sectors}.

The above direct path estimate immediately implies
\begin{equation}
  d_X(\omega_\Lambda,\omega_{\Lambda,R})
  \le
  C|X|\snorm{\Phi}_{F_\eta}(1+R)^{-\eta}.
  \label{eq:finite_direct_power}
\end{equation}
This is sharper, for the present observable-level objective, than first
controlling a complete interaction norm of the spectral-flow generator
through Proposition~8 of
Ref.~\cite{teufel2025liebrobinsonboundsautomorphicequivalence}.  That
proposition gives the stronger structural conclusion that the generator
itself is polynomially localized, but its use here would impose stronger
input decay and an avoidable loss in the truncation exponent.
\medskip
\phantomsection
\label{subsec:successive-shell-finite-volume}
\noindent\emph{Successive-shell concatenation.}
The preceding result is a one-step comparison between two finite cutoffs;
it is not yet a shell construction.  If a direct-path gap is available,
Eq.~\eqref{eq:finite_direct_power} already gives the full finite-volume
truncation estimate.  If instead gaps are known only along a compatible
sequence $R_k\to R_{k+1}$, Eq.~\eqref{eq:main_shell_step} can be applied
successively and summed by the triangle inequality.  This changes the
adiabatic-connectivity assumption, not the decay rate.  A geometric
sequence makes the accumulated tail masses summable without changing
their decay class.

\begin{corollary}[Finite-volume successive-shell comparison]
\label{cor:finite_shell}
Fix $R>1$ and $\lambda>1$, and let $R_k=\lambda^kR$.  For every finite volume
$\Lambda$ under consideration, let $N_\Lambda$ be the first index for
which $R_{N_\Lambda}\ge\diam\Lambda$.  Assume that the shell family
\begin{equation}
  \begin{aligned}
    \mathcal S_{R,\lambda}
    =\bigl\{(\Lambda,R_k,R_{k+1}):\,
    &\Lambda\text{ is under consideration},\\
    &0\le k<N_\Lambda\bigr\}.
  \end{aligned}
  \label{eq:finite_shell_family}
\end{equation}
satisfies Assumption~\ref{assumption:adiabaticity-shell-path}, with every
prescribed sector being the rank-one projection onto the unique ground
state.  Then, for every such $\Lambda$,
\begin{equation}
  d_X(\omega_\Lambda,\omega_{\Lambda,R})
  \le
  C|X|\sum_{k=0}^{N_\Lambda-1}\tailmass(R_k).
  \label{eq:finite_shell_general}
\end{equation}
The sum is understood as zero when $N_\Lambda=0$.
If $\snorm{\Phi}_{F_\eta}<\infty$ with $\eta>d$, then
\begin{equation}
  d_X(\omega_\Lambda,\omega_{\Lambda,R})
  \le
  \frac{2^\eta C}{1-\lambda^{-\eta}}
  |X|\snorm{\Phi}_{F_\eta}(1+R)^{-\eta}.
  \label{eq:finite_shell_power}
\end{equation}
\end{corollary}

\begin{proof}
Apply Eq.~\eqref{eq:main_shell_step} to every adjacent pair
$(R_k,R_{k+1})$ and use the triangle inequality.  Once
$R_k\ge\diam\Lambda$, the truncated and full Hamiltonians coincide, so
the sum terminates.  Since $R>1$,
$1+\lambda^kR\ge\frac12\lambda^k(1+R)$.  Equation
\eqref{eq:tail_mass_power} and the resulting geometric series give
Eq.~\eqref{eq:finite_shell_power}.
\end{proof}

Section~\ref{sec:coherent_shell_limit} uses the same concatenation to
construct a convergent infinite-volume branch.

\subsection{Rapidly decaying interactions}\label{subsec:rapid-decay-interactions}

The preceding results are expressed through the tail mass
$\tailmass(R)$.  We now evaluate it for superpolynomial and
exponentially weighted interactions.  For the optional shell
comparison, geometric summation preserves the same decay class.

\begin{corollary}[Finite-volume superpolynomial and exponential tails]
\label{cor:finite_rapid}
Fix a background exponent $\eta_0>d$, $R>1$, and $\lambda>1$.
Let $R_k=\lambda^kR$ and, for every finite volume $\Lambda$ under
consideration, let $N_\Lambda$ be the first index for which
$R_{N_\Lambda}\ge\diam\Lambda$.  Assume either that the direct family
$\mathcal D_R:=\{(\Lambda,R):\Lambda\text{ is under consideration}\}$
satisfies the finite-volume Assumption~\ref{assumption:adiabaticity-direct-path}, or, in the shell
case, that the family of shells $\mathcal S_{R,\lambda}$ in
Eq.~\eqref{eq:finite_shell_family} satisfies
Assumption~\ref{assumption:adiabaticity-shell-path}.

Let $\mathcal U_{\Lambda,R}:=U_\Lambda^{(R,\infty)}$ in the direct
case.  In the shell case, write
$U_k:=U_\Lambda^{(R_k,R_{k+1})}$, set $V_0:=\mathbf1$ and
$V_{k+1}:=U_kV_k$, and define
$\mathcal U_{\Lambda,R}:=V_{N_\Lambda}$.  Compatibility then
identifies $P_{\Lambda,R_{N_\Lambda}}=P_\Lambda$.  Relative to
$\mathcal U_{\Lambda,R}$, denote the compressed-observable error
between $P_{\Lambda,R}$ and $P_\Lambda$ by
$\locerr_{\Lambda,R}(A_X)$.
If $\Phi$ is superpolynomially decaying, then, for every $m>0$,
\begin{equation}
  \locerr_{\Lambda,R}(A_X)
  \le
  C_m|X|\snorm{A_X}(1+R)^{-m}.
  \label{eq:finite_superpoly}
\end{equation}
If $\snorm{\Phi}_{f_\mu}<\infty$, then
\begin{equation}
  \locerr_{\Lambda,R}(A_X)
  \le
  C_\mu|X|\snorm{A_X}e^{-\mu R}\snorm{\Phi}_{f_\mu}.
  \label{eq:finite_exponential}
\end{equation}
Here $C_m$ and $C_\mu$ are independent of $\Lambda$, $R$, the sector
rank, and the internal band width. If the sectors are rank-one ground-state
projections, the same two bounds hold with the left-hand side replaced
by $d_X(\omega_\Lambda,\omega_{\Lambda,R})$ and with
$\snorm{A_X}$ omitted.
\end{corollary}

\begin{proof}
Set
\begin{equation}
  M_0:=\snorm{\Phi}_{F_{\eta_0}},
  \qquad
  C_{\mathrm{resp}}
  :=C_{\mathrm{resp}}(d,\eta_0,\gamma,M_0),
  \label{eq:finite_rapid_response_constant}
\end{equation}
where $C_{\mathrm{resp}}$ is the response constant in
Theorem~\ref{thm:main_finite_response}, applied with the fixed exponent
$\eta_0$.  Once $\eta_0$ and $M_0$ are fixed, it is independent of
$m$, $\mu$, $R$, and $\Lambda$.
Equation~\eqref{eq:main_direct_sector} gives
\begin{equation}
  \locerr_{\Lambda,R}(A_X)
  \le
  C_{\mathrm{resp}}|X|\snorm{A_X}\tailmass(R)
  \label{eq:finite_rapid_direct_input}
\end{equation}
in the direct case.

For the shell construction, spectral transport and compatibility give
$P_{\Lambda,R_{k+1}}=U_kP_{\Lambda,R_k}U_k^*$, and telescoping yields
\begin{align}
  &\mathcal U_{\Lambda,R}^*P_\Lambda A_XP_\Lambda
    \mathcal U_{\Lambda,R}
    -P_{\Lambda,R}A_XP_{\Lambda,R}
  \nonumber\\
  &\quad=
  \sum_{k=0}^{N_\Lambda-1}V_k^*
  \Bigl[
    U_k^*P_{\Lambda,R_{k+1}}A_XP_{\Lambda,R_{k+1}}U_k
  \nonumber\\
  &\qquad\qquad
    -P_{\Lambda,R_k}A_XP_{\Lambda,R_k}
  \Bigr]V_k.
  \label{eq:finite_rapid_sector_telescoping}
\end{align}
Applying Eq.~\eqref{eq:main_shell_sector} term by term therefore gives
\begin{equation}
  \locerr_{\Lambda,R}(A_X)
  \le
  C_{\mathrm{resp}}|X|\snorm{A_X}
  \sum_{k=0}^{N_\Lambda-1}\tailmass(R_k).
  \label{eq:finite_rapid_shell_input}
\end{equation}

For a superpolynomial interaction, set
$M_m:=\snorm{\Phi}_{F_m}$.  Equation~\eqref{eq:tail_mass_power} gives
$\tailmass(R_k)\le M_m(1+R_k)^{-m}$.  Moreover,
\begin{equation}
  \sum_{k\ge0}(1+\lambda^kR)^{-m}
  \le
  \frac{2^m}{1-\lambda^{-m}}(1+R)^{-m},
  \qquad R>1,
  \label{eq:finite_superpoly_geometric_sum}
\end{equation}
where we used
$1+\lambda^kR\ge\frac12\lambda^k(1+R)$.  Since the geometric factor is larger than
one, Eqs.~\eqref{eq:finite_rapid_direct_input} and
\eqref{eq:finite_rapid_shell_input} prove Eq.~\eqref{eq:finite_superpoly}
in both cases with
\begin{equation}
  C_m
  =
  C_{\mathrm{resp}}M_m
  \frac{2^m}{1-\lambda^{-m}}.
  \label{eq:finite_superpoly_constant}
\end{equation}

For the exponential class, Eq.~\eqref{eq:exponential_controls_polynomial}
gives
\begin{equation}
  M_0
  \le
  c_{\eta_0,\mu}\snorm{\Phi}_{f_\mu},
  \qquad
  c_{\eta_0,\mu}
  :=\sup_{r\ge0}(1+r)^{\eta_0}e^{-\mu r}<\infty.
  \label{eq:finite_exponential_background_constant}
\end{equation}
Thus the required background norm is finite, while
Eq.~\eqref{eq:tail_mass_exp} gives
$\tailmass(R_k)\le e^{-\mu R_k}\snorm{\Phi}_{f_\mu}$.  Along the shell
path,
\begin{equation}
  \sum_{k\ge0}e^{-\mu\lambda^kR}
  \le
  \frac{e^{-\mu R}}{1-e^{-\mu(\lambda-1)}},
  \qquad R>1,
  \label{eq:finite_exponential_geometric_sum}
\end{equation}
because $\lambda^kR\ge R+k(\lambda-1)$.  Thus geometric summation
preserves the exponent $\mu$, and the direct and shell cases are both
covered by
\begin{equation}
  C_\mu
  =
  \frac{C_{\mathrm{resp}}}{1-e^{-\mu(\lambda-1)}}.
  \label{eq:finite_exponential_constant}
\end{equation}
Thus, in the shell case, $C_m$ depends on $d$, $\eta_0$, $\gamma$,
$M_0$, $m$, $M_m$, and $\lambda$, whereas $C_\mu$ depends on $d$,
$\eta_0$, $\gamma$, $M_0$, $\mu$, and $\lambda$.
For the direct path alone, the geometric factors in
Eqs.~\eqref{eq:finite_superpoly_constant} and
\eqref{eq:finite_exponential_constant} may be omitted.  Finally, when
the projections are rank one, spectral flow matches the unique ground
states; taking the supremum over self-adjoint $A_X$ with
$\snorm{A_X}\le1$ gives the stated distinguishability bounds.
\end{proof}
\section{Infinite-volume truncation}
\label{sec:infinite_volume}

We now pass to two complementary infinite-volume constructions.
Section~\ref{sec:direct_superpoly} discusses a prescribed differentiable
ground-state branch by intrinsic infinite-volume spectral flow. Because an intrinsic infinite-volume automorphic-equivalence theorem for polynomially decaying interactions in the GNS framework used here
is not presently available from known literature, we do not obtain an intrinsic direct-path
result for power-law interactions. 
Hence, we take a new route of connecting finite-volume and infinite-volume setting in
Sections~\ref{sec:thermodynamic_shell_flow} and
\ref{sec:coherent_shell_limit}. We first take the thermodynamic limit of each finite-range shell flow and then assemble a compatible sequence of cutoff states. Summability of their successive changes yields the
full-interaction limit.  Appendix~\ref{app:intrinsic_infinite_volume}
explains why we do not replace it by an intrinsic infinite-volume shell
hypothesis.

\subsection{Direct superpolynomial interpolation}
\label{sec:direct_superpoly}

We work here with superpolynomially decaying interactions in infinite
volume.  The proof combines automorphic equivalence along the direct path
with a uniform inverse-Liouvillian estimate that converts the discarded
tail into a bound on the local generator.  This proves
Theorem~\ref{thm:main_infinite_direct}; we then specialize the estimate to
exponential tails.

By Assumption~\ref{assumption:direct-path-infinite}, the states $\omega_s^{(R)}$ and the direct interaction path form the
differentiable path in Eq.~\eqref{eq:differentiable_path}.  The
interaction seminorms required below are automatically uniform:
\begin{equation}
  \sup_{R,s}\snorm{\Phi_s^{(R,\infty)}}_{F_\eta}
  \le\snorm{\Phi}_{F_\eta}
  \label{eq:uniform_background_superpoly}
\end{equation}
for every fixed $\eta$.

The automorphic-equivalence theorem of
Ref.~\cite{becker2025automorphicequivalencegappedphases} gives a cocycle satisfying
\begin{equation}
  \omega_s^{(R)}=\omega_0^{(R)}\circ\alpha_{0,s}^{(R)},
\end{equation}
whose generator is the inverse-Liouvillian interaction
\begin{equation}
  \Psi_s^{(R)}=-\mathcal I_s^{(R)}(\Phi_{>R}).
\end{equation}
For every local $A$, the generator equation holds in norm:
\begin{equation}
  \frac{\dd}{\dd s}\alpha_{0,s}^{(R)}(A)
  =
  \alpha_{0,s}^{(R)}
  \bigl(\ii\mL_{\Psi_s^{(R)}}(A)\bigr).
  \label{eq:infinite_cocycle_generator}
\end{equation}
Appendix~\ref{app:inverse_liouvillian} extracts from Lemma~B.6 of that
reference the following uniform estimate: for some finite $\eta_0$,
\begin{equation}
  \sup_{R,s}\snorm{\Psi_s^{(R)}}_{F_0}
  \le
  C\snorm{\Phi_{>R}}_{F_{\eta_0}},
  \label{eq:uniform_inverse_liouvillian}
\end{equation}
where $C$ depends on the common gap and finitely many uniform background
seminorms, but not on $R$.  The appendix proves that one may take
\begin{equation}
  \eta_0=5d+5,
  \label{eq:eta_zero_choice}
\end{equation}
with $C$ depending only on $d$, $\gamma$, and
$\snorm{\Phi}_{F_{17d+12}}$.

\begin{proof}[Proof of Theorem~\ref{thm:main_infinite_direct}]
For a local observable $A_X$, state transport and
Eq.~\eqref{eq:infinite_cocycle_generator} give
\begin{equation}
  \frac{\dd}{\dd s}\omega_s^{(R)}(A_X)
  =
  \omega_s^{(R)}
  \bigl(\ii\mL_{\Psi_s^{(R)}}(A_X)\bigr).
  \label{eq:infinite_state_derivative}
\end{equation}
The definition of the Liouvillian and the fact that every $Z$ meeting
$X$ contains at least one $x\in X$ imply
\begin{align}
  \snorm{\mL_{\Psi_s^{(R)}}(A_X)}
  &\le
  2\snorm{A_X}
  \sum_{Z:\,Z\cap X\ne\varnothing}
  \snorm{\Psi_s^{(R)}(Z)}
  \nonumber\\
  &\le
  2|X|\snorm{A_X}
  \snorm{\Psi_s^{(R)}}_{F_0}.
  \label{eq:infinite_local_derivation_bound}
\end{align}
Equations~\eqref{eq:uniform_inverse_liouvillian},
\eqref{eq:infinite_state_derivative}, and
\eqref{eq:infinite_local_derivation_bound} therefore give
\begin{align}
  \abs{\frac{\dd}{\dd s}\omega_s^{(R)}(A_X)}
  &\le
  C|X|\snorm{A_X}
  \snorm{\Phi_{>R}}_{F_{\eta_0}}.
  \label{eq:infinite_direct_derivative}
\end{align}
For every $m>0$, superpolynomial decay implies
\begin{equation}
  \snorm{\Phi_{>R}}_{F_{\eta_0}}
  \le
  (1+R)^{-m}\snorm{\Phi}_{F_{\eta_0+m}}.
\end{equation}
Integrating Eq.~\eqref{eq:infinite_direct_derivative} over $s\in[0,1]$
and using $\omega_0^{(R)}=\omega_R$ and
$\omega_1^{(R)}=\omega$, then taking the supremum over self-adjoint $A_X$ with
$\snorm{A_X}\le1$ proves Eq.~\eqref{eq:main_infinite_superpoly}.  For
each fixed $m$, the resulting constant is independent of $R$; it may
depend on $m$ through $\snorm{\Phi}_{F_{\eta_0+m}}$.
\end{proof}

Under the path assumptions of Theorem~\ref{thm:main_infinite_direct},
exponential diameter summability gives a more precise version of the
same argument.  Set $q_0=\eta_0=5d+5$.  If $r>R$, then
\begin{align}
  (1+r)^{q_0}e^{-\mu r}
  &\le
  (1+R)^{q_0}e^{-\mu R}
  (1+r-R)^{q_0}e^{-\mu(r-R)}
  \nonumber\\
  &\le
  c_{q_0,\mu}(1+R)^{q_0}e^{-\mu R}.
  \label{eq:exponential_tail_weight_pointwise}
\end{align}
Consequently,
\begin{align}
  \snorm{\Phi_{>R}}_{F_{q_0}}
  &\le
  c_{q_0,\mu}(1+R)^{q_0}e^{-\mu R}
  \snorm{\Phi}_{f_\mu},
  \label{eq:infinite_direct_exp_tail_norm}\\
  d_X(\omega,\omega_R)
  &\le
  C_{\mu,d}|X|(1+R)^{5d+5}e^{-\mu R}
  \snorm{\Phi}_{f_\mu}.
  \label{eq:infinite_direct_exp_prefactor}
\end{align}
Since
$\sup_{R\ge0}(1+R)^{5d+5}e^{-(\mu-\kappa)R}<\infty$
for every $0<\kappa<\mu$, Eq.~\eqref{eq:infinite_direct_exp_prefactor}
implies
\begin{equation}
  d_X(\omega,\omega_R)
  \le
  C_\kappa|X|e^{-\kappa R}\snorm{\Phi}_{f_\mu},
  \qquad 0<\kappa<\mu.
  \label{eq:infinite_direct_exponential}
\end{equation}
Here the strict subrate \(\kappa<\mu\) results from absorbing the polynomial prefactor in Eq.~\eqref{eq:infinite_direct_exp_prefactor} into the exponential.

This result is a quantitative application of
Ref.~\cite{becker2025automorphicequivalencegappedphases}: that work establishes the
infinite-volume automorphic transport, while the explicit dependence on
$R$ follows from applying its inverse Liouvillian to the tail $\Phi_{>R}$.

\subsection{Thermodynamic limit of a shell flow}
\label{sec:thermodynamic_shell_flow}

Passing the finite-volume shell-path result to infinite volume requires control of
both its endpoint states and its spectral-flow maps.  The first lemma
identifies weak-$*$ limits of finite-range finite-volume ground states;
the second constructs the thermodynamic-limit shell flow automorphism and shows that it
transports those limits.  Section~\ref{sec:coherent_shell_limit} then
iterates these fixed-shell statements.

\begin{lemma}[Thermodynamic limits of finite-range ground states]
\label{lem:finite_range_ground_state_limit}
Let $\Xi$ be a finite-range interaction and let
$\Lambda_n\nearrow\Gamma$ be an exhaustion.  If $\nu_n$ is a ground
state of
\begin{equation}
  H_{\Lambda_n}^{\Xi}
  =
  \sum_{Z\Subset\Lambda_n}\Xi(Z),
\end{equation}
then every weak-$*$ limit point of $(\nu_n)_n$ is an
infinite-volume ground state of $\Xi$ in the sense of
Eq.~\eqref{eq:ground_state_condition}.
\end{lemma}

\begin{proof}
Let $r<\infty$ be an interaction range for $\Xi$ and fix
$A\in\mA_X$.  The $r$-neighborhood
\begin{equation}
  X^{(r)}
  =
  \{x\in\Gamma:\dist(x,X)\le r\}
\end{equation}
is finite.  Hence there is $n_0$ such that
$X^{(r)}\Subset\Lambda_n$ for all $n\ge n_0$.  Every interaction set
$Z$ that meets $X$ and satisfies $\Xi(Z)\ne0$ is then contained in
$\Lambda_n$, and all terms disjoint from $X$ commute with $A$.  Thus
\begin{equation}
  [H_{\Lambda_n}^{\Xi},A]
  =
  \mL_\Xi(A),
  \qquad n\ge n_0.
  \label{eq:finite_range_exact_liouvillian}
\end{equation}
If $E_n$ is the lowest eigenvalue of $H_{\Lambda_n}^{\Xi}$ and
$\rho_n$ is the density matrix of $\nu_n$, then
$\rho_n(H_{\Lambda_n}^{\Xi}-E_n)=0$ and
\begin{align}
  \nu_n\bigl(A^*[H_{\Lambda_n}^{\Xi},A]\bigr)
  &=
  \operatorname{Tr}\!\left(
    \rho_n A^*(H_{\Lambda_n}^{\Xi}-E_n)A
  \right)
  \nonumber\\
  &\ge0.
  \label{eq:finite_volume_ground_positivity}
\end{align}
Let $\nu_{n_j}\to\nu$ weak-$*$.  Combining
Eqs.~\eqref{eq:finite_range_exact_liouvillian} and
\eqref{eq:finite_volume_ground_positivity} and then taking
$j\to\infty$ gives
$\nu(A^*\mL_\Xi(A))\ge0$.  Since $A$ was arbitrary, $\nu$ is a ground
state of $\Xi$.
\end{proof}

\begin{lemma}[Thermodynamic-limit shell flow automorphism]
\label{lem:limiting_shell_flow}
Fix $1<R<R'$ and let $\Lambda_n\nearrow\Gamma$ be an exhaustion.  Assume
that the family $\{(\Lambda_n,R,R'):n\ge1\}$ satisfies
Assumption~\ref{assumption:adiabaticity-shell-path}, with every prescribed
sector being the rank-one projection onto the unique ground state.  The finite-volume shell flows
$\alpha_{n,s}^{(R,R')}$ converge in norm on every local observable,
uniformly for $s\in[0,1]$, to an automorphism
$\alpha_s^{(R,R')}$ of $\mA$.  Moreover,
\begin{equation}
  \snorm{\alpha_1^{(R,R')}(A_X)-A_X}
  \le
  C|X|\snorm{A_X}\tailmass(R).
  \label{eq:infinite_shell_automorphism_bound}
\end{equation}
If $\omega_R$ is a weak-$*$ limit of the finite-volume ground states at
range $R$ along a subsequence, then
$\omega_R\circ\alpha_1^{(R,R')}$ is the weak-$*$ limit, along the same
subsequence, of the corresponding ground states at range $R'$.
\end{lemma}

\begin{proof}
For fixed $R'$, the path interaction
$\Phi_s^{(R,R')}$ and its derivative have range at most $R'$, and their
finite-range interaction norms are bounded uniformly in $n$ and $s$.
Together with the assumed uniform gap, these are precisely the
hypotheses of the thermodynamic-limit spectral-flow theorem,
Theorem~5.2 of Ref.~\cite{Bachmann_2011}.  It follows that, for every
local $A$,
\begin{align}
  &\lim_{n\to\infty}\sup_{s\in[0,1]}
  \nonumber\\[-2mm]
  &\qquad
  \snorm{
    \alpha_{n,s}^{(R,R')}(A)
    -
    \alpha_s^{(R,R')}(A)}
  =0.
  \label{eq:shell_flow_norm_limit}
\end{align}
The same theorem gives convergence of the inverse finite-volume flows.
Consequently, the two limiting maps extend by norm density to mutual
inverse $*$-automorphisms of $\mA$.
Furthermore,
\begin{equation}
  \mT_{\Phi_{(R,R']}}
  =
  \sup_x
  \sum_{\substack{Z\ni x\\R<\diam Z\le R'}}
  |Z|\snorm{\Phi(Z)}
  \le
  \tailmass(R).
  \label{eq:shell_mass_bound_again}
\end{equation}
The operator-norm estimate \eqref{eq:local_mass_response}, applied in
each $\Lambda_n$ to $\dot\Phi_s=\Phi_{(R,R']}$, gives
\begin{equation}
  \snorm{\alpha_{n,1}^{(R,R')}(A_X)-A_X}
  \le
  C|X|\snorm{A_X}\tailmass(R)
\end{equation}
with a constant independent of $n$.  Taking the norm limit in
Eq.~\eqref{eq:shell_flow_norm_limit} proves
Eq.~\eqref{eq:infinite_shell_automorphism_bound}.

Let $\omega_{n,R}$ and $\omega_{n,R'}$ be the finite-volume ground
states at truncation scales $R$ and $R'$.  Exact spectral transport gives
\begin{equation}
  \omega_{n,R'}(A)
  =
  \omega_{n,R}\bigl(\alpha_{n,1}^{(R,R')}(A)\bigr).
\end{equation}
If $\omega_{n_j,R}\to\omega_R$ weak-$*$, then for local $A$ the norm
convergence of the automorphisms and weak-$*$ convergence give
\begin{align}
  &\abs{
  \omega_{n_j,R}\bigl(\alpha_{n_j,1}^{(R,R')}(A)\bigr)
  -
  \omega_R\bigl(\alpha_1^{(R,R')}(A)\bigr)}
  \nonumber\\
  &\quad\le
  \snorm{
    \alpha_{n_j,1}^{(R,R')}(A)
    -
    \alpha_1^{(R,R')}(A)}
  \nonumber\\
  &\qquad+
  \abs{
    (\omega_{n_j,R}-\omega_R)
    \bigl(\alpha_1^{(R,R')}(A)\bigr)}
  \longrightarrow0.
  \label{eq:shell_state_limit}
\end{align}
Consequently,
\begin{equation}
  \omega_{n_j,R'}(A)
  \longrightarrow
  \omega_R\bigl(\alpha_1^{(R,R')}(A)\bigr).
\end{equation}
Since all states have norm one and the local algebra is norm dense in
$\mA$, convergence on local $A$ extends to weak-$*$ convergence on
all of $\mA$.
Lemma~\ref{lem:finite_range_ground_state_limit}, applied to
$\Xi=\Phi_{\le R'}$, shows that this limit is a ground state of
$\Phi_{\le R'}$.
\end{proof}

\subsection{Shell-transported limit and ground-state identification}
\label{sec:coherent_shell_limit}

We now iterate the compatible finite-volume shell flows along the subsequence chosen at the initial scale.  Summability of the successive errors gives a limiting state, and convergence of the truncated Liouvillians identifies it as a ground state of the full interaction.

\begin{proof}[Proof of Theorem~\ref{thm:main_infinite_shell}]
Fix once and for all a subsequence $(n_j)_j$ for which
\begin{equation}
  \omega_{n_j,R_0}
  \xrightarrow{\mathrm{w}^*}
  \omega_{R_0}.
\end{equation}
Apply Lemma~\ref{lem:limiting_shell_flow} successively to the scales
$R_k$ and define the shell-transported branch by
Eq.~\eqref{eq:main_coherent_branch}.  We claim inductively, without
extracting any further subsequence, that
\begin{equation}
  \omega_{n_j,R_k}
  \xrightarrow{\mathrm{w}^*}
  \omega_{R_k}
  \qquad\text{for every fixed }k.
  \label{eq:fixed_subsequence_induction}
\end{equation}
The assertion holds for $k=0$ by construction.  The state-limit part
of Lemma~\ref{lem:limiting_shell_flow} proves the induction step.
Uniqueness of the finite-volume ground state at each truncation scale
identifies the output of the $(R_k,R_{k+1})$ flow with the ground state
used as the input to the next shell.  By
Lemma~\ref{lem:finite_range_ground_state_limit}, the initially chosen
$\omega_{R_0}$ is a ground state of $\Phi_{\le R_0}$.  Equation
\eqref{eq:fixed_subsequence_induction} and
Lemma~\ref{lem:finite_range_ground_state_limit} then show that every
$\omega_{R_k}$ is a ground state of $\Phi_{\le R_k}$.  In addition,
\begin{align}
  \abs{\omega_{R_{k+1}}(A_X)-\omega_{R_k}(A_X)}
  &\le
  \snorm{\alpha_1^{(R_k,R_{k+1})}(A_X)-A_X}
  \nonumber\\
  &\le
  C|X|\snorm{A_X}\tailmass(R_k),
\end{align}
which is Eq.~\eqref{eq:main_shell_cauchy}.

The tail masses along the chosen scale sequence are summable by
Eq.~\eqref{eq:main_shell_tail_summability}; hence the sequence
$\omega_{R_k}(A)$ is Cauchy for every local $A$.  Define
\begin{equation}
  \omega_\infty(A)=\lim_{k\to\infty}\omega_{R_k}(A)
\end{equation}
on the local algebra.  Positivity, normalization, and
$\abs{\omega_\infty(A)}\le\snorm{A}$ pass to the limit, so
$\omega_\infty$ extends uniquely to a state on $\mA$.  Summing the
adjacent-scale estimate from $k=j$ onward gives
Eq.~\eqref{eq:main_shell_limit_bound}.  Pointwise convergence on local
observables and uniform boundedness also imply weak-$*$ convergence on
all of $\mA$ by density of the local algebra.

It remains to identify the limit.  For $A\in\mA_X$,
\begin{align}
  \snorm{(\mL_\Phi-\mL_{\Phi_{\le R}})(A)}
  &\le
  2\snorm{A}
  \sum_{\substack{Z\cap X\ne\varnothing\\\diam Z>R}}
  \snorm{\Phi(Z)}
  \nonumber\\
  &\le
  2|X|\snorm{A}\tailmass(R).
  \label{eq:liouvillian_convergence}
\end{align}
Summability also implies $\tailmass(R_k)\to0$, so
$\mL_{\Phi_{\le R_k}}(A)\to\mL_\Phi(A)$ in norm.  In particular,
$\mL_\Phi(A)$ is the norm limit of local observables and belongs to
$\mA$.  Since
$\omega_{R_k}$ is a ground state of $\Phi_{\le R_k}$,
\begin{equation}
  \omega_{R_k}\bigl(A^*\mL_{\Phi_{\le R_k}}(A)\bigr)\ge0.
\end{equation}
Furthermore,
\begin{align}
  &\abs{
  \omega_{R_k}\bigl(A^*\mL_{\Phi_{\le R_k}}(A)\bigr)
  -
  \omega_\infty\bigl(A^*\mL_\Phi(A)\bigr)}
  \nonumber\\
  &\quad\le
  \abs{(\omega_{R_k}-\omega_\infty)
  \bigl(A^*\mL_\Phi(A)\bigr)}
  \nonumber\\
  &\qquad+
  \snorm{A}\,
  \snorm{(\mL_{\Phi_{\le R_k}}-\mL_\Phi)(A)}.
\end{align}
The first term vanishes by weak-$*$ convergence and the second by
Eq.~\eqref{eq:liouvillian_convergence}.  Taking $k\to\infty$ proves
\begin{equation}
  \omega_\infty\bigl(A^*\mL_\Phi(A)\bigr)\ge0,
\end{equation}
so $\omega_\infty$ is a ground state of the full interaction.  If that
ground state is unique, then $\omega_\infty=\omega$.
\end{proof}

Consistent shell transport is essential in the nonunique case.  A uniform
finite-volume gap does not force independently chosen finite-volume ground
states to converge along the full exhaustion; boundary conditions or
symmetry breaking may produce different weak-$*$ limit points.

\section{Applications}
\label{sec:applications}

This section converts the abstract tail-mass estimates into concrete
gap criteria, decay rates, and observable consequences.  We begin with
spin systems: first a verifiable sufficient condition for the path gap,
then two-body power-law and exponential examples.  We next extend the
results to parity-even fermionic lattice systems and conclude with exactly
gapped free-fermion benchmarks and covariance matrix numerics.

\subsection{A sufficient gap-stability criterion for spin systems}
\label{sec:gap_stability_criterion}

We first address the path-gap hypothesis for spin systems.  A fermionic
application instead requires an analogous fermionic stability theorem or
a direct gap estimate such as those used later in
Section~\ref{sec:massive_fermion_benchmark}.  For spin systems, the
path-gap assumptions can be verified perturbatively when a
finite-range truncation belongs to a class covered by an extensive
gap-stability theorem.  We formulate one useful version based on the
stability theorem of Michalakis and Zwolak
\cite{Bravyi_Hastings_Michalakis_2010,Michalakis_Zwolak_2013}.  Consider
periodic boxes $\Gamma_L$ and fix an integer $R_*$.  Suppose that the
reference Hamiltonians, written in the uniformly bounded local-projector
form of Ref.~\cite{Michalakis_Zwolak_2013},
\begin{equation}
  H_{L,*}=H_{\Gamma_L,R_*}
  \label{eq:gap_reference_hamiltonian}
\end{equation}
are finite range and frustration free, and 
have a gap at least
$\gamma_*>0$ uniformly in $L$ such that the reference ground state subspaces are  separated from the remainder of their spectra by at least \(\gamma_*\).
Also suppose reference Hamiltonians satisfy the uniform Local-TQO and
Local-Gap conditions of Ref.~\cite{Michalakis_Zwolak_2013}.\footnote{For
$A=b_u(r)$, its enlargement $A(\ell)=b_u(r+\ell)$, and the ground-space
projection $P_{A(\ell)}$ of the Hamiltonian restricted to $A(\ell)$,
Local-TQO requires
$\|P_{A(\ell)}O_AP_{A(\ell)}-c_\ell(O_A)P_{A(\ell)}\|
\le\|O_A\|\Delta_0(\ell)$ with
$c_\ell(O_A)=\Tr(P_{A(\ell)}O_A)/\Tr P_{A(\ell)}$
for every $O_A$ supported in $A$, with
$\Delta_0(\ell)\to0$ as the buffer $\ell$ grows.  Local-Gap requires
the restricted Hamiltonians on balls of radius $r$ to have gaps
$\gamma(r)>0$ that decrease at most polynomially in $r$.  The decay
bounds are uniform in the center and system size.}  Assign
each interaction set $Z$ to one anchor $a(Z)\in Z$ and, using the
integer-valued graph diameter, decompose the omitted tail as
\begin{align}
  V_L=H_{\Gamma_L}-H_{L,*}
  &=\sum_{x\in\Gamma_L}\sum_{r>R_*}V_{x,r},
  \\
  V_{x,r}&:=\sum_{\substack{a(Z)=x\\\diam Z=r}}\Phi(Z).
  \label{eq:ball_shell_decomposition}
\end{align}
Then $V_{x,r}$ is supported in $B_L(x,r)$.

\begin{proposition}[Reference-Hamiltonian stability criterion]
\label{prop:reference_stability_criterion}
Assume the reference family in
Eq.~\eqref{eq:gap_reference_hamiltonian} satisfies the hypotheses just
stated.  Suppose
\begin{equation}
  \snorm{V_{x,r}}\le J_* (1+r)^{-q}, \quad q>d+2,
  \label{eq:gap_stability_strength}
\end{equation}
with $J_*$ uniform in $x,r,L$,
and consider Hamiltonian
\begin{equation}
  H_L(\boldsymbol{c}) = H_{L,*}+\sum_{x,r}c_{x,r}V_{x,r},
  \qquad 0\le c_{x,r}\le1.
  \label{eq:coefficientwise_tail_path}
\end{equation}
Let $P_{L,*}$ denote the ground-space projection of $H_{L,*}$.
Then there exist constants $J_{\mathrm{stab}}>0$ and $L_0<\infty$,
depending only on $q$ and the uniform reference-family data, such that,
if $J_*<J_{\mathrm{stab}}$, then for every $L\ge L_0$ and every
coefficient pattern $0\le c_{x,r}\le1$, the Hamiltonian $H_L(\boldsymbol{c})$ in
Eq.~\eqref{eq:coefficientwise_tail_path} has an isolated lowest
spectral band $\Sigma_L^{\mathrm{low}}$ with spectral projection $P_L$
satisfying
\[
  \operatorname{rank}P_L=\operatorname{rank}P_{L,*},
  \qquad
  \operatorname{dist}\!\left(
    \Sigma_L^{\mathrm{low}},
    \sigma(H_L)\setminus\Sigma_L^{\mathrm{low}}
  \right)
  \ge \frac{\gamma_*}{2}.
\]
The band need not remain exactly degenerate and may have nonzero
internal width.  If $P_{L,*}$ has rank one, then
$\Sigma_L^{\mathrm{low}}$ is a singleton and the displayed separation
is the ordinary ground-state gap
$E_1-E_0\ge\gamma_*/2$.
\end{proposition}

The cited stability theorem also supplies a finite-size bound on the internal splitting, but no such width estimate is needed here.

\begin{proof}
For every coefficient pattern in
Eq.~\eqref{eq:coefficientwise_tail_path}, the perturbation is still a
$(J_*,f_q)$ perturbation, with $f_q(r)=(1+r)^{-q}$, in the ball decomposition of
Ref.~\cite{Michalakis_Zwolak_2013}. 
Its stability theorem supplies such a same-rank isolated low-energy band with separating gap at least $\gamma_*/2$,
uniformly in the coefficients and
in sufficiently large $L$.  If $R\ge R_*$, the direct path can be
written relative to $H_{L,*}$ with coefficients equal to $1$ for
$R_*<r\le R$ and $s$ for $r>R$.  A shell path has coefficients $1$,
$s$, and $0$ in the corresponding three ranges.  Both are special
cases of Eq.~\eqref{eq:coefficientwise_tail_path}.
\end{proof}

The tail condition in the proposition follows directly from the
interaction norm used in this paper.  If
$\snorm{\Phi}_{F_\eta}<\infty$ and $d+2<q<\eta$, then for $r>R_*$,
\begin{align}
  \snorm{V_{x,r}}
  &\le (1+r)^{-\eta}\snorm{\Phi}_{F_\eta}
  \nonumber\\
  &\le J_*(1+r)^{-q},
  \qquad
  J_*=(1+R_*)^{-(\eta-q)}\snorm{\Phi}_{F_\eta}.
  \label{eq:practical_gap_stability_test}
\end{align}
Consequently, $J_*<J_{\mathrm{stab}}$ is a directly verifiable sufficient
condition for all the direct and shell paths above to have a uniform separating gap
in all sufficiently large volumes.  The reference-family conditions must
still be checked separately.  For the product reference
$H_{L,*}=\sum_{x\in\Gamma_L}(\mathbf1-|\psi_x\rangle\langle\psi_x|)$,
with each $\psi_x$ normalized, the global and local gaps equal $1$ and
Local-TQO holds with
$\Delta_0=0$; treating each fixed-size product cell as one site gives
the same example. For general frustration-free models, Local-TQO and Local-Gap require separate verification.  We stress that this is only a special sufficient condition for the path gap, not a necessary condition.

This criterion addresses the path-gap hypothesis.  Once the hypothesis
is available, the truncation rate is determined by the local tail mass,
which we now evaluate for two-body interactions.

\subsection{Two-body power-law and exponential interactions}
\label{sec:two_body}

Section~\ref{sec:fermionic_extension} shows that the resulting
rates also apply to parity-even fermionic interactions.

\medskip
\noindent\emph{Power-law bonds.}
Let us first consider power-law two-body interactions, which arise in Hamiltonians such as
\[
H=\frac12\sum_{\substack{x,y\in\mathbb Z^d\\x\ne y}}J_{xy}O_xO_y,\quad|J_{xy}|\le J(1+|x-y|)^{-p},
\]
with $O_x, O_y$ on-site operators.
For an interaction $\Phi$, the truncation at range $R$ removes every pair
term $\Phi(\{x,y\})$ with $\dist(x,y)>R$. As shown in
Section~\ref{sec:weighted_response}, the local-distinguishability error
is controlled by the local tail mass.  Here this local tail mass is
obtained by summing the norms of all removed terms involving a fixed
site $x$.  Counting the sites in lattice shells yields the effective
decay exponent $p-d$; in particular, $p>2d$ allows one to choose a
weighted-response exponent satisfying $d<\eta<p-d$.

More concretely, suppose $\Phi(Z)=0$ for $|Z|>2$, let
$J_0=\sup_x\snorm{\Phi(\{x\})}<\infty$, and assume, for some $p>d$,
\begin{equation}
  \snorm{\Phi(\{x,y\})}
  \le J(1+\dist(x,y))^{-p}.
  \label{eq:two_body_decay}
\end{equation}
Since a distance-$r$ shell contains $O(r^{d-1})$ sites,
\begin{align}
  \tailmass(R)
  &\le C_dJ\sum_{r>R}r^{d-1}(1+r)^{-p}
  \nonumber\\
  &\le C_{d,p}J(1+R)^{-(p-d)}.
  \label{eq:two_body_tail_mass}
\end{align}
Thus lattice-shell multiplicity changes the individual-coupling
exponent $p$ into the local-tail exponent $p-d$.

The same count gives $\snorm{\Phi}_{F_\eta}<\infty$ for every
$\eta<p-d$.  Since Proposition~\ref{prop:weighted_response} requires
$\eta>d$, the coupling-decay bound verifies its background locality
condition when $p>2d$.  Once any $d<\eta<p-d$ is fixed,
Theorem~\ref{thm:main_finite_response} can be applied directly to
Eq.~\eqref{eq:two_body_tail_mass}, yielding the full rate
$O(R^{-(p-d)})$.  For $d<p\le2d$, the tail still vanishes, but the
coupling-decay bound alone does not supply the required
Lieb--Robinson input.

This distinction limits the physical range of the result.  Dipolar
couplings ($p=3$) satisfy $p>2d$ in one dimension but not in two or
three dimensions.  Couplings of van der Waals type ($p=6$) are covered in one
and two dimensions, while the strict inequality excludes the
three-dimensional endpoint $p=2d$.  For one-dimensional trapped-ion
power-law models, the result applies when $p>2$.  Extending the
observable response estimate to $d<p\le2d$ remains open.

Along a geometrically increasing scale sequence
$R_k=\lambda^kR$, $\lambda>1$, summation introduces only a constant,
since
\begin{equation}
  \sum_{k\ge0}(1+\lambda^kR)^{-(p-d)}
  \le
  \frac{2^{p-d}}{1-\lambda^{-(p-d)}}(1+R)^{-(p-d)}.
  \label{eq:two_body_power_geometric_sum}
\end{equation}
Consequently, under the finite-volume direct-path
hypothesis (Assumption~\ref{assumption:adiabaticity-direct-path}), or when the successive-shell family satisfies
Assumption~\ref{assumption:adiabaticity-shell-path}, with unique
ground state,
\begin{equation}
  d_X(\omega_\Lambda,\omega_{\Lambda,R})
  \le C|X|J(1+R)^{-(p-d)},
  \qquad p>2d,
  \label{eq:two_body_final}
\end{equation}
for $R>1$.  Under the hypotheses of
Theorem~\ref{thm:main_infinite_shell}, the corresponding statement is
\begin{equation}
  d_X(\omega_\infty,\omega_{R_j})
  \le C|X|J(1+R_j)^{-(p-d)},
  \qquad R_j=\lambda^jR_0.
  \label{eq:two_body_infinite_shell}
\end{equation}
Equations~\eqref{eq:two_body_final} and
\eqref{eq:two_body_infinite_shell} do not assert an infinite-volume
direct power-law estimate.  Under the corresponding finite-volume
direct-family clause or Assumption~\ref{assumption:adiabaticity-shell-path},
the same $R^{-(p-d)}$ rate holds for the corresponding
compressed-observable error in
Eqs.~\eqref{eq:main_direct_sector} and
\eqref{eq:main_shell_sector}.

Within the power-law interaction class considered here, the exponent in $O(R^{-(p-d)})$ is optimal. The example in Appendix~\ref{app:sharpness} attains this rate within the full class allowed by our hypotheses, which does not require translation invariance; it rules out uniform improvement over that class, but not over translation-invariant
or otherwise structured subclasses.

\medskip
\phantomsection
\label{sec:two_body_exponential}
\noindent\emph{Exponentially decaying bonds.} Exponentially decaying interactions arise in Hamiltonians such as
\[
H=\frac12\sum_{\substack{x,y\in\mathbb Z^d\\x\ne y}}J_{xy}O_xO_y,\quad|J_{xy}|\le J_e e^{-\mu_0 \dist(x,y)},
\]

For such exponentially decaying bonds, coupling decay and the exponential diameter summability lead to different endpoint statements.  

Suppose again that $\Phi(Z)=0$ for $|Z|>2$ and
$J_0=\sup_x\snorm{\Phi(\{x\})}<\infty$, but now assume the following coupling bound:
\begin{equation}
  \snorm{\Phi(\{x,y\})}
  \le
  J e^{-\mu_0\dist(x,y)}.
  \label{eq:two_body_exponential_decay}
\end{equation}
By shell counting, for $R>1$, we have
\begin{align}
  \tailmass(R)
  &\le
  C_dJ\sum_{r>R}(1+r)^{d-1}e^{-\mu_0r}
  \nonumber\\
  &\le
  C_{d,\mu_0}J(1+R)^{d-1}e^{-\mu_0R}.
  \label{eq:two_body_exponential_tail}
\end{align}
Geometric summation preserves the same polynomial-exponential form
because
\begin{align}
  &\sum_{k\ge0}(1+\lambda^kR)^{d-1}e^{-\mu_0\lambda^kR}
  \nonumber\\
  &\quad\le
  (1+R)^{d-1}e^{-\mu_0R}
  \sum_{k\ge0}
  \lambda^{k(d-1)}e^{-\mu_0(\lambda^k-1)}
  \nonumber\\
  &\quad\le
  C_{d,\mu_0,\lambda}(1+R)^{d-1}e^{-\mu_0R}.
  \label{eq:two_body_exponential_geometric}
\end{align}
Hence in the unique ground state case the finite-volume direct-path and successive-shell estimates
obey
\begin{equation}
  d_X(\omega_\Lambda,\omega_{\Lambda,R})
  \le
  C|X|J(1+R)^{d-1}e^{-\mu_0R},
  \label{eq:two_body_exponential_finite}
\end{equation}
and the thermodynamic-limit shell branch obeys
\begin{equation}
  d_X(\omega_\infty,\omega_{R_j})
  \le
  C|X|J(1+R_j)^{d-1}e^{-\mu_0R_j}.
  \label{eq:two_body_exponential_infinite_shell}
\end{equation}
Each statement is subject to its corresponding gap hypotheses above.
For an isolated finite-volume low-energy sector, the same
polynomial-exponential rate holds for the compressed-observable error
in a direct or shell comparison.

The polynomial prefactor $(1+R)^{d-1}$ does not change the exponential decay class:
for every $0<\kappa<\mu_0$,
\begin{equation}
  (1+R)^{d-1}e^{-\mu_0R}
  \le
  C_{d,\mu_0,\kappa}e^{-\kappa R}.
\end{equation}
Moreover, the coupling-decay bound
\eqref{eq:two_body_exponential_decay} implies exponential diameter
summability at every rate $\mu<\mu_0$.  Consequently, under the
corresponding gap hypotheses, the finite-volume, shell, and intrinsic
direct estimates are all $O(e^{-\kappa R})$ for every
$\kappa<\mu_0$.

Closely related static bounds were obtained by Wang and Hazzard
\cite{Wang_Hazzard_2023}, who proved a power-law
local-perturbations-perturb-locally principle and applied it term by
term to bound finite-size errors of local ground-state observables.
For two-body \(r^{-p}\) interactions with \(p>2d\), their finite-size
bound has the power \(L^{-(p-d)}\), up to logarithmic factors. Their
application concerns interactions crossing a finite-volume boundary,
whereas we globally truncate the interaction range and derive a
general local tail mass bound for many-body interactions, together
with sector-level and infinite-volume extensions.

\subsection{Extension to fermionic lattice systems}
\label{sec:fermionic_extension}

Having completed the spin-system applications, we now turn to fermions.
The finite-volume and infinite-volume conclusions of this paper extend
to lattice fermions, with the following precise conventions.  Put a
finite number of fermionic modes at each $x\in\Gamma$, let
$\mA_X^{\mathrm{CAR}}$ be the local CAR algebra on $X$, and let
$\Theta$ denote the parity automorphism.  We require every interaction
term to be even,
\begin{equation}
  \Theta\bigl(\Phi(Z)\bigr)=\Phi(Z),
  \label{eq:fermion_even_interaction}
\end{equation}
and use the same metric supports, diameter norms, and tail masses as in
Section~\ref{sec:setup}.  The physical observable algebra is the even
subalgebra $(\mA_X^{\mathrm{CAR}})^+$.  The statements about local
distinguishability and multipoint moments apply, with the same rates and parameter dependence, to observables in these parity-even local subalgebras.

We work directly with the CAR algebra; no spin representation is used, which is important because Jordan–Wigner transformations do not generally preserve geometric locality.
First, disjoint CAR algebras graded-commute, and therefore an even
observable supported in $X$ commutes with every observable supported in
$X^c$.  Second, an even Hamiltonian preserves parity under its
physical-time dynamics.  Consequently every term
\begin{equation}
  K_{s,Z}=-\mathcal I_s\bigl(\dot\Phi_s(Z)\bigr)
  \label{eq:fermion_even_generator_term}
\end{equation}
of the spectral-flow generator is even.  Third, the fermionic
conditional expectation constructed for the CAR algebra is unital,
contractive, parity preserving, and maps an even quasi-local operator
to an even operator in the target local algebra
\cite{teufel2025liebrobinsonboundsautomorphicequivalence}.  Its
localization estimate is obtained from commutator bounds against all
operators in the complementary algebra.  Here the improved fermionic
Lieb--Robinson bound applies because the evolved source
$\dot\Phi_s(Z)$ is even, while the test operator may be arbitrary.
Integrating that estimate against the inverse-Liouvillian filter gives
the fermionic counterpart of
Eq.~\eqref{eq:main_termwise_localization},
\begin{align}
  &\snorm{K_{s,Z}-\mathbb E^{\mathrm{CAR}}_{\Lambda\setminus X}
  (K_{s,Z})}
  \nonumber\\
  &\qquad\le
  C|Z|\snorm{\dot\Phi_s(Z)}
  F_\beta\bigl(\dist(Z,X)\bigr).
  \label{eq:fermion_termwise_localization}
\end{align}
The localized auxiliary generator is even and supported in $X^c$, so
its unitary commutes with every physical observable in
$(\mA_X^{\mathrm{CAR}})^+$.  The Duhamel argument proving
Proposition~\ref{prop:weighted_response}, and hence all subsequent
tail-mass summations, is therefore unchanged.

The remaining parts of the argument are algebra independent.  In
finite volume, the inverse-Liouvillian identity, spectral projections,
and unitary transport of an isolated sector are Hilbert-space
statements.  In infinite volume, the derivation and ground-state
condition are defined on the CAR quasi-local algebra in the same way;
an even reference product state gives the canonical graded-product
extension used in the shell compactness argument.  For the intrinsic
direct-path construction, the automorphic-equivalence result invoked in
Section~\ref{sec:infinite_volume} was proved directly for fermion
systems \cite{becker2025automorphicequivalencegappedphases}.  Thus the
fermionic version requires no extra decay exponent beyond the one in
the corresponding spin statement.

Finally, if $\omega$ and $\nu$ are parity-even fermionic states, their local density
matrices commute with the local parity operator.  Hence
$D_X=\rho_X^\omega-\rho_X^\nu$ is even, and the norm-dual optimizer
$\operatorname{sgn}(D_X)$ is even as well.  It follows that
\begin{equation}
  \sup_{\substack{A=A^*\in(\mA_X^{\mathrm{CAR}})^+\\
                   \snorm{A}\le1}}
  \abs{\omega(A)-\nu(A)}
  =\snorm{\rho_X^\omega-\rho_X^\nu}_1.
\label{eq:fermion_trace_distance}
\end{equation}
The local distinguishability conclusion is therefore identical.  Similar
arguments apply to the low-energy sector projection in
Theorem~\ref{thm:main_finite_response}.  Thus the preceding two-body rates
apply without change to even fermionic interactions.  We next illustrate
them in quadratic models for which the path gaps and local observables can
both be computed exactly.

\subsection{Fermionic Gaussian models and numerics}
\label{sec:massive_fermion_benchmark}

We conclude the applications with three free-fermion benchmarks built
from the two-orbital, number-conserving Bloch form
\begin{equation}
  h_R(k)=\mathbf d_R(k)\cdot\boldsymbol\sigma,
  \qquad
  P_R(k)=\frac12\left[
  \mathbf1-\frac{h_R(k)}{|\mathbf d_R(k)|}
  \right],
  \label{eq:gaussian_general_projector}
\end{equation}
where $P_R$ projects onto the occupied band.  The first model is
uniformly massive and converges substantially faster than the general
tail-mass bound.  The second and third models approach a critical point
in one and two dimensions, respectively, and probe the predicted
$R^{-(p-d)}$ behavior over finite near-critical crossover windows.  In
all three cases, the gap along every direct and shell path can be shown
analytically; so no numerical gap extrapolation is needed. For simplicity, the numerical comparisons below use the direct paths.

For later use, let $n_{a,x}=a_x^\dagger a_x$ and
$M_x=n_{a,x}-n_{b,x}$.  Translation invariance and Wick's theorem give
\begin{align}
  \langle n_{a,x}\rangle_R
  &=\int_{\mathrm{BZ}}\frac{\dd^d k}{(2\pi)^d}
  [P_R(k)]_{aa},
  \label{eq:gaussian_density_observable}
  \\
  \langle n_{a,x}n_{a,x+r}\rangle_R
  &=\langle n_{a,x}\rangle_R^2
  \nonumber\\
  &\quad-\left|
  \int_{\mathrm{BZ}}\frac{\dd^d k}{(2\pi)^d}
  e^{\ii k\cdot r}[P_R(k)]_{aa}
  \right|^2,
  \quad r\ne0.
  \label{eq:gaussian_density_pair_observable}
\end{align}
For any local observable $A$ considered below, define its full-tail
truncation error by
\begin{equation}
  \Delta A_R
  :=
  \langle A\rangle_R-\langle A\rangle_\infty,
  \qquad
  \langle A\rangle_\infty
  :=
  \lim_{R'\to\infty}\langle A\rangle_{R'}.
  \label{eq:gaussian_error_definition}
\end{equation}
The limit exists in the models below because the hopping tails are
absolutely summable and the occupied-band projectors remain uniformly
gapped.  In the numerics we use
$\Delta^{(\mathrm{ref})}A_R
=\langle A\rangle_R-\langle A\rangle_{R_{\mathrm{ref}}}$.  All plotted
errors use the stated finite reference cutoff $R_{\mathrm{ref}}$; the grid
and reference-cutoff checks below quantify the remaining numerical
sensitivity.  Reported exponents are ordinary least-squares fits of
$\log|\Delta^{(\mathrm{ref})}A_R|$ against $\log R$ over the stated
windows, treating each dyadic cutoff equally; values below $100$ times
machine precision are excluded.  The data are deterministic, so we assign
no sampling error bars.

\paragraph{Uniformly massive model.}
Put two fermionic modes $a_x,b_x$ in each unit cell and, on a finite
volume $\Lambda$, define
\begin{align}
  H_\Lambda(T)
  ={}&m\sum_{x\in\Lambda}
  \bigl(a_x^\dagger a_x-b_x^\dagger b_x\bigr)
  \nonumber\\
  &+\sum_{x,y\in\Lambda}
  \left(T_{xy}a_x^\dagger b_y+{T_{xy}^*}b_y^\dagger a_x\right),
  \qquad m\ne0.
  \label{eq:massive_two_orbital_model}
\end{align}
In the one-particle basis $(a,b)$ its matrix and square are
\begin{align}
  h(T)&=
  \begin{pmatrix}m\mathbf1&T\\T^\dagger&-m\mathbf1\end{pmatrix},
  \\
  h(T)^2&=
  \begin{pmatrix}
    m^2\mathbf1+TT^\dagger&0\\
    0&m^2\mathbf1+T^\dagger T
  \end{pmatrix}
  \ge m^2\mathbf1.
  \label{eq:massive_model_gap_identity}
\end{align}
Hence $\sigma(h(T))\cap(-|m|,|m|)=\varnothing$ for every hopping
matrix $T$, with open or periodic boundaries.  Filling all negative
one-particle modes at chemical potential zero gives a unique Gaussian
many-body ground state with Fock-space gap at least $|m|$.  Range
truncation and both interpolations in
Eqs.~\eqref{eq:direct_path}--\eqref{eq:shell_path} only replace $T$ by
another hopping matrix, so the same bound holds along every path.

For $T_{xy}=t_{y-x}$ the Bloch Hamiltonian is
\begin{align}
  h_R(k)&=m\sigma_z+\operatorname{Re}q_R(k)\sigma_x
  -\operatorname{Im}q_R(k)\sigma_y,
  \\
  q_R(k)&=\sum_{|r|\le R}t_r e^{\ii k\cdot r}.
  \label{eq:massive_model_bloch_symbol}
\end{align}
We choose in one dimension
$t_0=J=1$, $t_r=J(1+|r|)^{-p}$ for $r\ne0$,
$m=0.7$, and $p=2.25$.  The momentum grid has $2^{16}$ points, the
large-cutoff reference is $R_{\mathrm{ref}}=2^{14}$, and the plotted
cutoffs are dyadic.  Fits over $32\le R\le512$ give slopes
$-3.358$ for $\langle n_{a,x}\rangle$ and $-3.345$ for
$\langle n_{a,x}n_{a,x+1}\rangle$, substantially faster than the
general $-(p-1)=-1.25$ exponent. The plot is shown in Fig.~\ref{fig:gaussian_benchmarks}(a)(d).

The faster convergence has a simple intuitive interpretation. It reflects the rigidity of the occupied-band projector. At \(m=0\), positivity of \(q_R(k)\) makes the occupied orbital independent of momentum, so changing the long-range dispersion does not change the many-body ground-state projector. For \(m\ne0\), the mass mixes the filled and empty local orbitals, and the response of a neutral local observable to a distant hopping proceeds through a uniformly gapped virtual interband process. The corresponding real-space projector response supplies a second \(r^{-p}\) factor, leading to the upper-bound scale \(R^{-(2p-d)}\). For the chosen parameters, this scale is $R^{-3.5}$, compatible with the faster decay observed in the numerical fitting window. A more quantitative discussion of this behavior can be found in Appendix\ \ref{app:massive_gaussian_rate}.

This model remains noncritical as $m$ is reduced: at $m=0$,
Eq.~\eqref{eq:massive_symbol_lower_bound} implies
$P_R(k)=(\mathbf1-\sigma_x)/2$ for every $R$, so the truncation error
vanishes exactly,  
while for small \(m\) the response coefficients are \(O(m)\).
Within this model family, a slow crossover can instead arise when
$\min_k|q_\infty(k)|\to0$, so that $m$ sets the gap.  Near the
incipient band touching, the occupied-band projector then varies across
a momentum window of width $O(\xi^{-1})$, where $\xi$ is the
gap-controlled crossover length.  Correspondingly, the Fourier
coefficients of the projector's first-variation kernel extend to
distances of order $\xi$; for $r\ll\xi$ they follow their near-critical
power-law or logarithmic profile rather than the eventual
$O(r^{-p})$ decay.  They therefore do not supply a second $r^{-p}$
factor in the truncation sum.  The following two examples realize this
near-critical regime, with $\delta$ as the controllable gap-opening parameter.

\paragraph{One-dimensional near-critical model.}
In a nearly critical but still gapped model with correlation length $\xi\gg1$, the two-point correlator connecting the distant endpoint of a discarded bond to the local observable varies only slowly for separations $r\ll\xi$.  It therefore provides little additional suppression in this regime, allowing the discarded bonds to accumulate nearly according to their bare tail mass.  We use these models to probe the truncation error of fixed local observables in this crossover regime.

To probe this regime, consider a one-dimensional model, for $\delta>0$ and $p>1$,
\begin{align}
  h_R^{(1)}(k)
  ={}&\sin k\,\sigma_x
  +\left[\delta+1-\cos k+v_R^{(1)}(k)\right]\sigma_z,
  \label{eq:near_critical_1d_model}
  \\
  v_R^{(1)}(k)
  ={}&2\lambda\sum_{r=1}^{R}
  \frac{\cos(rk)}{(1+r)^p}.
  \label{eq:near_critical_1d_tail}
\end{align}

For the numerical comparison, it suffices to consider the direct
interpolation between the range-$R$ and full models:
\begin{equation}
  h_{R,s}^{(1)}(k)
  =
  h_R^{(1)}(k)
  +s\left[h_\infty^{(1)}(k)-h_R^{(1)}(k)\right],
  \qquad 0\le s\le1.
  \label{eq:near_critical_1d_direct_path}
\end{equation}
Thus $s=0$ gives the range-$R$ model and $s=1$ gives the full model.
Let $E_{R,s}^{(1)}(k)$ denote its positive band energy.  Along this
path, the total long-range contribution satisfies
\begin{equation}
  \abs{v_{R,s}^{(1)}(k)}
  \le
  2|\lambda|\sum_{r\ge1}(1+r)^{-p}.
\end{equation}
Since the model without the long-range contribution has band energy at
least $\delta$, the triangle inequality gives
\begin{equation}
  E_{R,s}^{(1)}(k)
  \ge
  \delta-2|\lambda|\sum_{r\ge1}(1+r)^{-p}.
  \label{eq:near_critical_1d_path_gap}
\end{equation}

For $p=2.25$, $\delta=10^{-4}$, and $\lambda=2\times10^{-5}$,
the right-hand side is $8.1592\times10^{-5}>0$.  The correlation
length is of order $\xi\sim\delta^{-1}=10^4$.  We use $2^{20}$
momenta and $R_{\mathrm{ref}}=2^{18}$.  Over $2\le R\le64$, well
inside $R\ll\xi$, the fits for $\langle M_x\rangle$ and
$\langle n_{a,x}n_{a,x+1}\rangle$ give slopes $-1.258$ and
$-1.257$, respectively, close to $-(p-1)=-1.25$. The plot is shown in Fig.~\ref{fig:gaussian_benchmarks}(b)(e).

Intuitively, in the window \(R\ll\xi\), the correlations entering the response vary slowly with distance and provide little additional suppression beyond the decay of the discarded couplings. The truncation error is expected to follow their summed strength, \(R^{-(p-1)}\). When \(R\) becomes comparable to \(\xi\), these correlations begin to decay, leading to faster convergence.

\begin{figure*}[t]
  \centering
  \includegraphics[width=\textwidth]{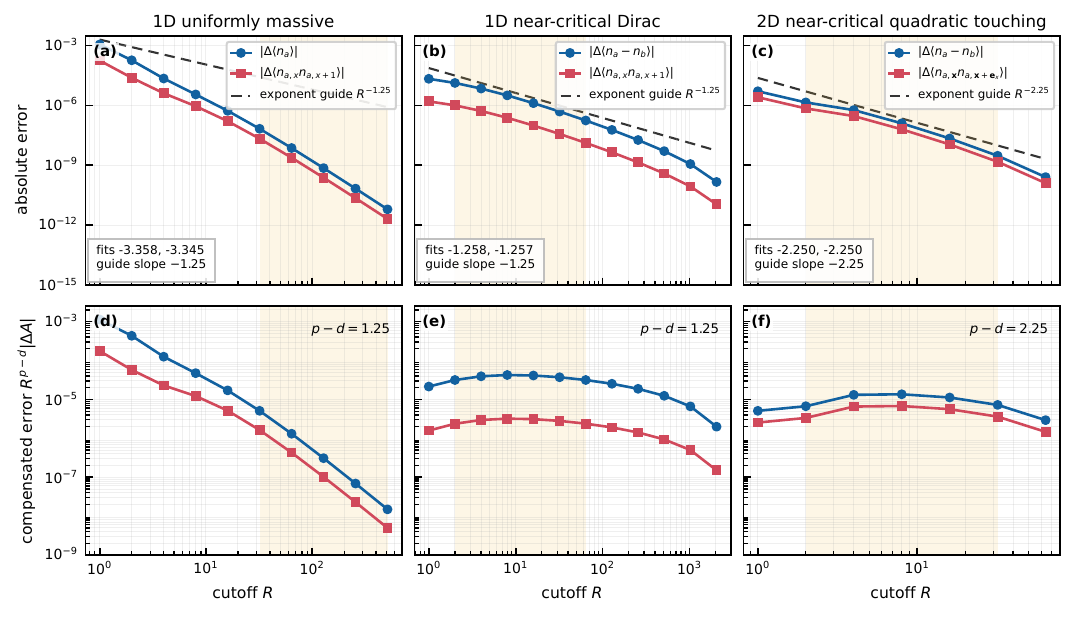}
  \caption{Gaussian truncation benchmarks.  The columns show the
  one-dimensional uniformly massive model, the one-dimensional
  near-critical Dirac model, and the two-dimensional near-critical
  quadratic touching.  The top row gives absolute errors relative to a
  large-cutoff reference; the bottom row multiplies the same errors by
  the power $R^{p-d}$ from Eq.~\eqref{eq:two_body_final}.  The shaded
  intervals are the fit windows, and the two fitted slopes in each top
  panel follow the legend order.  In panels (a) and (d), the convergence is much faster than the general
  $R^{-1.25}$ rate.  Panels (b) and (e) approach that rate below the
  one-dimensional correlation length.  Panels (c) and (f) are consistent
  with the two-dimensional $R^{-2.25}$ crossover and logarithmic correction
  in Eq.~\eqref{eq:qbt_log_corrected_scaling}.  All panels in a given row
  use the same vertical scale.  Each dashed curve is a display-only guide
  proportional to $R^{-(p-d)}$, with its prefactor selected separately in
  each column so that it lies above every displayed point.  The dashed
  curves are therefore guides to the exponent in the general bound, not
  evaluations of the theorem's nonuniversal quantitative prefactor; the
  compensated panels likewise test scaling rather than that prefactor.}
  \label{fig:gaussian_benchmarks}
\end{figure*}

\paragraph{Two-dimensional near-critical model.}
To test the dimensional shift explicitly, we consider a two-dimensional quadratic band
touching with Euclidean range cutoff,
\begin{align}
  h_R^{(2)}(\mathbf k)
  ={}&d_x(\mathbf k)\sigma_x+d_y(\mathbf k)\sigma_y
  +d_{z,R}(\mathbf k)\sigma_z,
  \label{eq:near_critical_2d_model}
  \\
  d_x(\mathbf k)&=\cos k_x-\cos k_y,
  \qquad
  d_y(\mathbf k)=\sin k_x\sin k_y,
  \\
  d_{z,R}(\mathbf k)
  ={}&\delta+\lambda
  \sum_{0<|\mathbf r|_2\le R}
  \frac{e^{\ii\mathbf k\cdot\mathbf r}}
  {(1+|\mathbf r|_2)^p}.
  \label{eq:near_critical_2d_tail}
\end{align}
The last expression is real because the sum contains both
$\mathbf r$ and $-\mathbf r$.  Moreover,
\begin{equation}
  d_x(\mathbf k)^2+d_y(\mathbf k)^2
  =\left(1-\cos k_x\cos k_y\right)^2.
  \label{eq:qbt_identity}
\end{equation}
The massless model has quadratic touchings at $(0,0)$ and
$(\pi,\pi)$; near either one the left-hand side of
Eq.~\eqref{eq:qbt_identity} is $|\mathbf q|^4/4+O(|\mathbf q|^6)$.
Thus $\delta>0$ opens a gap and gives
$\xi\sim\delta^{-1/2}$.

Writing
$S_p=\sum_{\mathbf r\ne0}(1+|\mathbf r|_2)^{-p}$, grouping sites into $\ell^{\infty}$ shells \footnote{The $\ell^\infty$ shell is \(Q_n=\{\mathbf r:\|\mathbf r\|_\infty:=\max\{|r_x|, |r_y|\}=n\}\), which contains \(8n\) sites, and then we use \(|\mathbf r|_2\ge n\) on \(Q_n\).} gives the rigorous estimate
\begin{equation}
  S_p\le8\sum_{n\ge1}\frac{n}{(1+n)^p}
  =8\bigl[\zeta(p-1)-\zeta(p)\bigr].
  \label{eq:qbt_lattice_sum_bound}
\end{equation}
Along the direct interpolation
\begin{equation}
  h_{R,s}^{(2)}(\mathbf k)
  =
  h_R^{(2)}(\mathbf k)
  +s\left[
    h_\infty^{(2)}(\mathbf k)-h_R^{(2)}(\mathbf k)
  \right],
  \qquad 0\le s\le1,
\end{equation}
the long-range mass contribution has magnitude at most
$|\lambda|S_p$.  Hence
\begin{equation}
  E_{R,s}^{(2)}(\mathbf k)
  \ge \delta-|\lambda|S_p.
  \label{eq:qbt_path_gap}
\end{equation}
For $p=4.25$, $\delta=3\times10^{-5}$, and
$\lambda=10^{-5}$, Eq.~\eqref{eq:qbt_path_gap} is at least
$2.2624\times10^{-5}$.  We use a $2048^2$ momentum grid,
$R_{\mathrm{ref}}=960$, and dyadic cutoffs through $R=64$.

For the orbital imbalance, the near-touching linear-response kernel is
\begin{equation}
  \frac{d_x^2+d_y^2}
  {(d_x^2+d_y^2+\delta^2)^{3/2}}
  \sim\frac{2}{|\mathbf q|^2}
  \qquad
  \bigl(\sqrt\delta\ll|\mathbf q|\ll1\bigr).
  \label{eq:qbt_response_kernel}
\end{equation}
Its two-dimensional Fourier transform is logarithmic.  Hence, for the
truncation error $\Delta A_R$ defined in
Eq.~\eqref{eq:gaussian_error_definition} and for $1\ll R\ll\xi$, one expects
\begin{equation}
  |\Delta A_R|
  \sim R^{-(p-2)}
  \left[c_1\log(\xi/R)+c_2\right].
  \label{eq:qbt_log_corrected_scaling}
\end{equation}
A fit over $2\le R\le32$ gives $-2.2503$ for both
$\langle M_x\rangle$ and
$\langle n_{a,\mathbf x}n_{a,\mathbf x+\mathbf e_x}\rangle$, compared
with $-(p-2)=-2.25$.  The compensated curve is not expected to be
perfectly flat because of the logarithm in
Eq.~\eqref{eq:qbt_log_corrected_scaling}.  The agreement of the fitted
slopes with $-2.25$,  supports the predicted finite-window crossover; see Fig.~\ref{fig:gaussian_benchmarks}(c)(f).

As a convergence check, increasing the one-dimensional near-critical
grid from $2^{20}$ to $2^{21}$ points and $R_{\mathrm{ref}}$ from
$2^{18}$ to $2^{19}$ changes every observable value used in the fit by
at most $7\times10^{-16}$.  In two dimensions, increasing the grid from
$1024^2$ to $2048^2$ at
$R_{\mathrm{ref}}=480$ changes either fitted slope by less than
$6\times10^{-4}$; on the $2048^2$ grid, increasing $R_{\mathrm{ref}}$
from $480$ to $960$ changes either slope by less than $4\times10^{-7}$.
These Gaussian examples complement
the interacting spin system construction in Appendix~\ref{app:sharpness}:
the latter establishes exponent sharpness over the full class allowed by
our hypotheses, without a translation-invariance restriction, whereas
Fig.~\ref{fig:gaussian_benchmarks} shows both enhanced Gaussian convergence
and a finite-window approach toward the general rate near criticality.

\section{Discussion}
\label{sec:discussions}

The results establish a volume-independent relation between the local
error caused by range truncation and the local mass of the discarded
interaction terms.  In finite volume this relation holds for a unique
ground state and more generally for an isolated
low-energy sector.  Successive finite-range shell flows give a
thermodynamic-limit branch for power-law interactions, while an
intrinsic infinite-volume direct flow gives superpolynomial estimates
under the stronger path assumptions of
Theorem~\ref{thm:main_infinite_direct}.  The argument also applies to
parity-even fermionic interactions.  For two-body couplings of order
$r^{-p}$, the resulting error is $O(R^{-(p-d)})$ when $p>2d$; the
non-translation-invariant example in Appendix~\ref{app:sharpness} attains
this rate and rules out a uniform improvement over the full class covered
by our hypotheses.  It does not settle the optimal rate of truncation errors in
translation-invariant or other structurally restricted subclasses.

These bounds provide a reference for range cutoffs used in
finite-size calculations.  For an observable supported in $X$, a
cutoff may be chosen by requiring the relevant tail bound to be smaller
than the desired accuracy divided by $|X|$.  This applies directly to
long-range lattice models and to effective Hamiltonians with nonlocal
couplings, provided the corresponding interpolation is gapped.  

The main unresolved assumption is the uniform gap along the direct or
shell path.  Proposition~\ref{prop:reference_stability_criterion}
verifies it in all sufficiently large volumes when a finite-range
truncation is a frustration-free
reference satisfying Local-TQO and Local-Gap and the remaining tail is
perturbatively small.  The Gaussian models of
Section~\ref{sec:massive_fermion_benchmark} verify it by exact
one-particle estimates.  These cases motivate a softer open question. Let $\Phi$ be a self-adjoint, polynomially decaying interaction on $\mathbb Z^d$ (i.e., $\snorm{\Phi}_{F_\eta}<\infty$ for some $\eta>d$),
and suppose that $\Phi$ has a locally unique GNS-gapped ground state.
Does there exist $R_*<\infty$ such that, for every $R>R_*$, the truncated
interaction $\Phi_{\le R}$ has a GNS-gapped ground state?  The gap may
depend on $R$.  A common lower bound, uniqueness of the ground
states after truncation, or a gap along the direct or shell interpolation would be a
separate, stronger conclusion.

Several further questions are worth exploring. First, can the response
estimate for two-body interactions be extended from \(p>2d\) to the
summable regime \(d<p\le2d\) under comparably general assumptions?
Second, which structural properties yield convergence faster than the
local-tail bound, as in the uniformly massive Gaussian example? It
would also be useful to construct an intrinsic infinite-volume direct
flow for power-law decaying interactions; the present shell
construction avoids assuming such a flow but selects a
cutoff-dependent branch. Open boundaries require separate treatment
in phases with gapless edge modes, because a bulk gap need not imply
the finite-volume gap assumed here. Extending these local-observable
estimates to gapless systems is also interesting, but may require methods beyond the
spectral flow method in gapped systems.

\begin{acknowledgments}
The author acknowledges GPT 5.6 Sol for assistance in refining the Applications section and in drafting the manuscript. The author has verified all AI-generated content and is responsible for all analytical and numerical results included in this manuscript.
\end{acknowledgments}

\section*{Data Availability}
The code and data supporting the numerical results in this article are
available at \url{https://github.com/kangle-morning/truncated-interaction-gaussian-fermion}.

\clearpage
\bibliography{reference}

\clearpage
\appendix

\section{Isolated finite-volume sectors and matched states}
\label{app:finite_volume_sectors}

This section proves the isolated low-energy sector part of
Theorem~\ref{thm:main_finite_response} and discusses its state
interpretation.  The comparison is made after identifying the two
isolated sectors by the chosen spectral flow.  It therefore requires no
assumption of local indistinguishability within either sector.

Let $s\mapsto H_{\Lambda}(s)$ be one of the direct or shell paths and
let $P_s$ be the continuously tracked projection in
Eq.~\eqref{eq:sector_projection}.  Let $U_s$ and $\alpha_s$ be the
spectral-flow unitary and observable flow defined in
Eqs.~\eqref{eq:finite_spectral_flow_unitary} and
\eqref{eq:finite_spectral_flow_automorphism}.  The transport identity
\eqref{eq:finite_spectral_flow_transport} uses only the external
separation \eqref{eq:isolated_sector_gap}; eigenvalues may split or
cross inside the isolated band.

\begin{proposition}[Compressed-observable and matched-state bounds]
\label{prop:compressed_sector_bound}
Under the hypotheses of Proposition~\ref{prop:weighted_response}, let
$P_0$ and $P_1$ be the spectral projections onto the isolated
low-energy sector at $s=0$ and $s=1$, respectively.
Then, for every $A_X\in\mA_X$,
\begin{equation}
  \snorm{
    U_1^*P_1A_XP_1U_1-P_0A_XP_0
  }
  \le
  \snorm{\alpha_1(A_X)-A_X}.
  \label{eq:compression_from_automorphism}
\end{equation}
For the truncation specializations, suppose that the chosen direct path satisfying the finite-volume 
Assumption~\ref{assumption:adiabaticity-direct-path}, or that the chosen
shell path belongs to a family satisfying
Assumption~\ref{assumption:adiabaticity-shell-path}, respectively.
Then the direct and shell paths satisfy
Eqs.~\eqref{eq:main_direct_sector} and
\eqref{eq:main_shell_sector}, respectively.
\end{proposition}

\begin{proof}
Equation~\eqref{eq:finite_spectral_flow_transport} gives the exact identity
\begin{align}
  &U_1^*P_1A_XP_1U_1-P_0A_XP_0
  \nonumber\\
  &\qquad=
  P_0\bigl(\alpha_1(A_X)-A_X\bigr)P_0.
  \label{eq:compressed_response_identity}
\end{align}
Compression by an orthogonal projection is contractive, which proves
Eq.~\eqref{eq:compression_from_automorphism}.  For the direct path,
$\dot\Phi_s=\Phi_{>R}$ and
$\mT_{\dot\Phi_s}=\tailmass(R)$.  For a shell path,
$\dot\Phi_s=\Phi_{(R,R']}$ and
$\mT_{\dot\Phi_s}\le\tailmass(R)$.  Applying
Eq.~\eqref{eq:local_mass_response} in the two cases proves the stated
projector bounds.  Neither this argument nor its constant depends on the
rank of $P_s$ or on the internal width of its spectral band.
\end{proof}

In particular, the result controls all diagonal and off-diagonal matrix
elements of the local observable in the spectrally matched bases.  For
$\psi,\varphi\in\operatorname{Ran}P_0$,
\begin{align}
  &\abs{
    \langle U_1\psi,A_XU_1\varphi\rangle
    -
    \langle\psi,A_X\varphi\rangle
  }
  \nonumber\\
  &\qquad\le
  C|X|\snorm{A_X}\tailmass(R)
  \snorm{\psi}\snorm{\varphi},
  \label{eq:sector_matrix_elements}
\end{align}
For a shell step from $R$ to $R'$, the same formula holds with its shell
spectral flow and the same upper bound by $\tailmass(R)$.

For completeness, let
\begin{equation}
  \mathfrak S(P)
  :=
  \left\{
    \rho\ge0:
    \operatorname{Tr}\rho=1,\ 
    \rho=P\rho P
  \right\}
  \label{eq:sector_state_space}
\end{equation}
be the density matrices supported in the range of $P$.  If
$\rho_0\in\mathfrak S(P_0)$, define its \emph{spectral-flow-matched}
state at $s=1$ by
\begin{equation}
  \rho_1=U_1\rho_0U_1^*\in\mathfrak S(P_1).
  \label{eq:matched_sector_state}
\end{equation}
For self-adjoint $A_X$,
\begin{align}
  &\abs{
    \operatorname{Tr}(\rho_1A_X)
    -
    \operatorname{Tr}(\rho_0A_X)
  }
  \nonumber\\
  &\qquad=
  \abs{
    \operatorname{Tr}\left[
      \rho_0P_0
      \bigl(\alpha_1(A_X)-A_X\bigr)P_0
    \right]
  }
  \nonumber\\
  &\qquad\le
  \snorm{
    P_0\bigl(\alpha_1(A_X)-A_X\bigr)P_0
  }.
  \label{eq:matched_state_estimate}
\end{align}
Thus every state in one sector has a matched state in the other sector
obeying the same local distinguishability bound as in the rank-one
case.

This matching can equivalently be expressed as a local Hausdorff bound.
For two sets of finite-volume states define
\begin{align}
  d_{\mathrm H,X}(\mathcal S_0,\mathcal S_1)
  :=
  \max\biggl\{
  &\sup_{\rho_0\in\mathcal S_0}
    \inf_{\rho_1\in\mathcal S_1}d_X(\rho_0,\rho_1),
  \nonumber\\[-1mm]
  &\sup_{\rho_1\in\mathcal S_1}
    \inf_{\rho_0\in\mathcal S_0}d_X(\rho_0,\rho_1)
  \biggr\}.
  \label{eq:local_hausdorff_definition}
\end{align}
Because conjugation by $U_1$ is a bijection from
$\mathfrak S(P_0)$ to $\mathfrak S(P_1)$,
\begin{equation}
  d_{\mathrm H,X}\bigl(
    \mathfrak S(P_0),\mathfrak S(P_1)
  \bigr)
  \le
  C|X|\tailmass(R)
  \label{eq:sector_hausdorff}
\end{equation}
for either interpolation path, with the appropriate tail bound.  This
does not say that arbitrary, unmatched states in the two sectors are
close.  Such a statement would require an additional
local-indistinguishability hypothesis.

For several finite-volume shell steps, the endpoint compatibility
specified below Eq.~\eqref{eq:sector_projection} permits the recursion
\begin{equation}
  \rho_{k+1}=U_k\rho_kU_k^*,
  \qquad
  \rho_k\in\mathfrak S(P_{\Lambda,R_k}).
  \label{eq:finite_sector_shell_recursion}
\end{equation}
Each matched pair obeys the shell comparison bound, and these bounds
may then be summed by the triangle inequality.  When the $P_{\Lambda,R_k}$ are exact
ground-space projections, every $\rho_k$ is an exact ground state.  For
approximately degenerate bands, the same finite-volume statements hold,
but a thermodynamic ground-state conclusion requires their widths
$\delta_{\Lambda_n}$ to vanish along the chosen exhaustion.
Theorem~\ref{thm:main_infinite_shell} deliberately retains its unique exact
ground-state assumptions and therefore does not use this additional
extension.

\section{Lieb--Robinson bounds for individual terms}
\label{app:lr_bound}

This section separates the two dynamical locality inputs used in the
manuscript.  We first record the standard finite-volume $F$-function
bound and its diameter-norm specialization.  We then state the improved
long-range bound of Ref.~\cite{teufel2025liebrobinsonboundsautomorphicequivalence} (c.f. Refs.~\onlinecite{Matsuta_2016,Else_2020})
in the spin-system conventions used here.  The final subsection combines
that bound with the inverse-Liouvillian filter and the spin conditional
expectation to prove the locality estimate for each interaction term used in
Proposition~\ref{prop:weighted_response}.  The two Lieb--Robinson
theorems are imported with precise assumptions; the passage from the
improved bound to this estimate is proved below.

\subsection{Standard finite-volume \texorpdfstring{$F$}{F}-function bound}

For completeness, we record the formulation of Theorem~3.1 of
Ref.~\cite{Nachtergaele_2019}.  Let
$(\Gamma,\dist)$ be a countable metric space, and write
$d_{xy}:=\dist(x,y)$.  A nonincreasing function
$F\colon[0,\infty)\to(0,\infty)$ is an \emph{$F$-function} if
\begin{align}
  \snorm{F}_1
  &:=
  \sup_{x\in\Gamma}
  \sum_{y\in\Gamma}F(d_{xy})
  <\infty,
  \label{eq:lr_F_integrability}\\
  C_F
  &:=
  \sup_{x,y\in\Gamma}
  \sum_{z\in\Gamma}
  \frac{F(d_{xz})F(d_{zy})}{F(d_{xy})}
  <\infty.
  \label{eq:lr_F_convolution}
\end{align}
The first condition is uniform integrability and the second is the
convolution condition.

Let $I\Subset\mathbb R$ be an interval and let
$t\mapsto\Xi_t$ be a time-dependent interaction such that
$t\mapsto\Xi_t(Z)\in\mA_Z$ is strongly continuous for every
$Z\Subset\Gamma$.  In the present finite-dimensional local algebras this
is equivalent to norm continuity.  Define its pair-distance norm by
\begin{equation}
  J_F(\Xi_t)
  :=
  \sup_{x,y\in\Gamma}
  \frac{1}{F\bigl(\dist(x,y)\bigr)}
  \sum_{\substack{Z\Subset\Gamma\\x,y\in Z}}
  \snorm{\Xi_t(Z)},
  \label{eq:lr_pair_norm_general}
\end{equation}
and assume that $t\mapsto J_F(\Xi_t)$ is locally bounded on $I$.  For
$\Lambda\Subset\Gamma$, let
\begin{align}
  H_\Lambda^\Xi(t)
  &:=
  \sum_{Z\Subset\Lambda}\Xi_t(Z),
  \nonumber\\
  \ii\partial_tU_\Lambda^\Xi(t,s)
  &=
  H_\Lambda^\Xi(t)U_\Lambda^\Xi(t,s),
  \qquad
  U_\Lambda^\Xi(s,s)=\mathbf 1,
  \nonumber\\
  \tau_{t,s}^{\Lambda,\Xi}(A)
  &:=
  U_\Lambda^\Xi(t,s)^*A\,U_\Lambda^\Xi(t,s).
  \label{eq:lr_finite_dynamics}
\end{align}
For $X\Subset\Gamma$, define its interaction boundary over $I$ as
\begin{equation}
  \partial_\Xi X
  :=
  \{x\in X:x\text{ satisfies Eq.~\eqref{eq:lr_boundary_condition}}\},
  \label{eq:lr_interaction_boundary}
\end{equation}
where the condition is that there exist $Z\Subset\Gamma$ and $r\in I$
such that
\begin{equation}
  x\in Z,\qquad
  Z\cap(\Gamma\setminus X)\ne\varnothing,\qquad
  \Xi_r(Z)\ne0.
  \label{eq:lr_boundary_condition}
\end{equation}

\begin{theorem}[Lieb--Robinson bound]
\label{thm:appendix_lr}
Let $F$ and $t\mapsto\Xi_t$ satisfy all the assumptions above.
Let $X,Y\Subset\Gamma$ be disjoint, let
$X\cup Y\Subset\Lambda\Subset\Gamma$, and let
$A\in\mA_X$ and $B\in\mA_Y$.  Set
\begin{align}
  I_{t,s}(\Xi)
  &:=
  C_F
  \int_{\min\{t,s\}}^{\max\{t,s\}}
  J_F(\Xi_r)\,\dd r,
  \label{eq:lr_time_integral}\\
  D_{F,1}^\Xi(X,Y)
  &:=
  \sum_{x\in X}\sum_{y\in\partial_\Xi Y}
  F\bigl(\dist(x,y)\bigr),
  \nonumber\\
  D_{F,2}^\Xi(X,Y)
  &:=
  \sum_{x\in\partial_\Xi X}\sum_{y\in Y}
  F\bigl(\dist(x,y)\bigr),
  \nonumber\\
  D_F^\Xi(X,Y)
  &:=
  \min\{D_{F,1}^\Xi(X,Y),D_{F,2}^\Xi(X,Y)\}.
  \label{eq:lr_boundary_factor}
\end{align}
Then, for all $s,t\in I$,
\begin{equation}
  \snorm{[\tau_{t,s}^{\Lambda,\Xi}(A),B]}
  \le
  \frac{2\snorm{A}\snorm{B}}{C_F}
  \left(e^{2I_{t,s}(\Xi)}-1\right)
  D_F^\Xi(X,Y).
  \label{eq:lr_full_bound}
\end{equation}
All quantities on the right-hand side are defined using the ambient
interaction on $\Gamma$; in particular, the estimate is uniform in
$\Lambda$.
\end{theorem}

No spectral-gap assumption enters Theorem~\ref{thm:appendix_lr}.  The
gap is needed later, when the physical-time dynamics is integrated
against an inverse-Liouvillian filter to construct the spectral flow.
Theorem~3.3 of Ref.~\cite{Nachtergaele_2019} also permits arbitrary
time-independent self-adjoint on-site Hamiltonians, treated in the
interaction picture, without changing Eq.~\eqref{eq:lr_full_bound}.  In
our finite-dimensional setting all on-site terms are bounded and may
simply be included among the singleton interaction terms.

\subsection{Diameter-norm specialization and interpolation paths}

We next express the standard bound in terms of the diameter norm used
throughout the manuscript and verify uniformity along both truncation
paths.

\begin{corollary}[Diameter-norm specialization]
\label{cor:appendix_lr_diameter}
Let $\Gamma=\mathbb Z^d$ with the graph metric and let $\eta>d$.  Then
$F_\eta(r)=(1+r)^{-\eta}$ is an $F$-function, with
\begin{equation}
  \snorm{F_\eta}_1<\infty,
  \qquad
  C_\eta:=C_{F_\eta}
  \le
  2^\eta\snorm{F_\eta}_1<\infty.
  \label{eq:lr_polynomial_F_constants}
\end{equation}
If
\begin{equation}
  M_\eta(t)
  :=
  \sup_{x\in\Gamma}
  \sum_{Z\ni x}
  |Z|(1+\diam Z)^\eta\snorm{\Xi_t(Z)}
  \label{eq:lr_time_dependent_diameter_norm}
\end{equation}
is locally bounded on $I$, then, under the remaining assumptions of
Theorem~\ref{thm:appendix_lr}, set
\begin{equation}
  \widehat I_{t,s}
  :=
  C_\eta
  \int_{\min\{t,s\}}^{\max\{t,s\}}M_\eta(r)\,\dd r.
  \label{eq:lr_diameter_time_integral}
\end{equation}
Then
\begin{align}
  \snorm{[\tau_{t,s}^{\Lambda,\Xi}(A),B]}
  &\le
  \frac{2\snorm{A}\snorm{B}}{C_\eta}
  \left(e^{2\widehat I_{t,s}}-1\right)
  \nonumber\\
  &\quad\times
  \sum_{x\in X}\sum_{y\in Y}
  F_\eta\bigl(\dist(x,y)\bigr).
  \label{eq:lr_diameter_bound}
\end{align}
\end{corollary}

\begin{proof}
Uniform lattice-shell counting gives
$\snorm{F_\eta}_1<\infty$ for $\eta>d$.  If
$a=\dist(x,z)$, $b=\dist(z,y)$, and $c=\dist(x,y)$, then
$1+c\le(1+a)+(1+b)$.  Since $\eta>1$, convexity gives
\begin{equation}
  \frac{F_\eta(a)F_\eta(b)}{F_\eta(c)}
  \le
  2^{\eta-1}\bigl[F_\eta(a)+F_\eta(b)\bigr].
\end{equation}
Summing over $z$ proves
$C_\eta\le2^\eta\snorm{F_\eta}_1$.  Moreover, if $x,y\in Z$, then
$\dist(x,y)\le\diam Z$, and therefore
\begin{align}
  J_{F_\eta}(\Xi_t)
  &=
  \sup_{x,y}
  (1+\dist(x,y))^\eta
  \sum_{Z\ni x,y}\snorm{\Xi_t(Z)}
  \nonumber\\
  &\le
  \sup_x
  \sum_{Z\ni x}
  |Z|(1+\diam Z)^\eta\snorm{\Xi_t(Z)}
  =
  M_\eta(t).
  \label{eq:lr_diameter_controls_pair}
\end{align}
Finally, $\partial_\Xi X\Subset X$ and
$\partial_\Xi Y\Subset Y$, so
\begin{equation}
  D_{F_\eta}^\Xi(X,Y)
  \le
  \sum_{x\in X}\sum_{y\in Y}
  F_\eta\bigl(\dist(x,y)\bigr).
\end{equation}
The claim follows from Theorem~\ref{thm:appendix_lr}.
\end{proof}

We now specialize the corollary to the manuscript's interpolation
paths.  Fix the path parameter $u\in[0,1]$ and let
$\Phi_u^\sharp$ denote either
$\Phi_u^{(R,\infty)}$ or $\Phi_u^{(R,R')}$.  Bounding each interaction
term by its counterpart in $\Phi$ gives
\begin{equation}
  J_{F_\eta}(\Phi_u^\sharp)
  \le
  \snorm{\Phi_u^\sharp}_{F_\eta}
  \le
  \snorm{\Phi}_{F_\eta}
  =:M.
  \label{eq:lr_uniform_path_pair_norm}
\end{equation}
Thus the physical-time dynamics generated by $\Phi_u^\sharp$ satisfies,
uniformly in $\Lambda,u,R$, and $R'$, the following bound.  Define
\begin{equation}
  L_{\eta,M}(t)
  :=
  \frac{2}{C_\eta}\left(e^{2C_\eta M|t|}-1\right).
  \label{eq:lr_path_prefactor}
\end{equation}
Then
\begin{align}
  \snorm{[\tau_t^{\Lambda,u,\sharp}(A),B]}
  &\le
  L_{\eta,M}(t)\snorm{A}\snorm{B}
  \nonumber\\[-1mm]
  &\quad\times
  \sum_{x\in X}\sum_{y\in Y}
  F_\eta\bigl(\dist(x,y)\bigr)
  \nonumber\\
  &\le
  L_{\eta,M}(t)\snorm{A}\snorm{B}|X|\,|Y|
  \nonumber\\[-1mm]
  &\quad\times F_\eta\bigl(\dist(X,Y)\bigr).
  \label{eq:lr_path_specialization}
\end{align}
Here
$\tau_t^{\Lambda,u,\sharp}:=\tau_{t,0}^{\Lambda,\Phi_u^\sharp}$
denotes physical-time evolution; $u$ is held fixed and is not the
physical-time variable.

\medskip
\noindent\emph{Exponential decay and $F$-functions.}
The pure exponential $r\mapsto e^{-\mu r}$ is not an $F$-function on
$\mathbb Z^d$.  Indeed, with $x=0$, $y=ne_1$, and
$z=ke_1$ for $0\le k\le n$,
\begin{equation}
  \frac{
  \sum_{z\in\mathbb Z^d}
  e^{-\mu\dist(x,z)}e^{-\mu\dist(z,y)}
  }{
  e^{-\mu\dist(x,y)}
  }
  \ge n+1,
\end{equation}
so the convolution constant is infinite.  On the other hand, for
$\eta>d$ and $\kappa>0$,
\begin{equation}
  F_{\eta,\kappa}(r)
  :=
  e^{-\kappa r}(1+r)^{-\eta}
\end{equation}
is an $F$-function.  Indeed,
$F_{\eta,\kappa}\le F_\eta$ gives uniform integrability, and its
convolution ratio is bounded by that of $F_\eta$, because the triangle
inequality makes the additional exponential factor at most one.  If
$\snorm{\Phi}_{f_\mu}<\infty$, then for every $0<\kappa<\mu$,
\begin{align}
  J_{F_{\eta,\kappa}}(\Phi)
  &\le
  K_{\eta,\mu-\kappa}\snorm{\Phi}_{f_\mu},
  \label{eq:exp_norm_to_weighted_F}\\
  K_{\eta,\delta}
  &:=
  \sup_{r\ge0}(1+r)^\eta e^{-\delta r}<\infty.
  \label{eq:exp_polynomial_constant}
\end{align}
To see the first inequality, if $x,y\in Z$, then
\begin{equation}
  e^{\kappa\dist(x,y)}
  (1+\dist(x,y))^\eta
  \le
  K_{\eta,\mu-\kappa}e^{\mu\diam Z};
\end{equation}
summing over $Z\ni x,y$ and using $|Z|\ge1$ gives
Eq.~\eqref{eq:exp_norm_to_weighted_F}.
Thus the standard Lieb--Robinson kernel is available at every rate
$\kappa<\mu$, but not necessarily at the endpoint $\mu$.  This
kernel-rate loss is distinct from the finite-volume and shell truncation
exponents; see Section~\ref{sec:two_body_exponential}.

\subsection{Improved long-range Lieb--Robinson input}
\label{subsec:tw_improved_lr}

The standard estimate \eqref{eq:lr_path_specialization} has an
exponential physical-time prefactor.  An exact spectral-flow filter has
all polynomial moments but need not have an exponential moment, so that
estimate alone cannot be integrated to obtain the spatial power used in
Proposition~\ref{prop:weighted_response}.  The required input is the
following spin-system specialization of Theorem~6 of
Ref.~\cite{teufel2025liebrobinsonboundsautomorphicequivalence}.  The
reference proves the result for fermions and explains in its Sec.~VI
that it applies to spin systems without a parity restriction.

\begin{theorem}[Improved long-range Lieb--Robinson bound]
\label{thm:tw_improved_lr}
Let $\eta>d$, let
\begin{equation}
  \frac{d+1}{\eta+1}<\sigma<1,
  \label{eq:tw_sigma_range}
\end{equation}
and let $\Xi$ be a time-independent interaction on
$\Lambda\Subset\mathbb Z^d$ satisfying
$\snorm{\Xi}_{F_\eta}\le M$.  Let
$\tau_t^{\Lambda,\Xi}$ be its physical-time dynamics, let
$A\in\mA_Y$ and $B\in\mA_X$, and set $r=\dist(X,Y)$.  Define
\begin{equation}
  c_{\eta,d}:=\max\{2e\snorm{F_\eta}_1,1\},
  \qquad
  v_*:=c_{\eta,d}M.
  \label{eq:tw_velocity_bound}
\end{equation}
There is a finite constant
$C_\sigma=C_\sigma(d,\eta,\sigma)$, independent of
$\Lambda,\Xi,X,Y,A$, and $B$, such that
\begin{align}
  \snorm{[\tau_t^{\Lambda,\Xi}(A),B]}
  &\le
  2\snorm{A}\snorm{B}\min\{|X|,|Y|\}
  \mathcal L_{\eta,\sigma}(r,t).
  \label{eq:tw_improved_lr}
\end{align}
Here
\begin{align}
  \mathcal L_{\eta,\sigma}(r,t)
  &:=
  e^{v_*|t|-r^{1-\sigma}}
  \nonumber\\
  &\quad+
  C_\sigma(1+r)^{-\sigma\eta}v_*|t|
  \nonumber\\[-1mm]
  &\qquad\times
  \bigl(1+(v_*|t|)^{d/(1-\sigma)}\bigr).
  \label{eq:tw_lr_profile}
\end{align}
The same estimate holds in the presence of arbitrary on-site terms.
In the fermionic formulation, $\Xi$ is even and the same estimate holds
provided at least one of $A$ and $B$ is even.
More explicitly, the dependence on $\sigma$ in the cited theorem can
be chosen as
\begin{equation}
  C_\sigma
  =
  C_{d,\eta}
  \left(\sigma-\frac{d+1}{\eta+1}\right)^{-2}
  \frac{1}{1-\sigma}
  \Gamma\left(\frac{d}{1-\sigma}\right).
  \label{eq:tw_sigma_constant}
\end{equation}
\end{theorem}

To match conventions, Ref.~\cite{teufel2025liebrobinsonboundsautomorphicequivalence}
uses
\begin{equation}
  \snorm{\Xi}_{\eta,n}
  :=
  \sup_{x\in\Gamma}
  \sum_{Z\ni x}|Z|^n(1+\diam Z)^\eta\snorm{\Xi(Z)}.
  \label{eq:tw_interaction_norm}
\end{equation}
Thus $\snorm{\Xi}_{\eta,1}=\snorm{\Xi}_{F_\eta}$ and
$\snorm{\Xi}_{\eta,0}\le\snorm{\Xi}_{F_\eta}$.  With the notation
\[
  \snorm{F_\eta}_\Lambda
  :=
  \sup_{x\in\Lambda}\sum_{y\in\Lambda}
  F_\eta\bigl(\dist(x,y)\bigr)
\]
the velocity in the cited theorem is
\begin{equation}
  \max\left\{
    2e\snorm{F_\eta}_{\Lambda}\snorm{\Xi}_{\eta,0},
    \snorm{\Xi}_{\eta,1}
  \right\},
\end{equation}
which is at most $v_*$ because
$\snorm{F_\eta}_{\Lambda}\le\snorm{F_\eta}_1$.  This proves that the
stated constants are uniform along any family satisfying
Eq.~\eqref{eq:background_bound}, including both interpolation paths.
The proof of Theorem~\ref{thm:tw_improved_lr} itself is not repeated
here; it is the iterative long-range Lieb--Robinson argument of the
cited reference.

\subsection{Inverse-Liouvillian bound for individual terms}
\label{subsec:termwise_localization}

We now prove the estimate for individual interaction terms used in the
main text.  Let $s\mapsto\Phi_s$ satisfy
Eq.~\eqref{eq:background_bound}, and suppose that $H_s$ carries an
isolated low-energy sector with separating gap at least $\gamma>0$.
We use the dynamics, filter, and inverse-Liouvillian map defined in
Eqs.~\eqref{eq:finite_path_physical_dynamics}--
\eqref{eq:finite_inverse_liouvillian}.
The same filter applies for all $s$ and $\Lambda$ and is the $\delta=0$
filter of
Ref.~\cite{teufel2025liebrobinsonboundsautomorphicequivalence}; thus no
assumption on the internal width of the isolated sector is needed.

For the finite spin system, the conditional expectation onto the
algebra outside $X$ has the Haar representation
\begin{equation}
  \mathbb E_{\Lambda\setminus X}(B)
  =
  \int_{\mathcal U(\mA_X)}U^*BU\,\dd\mu_X(U),
  \label{eq:conditional_expectation_haar}
\end{equation}
where $\mu_X$ is normalized Haar measure.  Hence
\begin{equation}
  \snorm{B-\mathbb E_{\Lambda\setminus X}(B)}
  \le
  \sup_{U\in\mathcal U(\mA_X)}\snorm{[B,U]}.
  \label{eq:conditional_expectation_commutator}
\end{equation}
Indeed,
$B-U^*BU=U^*[U,B]$, and Eq.~\eqref{eq:conditional_expectation_commutator}
follows by integrating and using the triangle inequality.

For a finite CAR algebra there is instead a unital, completely
positive, contractive conditional expectation
$\mathbb E^{\mathrm{CAR}}_{\Lambda\setminus X}$ which preserves parity
and maps even operators into
$(\mA^{\mathrm{CAR}}_{\Lambda\setminus X})^+$
\cite{teufel2025liebrobinsonboundsautomorphicequivalence}.  Its
commutator characterization says that an even $B$ is close to this
conditional expectation whenever $[B,U]$ is small for every
$U\in\mA_X^{\mathrm{CAR}}$.  If the interaction and its path
derivative are even, then $O_Z=\dot\Phi_s(Z)$ and
$\tau_t^s(O_Z)$ are even.  The fermionic version of
Theorem~\ref{thm:tw_improved_lr} therefore supplies the required
commutator estimate with an arbitrary $U$, because one of the two
operators is even.  We use $\mathbb E$ below for either conditional
expectation.

\begin{proposition}[Locality bound for individual terms]
\label{prop:appendix_termwise}
Fix $\eta>d$ and $0<\varepsilon<\eta-d$, and set
$\beta=\eta-\varepsilon>d$.  Under the uniform background and gap
assumptions above, there is a
constant $C=C(d,\eta,\varepsilon,\gamma,M)$ such that, for every
$s\in[0,1]$, every nonempty $Z\Subset\Lambda$, and every local region
$X$, the terms $K_{s,Z}$ defined in
Eq.~\eqref{eq:finite_spectral_flow_generator} satisfy
\begin{equation}
  \snorm{K_{s,Z}-\mathbb E_{\Lambda\setminus X}(K_{s,Z})}
  \le
  C|Z|\snorm{\dot\Phi_s(Z)}
  F_\beta\bigl(\dist(Z,X)\bigr).
  \label{eq:appendix_termwise}
\end{equation}
The constant is independent of $\Lambda,s,X,Z$, and $\diam Z$.
For a fermionic system the same statement holds when the interaction
and its path derivative are even, with $\mathbb E$ interpreted as the
CAR conditional expectation.
\end{proposition}

\begin{proof}
Choose
\begin{equation}
  \max\left\{
    \frac{d+1}{\eta+1},
    \frac{\beta}{\eta}
  \right\}
  <\sigma<1.
  \label{eq:termwise_sigma_choice}
\end{equation}
Such a choice is possible because $\eta>d$ and $\beta<\eta$.  Fix
$s,Z,X$, abbreviate $O_Z:=\dot\Phi_s(Z)$, and set
$r:=\dist(Z,X)$.  If $r=0$, contractivity of the conditional
expectation and Eqs.~\eqref{eq:finite_filter_moments} and
\eqref{eq:finite_inverse_liouvillian} give
\begin{align}
  \snorm{K_{s,Z}-\mathbb E_{\Lambda\setminus X}(K_{s,Z})}
  &\le
  2m_0(W_\gamma)\snorm{O_Z}
  \nonumber\\
  &\le
  2m_0(W_\gamma)|Z|\snorm{O_Z}F_\beta(0).
  \label{eq:termwise_zero_distance}
\end{align}

Suppose $r\ge1$.  Put
\begin{align}
  a:=1-\sigma,\qquad
  q:=\frac{d}{1-\sigma},\qquad
  \overline v:=\max\{1,v_*\},
  \nonumber\\
  T_r:=\frac{r^a}{2\overline v}.
  \label{eq:termwise_time_split}
\end{align}
For $|t|\le T_r$, apply
Theorem~\ref{thm:tw_improved_lr} to $O_Z$ and a unitary
$U\in\mathcal U(\mA_X)$, use
$\min\{|Z|,|X|\}\le|Z|$, and then use
Eq.~\eqref{eq:conditional_expectation_commutator}.  Since
$v_*|t|-r^a\le-r^a/2$, this gives
\begin{align}
  &\snorm{
    (\id-\mathbb E_{\Lambda\setminus X})\tau_t^s(O_Z)
  }
  \nonumber\\
  &\quad\le
  2|Z|\snorm{O_Z}
  \Bigl[
    e^{-r^a/2}
    +C_\sigma(1+r)^{-\sigma\eta}\overline v|t|
    \nonumber\\[-1mm]
  &\hspace{42mm}\times
    \bigl(1+(\overline v|t|)^q\bigr)
  \Bigr].
  \label{eq:termwise_short_time}
\end{align}
For $|t|>T_r$, we instead use contractivity:
\begin{equation}
  \snorm{
    (\id-\mathbb E_{\Lambda\setminus X})\tau_t^s(O_Z)
  }
  \le
  2\snorm{O_Z}
  \le
  2|Z|\snorm{O_Z}.
  \label{eq:termwise_long_time}
\end{equation}
Integrating Eqs.~\eqref{eq:termwise_short_time} and
\eqref{eq:termwise_long_time} against $|W_\gamma(t)|$ yields
\begin{align}
  &\snorm{K_{s,Z}-\mathbb E_{\Lambda\setminus X}(K_{s,Z})}
  \nonumber\\
  &\quad\le
  2|Z|\snorm{O_Z}
  \Bigl[
    m_0(W_\gamma)e^{-r^a/2}
    +C_\sigma J_{\sigma,\overline v}(1+r)^{-\sigma\eta}
    \nonumber\\[-1mm]
  &\hspace{43mm}
    +\int_{|t|>T_r}|W_\gamma(t)|\,\dd t
  \Bigr],
  \label{eq:termwise_integrated}
\end{align}
where
\begin{equation}
  J_{\sigma,\overline v}
  :=
  \int_{\mathbb R}|W_\gamma(t)|\overline v|t|
  \bigl(1+(\overline v|t|)^q\bigr)\,\dd t
  <\infty
  \label{eq:termwise_filter_integral}
\end{equation}
by Eq.~\eqref{eq:finite_filter_moments}.

Choose an integer $N$ such that $Na\ge\beta$.  The same moment bound
gives
\begin{align}
  \int_{|t|>T_r}|W_\gamma(t)|\,\dd t
  &\le
  m_N(W_\gamma)(1+T_r)^{-N}
  \nonumber\\
  &\le
  (2\overline v)^Nm_N(W_\gamma)(1+r)^{-Na}
  \nonumber\\
  &\le
  (2\overline v)^Nm_N(W_\gamma)(1+r)^{-\beta}.
  \label{eq:termwise_filter_tail}
\end{align}
Here we used
$1+T_r\ge(1+r^a)/(2\overline v)$ and
$1+r^a\ge(1+r)^a$.  Finally,
\begin{equation}
  e^{-r^a/2}
  \le C_{\beta,\sigma}(1+r)^{-\beta},
  \qquad
  (1+r)^{-\sigma\eta}
  \le(1+r)^{-\beta},
  \label{eq:termwise_final_decay}
\end{equation}
where the second inequality follows from
$\sigma\eta>\beta$ in Eq.~\eqref{eq:termwise_sigma_choice}.
Combining Eqs.~\eqref{eq:termwise_zero_distance},
\eqref{eq:termwise_integrated}, \eqref{eq:termwise_filter_tail}, and
\eqref{eq:termwise_final_decay} proves
Eq.~\eqref{eq:appendix_termwise}.  All filter moments depend only on
the fixed gap $\gamma$, while $\overline v$ is controlled by
$d,\eta$, and $M$, which proves the asserted uniformity.
\end{proof}

Proposition~\ref{prop:appendix_termwise} proves directly, for the spin
conditional expectation, the single-term estimate used before spatial
summation; its relation to the common-region LPPL summation in Theorem~10
of Ref.~\cite{teufel2025liebrobinsonboundsautomorphicequivalence} is
explained after Proposition~\ref{prop:weighted_response}.

\section{An alternative finite-range shell estimate}
\label{app:finite_range_shell}

The sharp shell estimate in the main text uses the improved long-range
Lieb--Robinson input of Appendix~\ref{app:lr_bound}.  For comparison, we
now give a proof independent of the improved long-range
Theorem~\ref{thm:tw_improved_lr} and
Proposition~\ref{prop:appendix_termwise}.  Its locality input is only
the standard $F$-function bound, Theorem~\ref{thm:appendix_lr}; it also
uses the usual exact quasi-adiabatic filter and the gap hypothesis.
No result of
Ref.~\cite{teufel2025liebrobinsonboundsautomorphicequivalence} enters
the derivation below.
The proof applies because every Hamiltonian along a shell path has
finite range $R'$.  Its price is a response kernel of width $R'$ and
hence an additional volume factor $(1+R')^d$.  This appendix is not
used in the proofs of the main theorems.

The distinction from the sharp shell estimate occurs in the spatial
summation.  Equation~\eqref{eq:anchor_sum} uses the fixed kernel
$F_\beta(r)=(1+r)^{-\beta}$, whose $\ell^1$ norm is finite and independent
of $R$ and $R'$.  The required background norm is uniform in both
scales by Eq.~\eqref{eq:lr_uniform_path_pair_norm}.  Consequently, the
reduction in Corollary~\ref{cor:local_mass_bound} produces only the
factor $|X|$; applying it to the shell derivative gives the shell
estimate in Theorem~\ref{thm:main_finite_response} with no factor growing
with the upper shell radius.  By contrast, the finite-range argument
below produces a rescaled kernel of width $L=1+R'$.  Its lattice sum is
of order $L^d$, as shown explicitly in
Eq.~\eqref{eq:finite_range_kernel_l1}; this is the source of the extra
volume factor in Eq.~\eqref{eq:finite_range_shell_response}.

\subsection{A one-shell spectral-flow bound}

We first derive the finite-range propagation estimate with all
constants uniform in the volume and the truncation scales.  The
argument retains the minimum-support boundary factor in
Theorem~\ref{thm:appendix_lr}; this avoids the unnecessary second factor
of $|X|$ that would result from estimating both supports separately.

\begin{proposition}[Alternative finite-range shell bound]
\label{prop:finite_range_shell}
Let $\eta>d$, let $\snorm{\Phi}_{F_\eta}\le M$, and fix
$1<R<R'$.  Suppose that the family
$\{(\Lambda,R,R'):\Lambda\text{ is a finite volume under consideration}\}$
satisfies Assumption~\ref{assumption:adiabaticity-shell-path}, with common
gap $\gamma>0$.  Let
$\alpha_s^{(R,R')}$ be its spectral flow.  Then
\begin{equation}
  \snorm{\alpha_1^{(R,R')}(A_X)-A_X}
  \le
  C|X|\snorm{A_X}(1+R')^d\tailmass(R),
  \label{eq:finite_range_shell_response}
\end{equation}
where $C=C(d,\eta,M,\gamma)$ is independent of
$\Lambda,R,R'$, and $X$.
\end{proposition}

\begin{proof}
Set $L:=1+R'$ and $a:=L^{-1}$, and exponentially tilt the polynomial
$F$-function:
\begin{equation}
  F_{\eta,a}(r):=e^{-ar}F_\eta(r).
  \label{eq:finite_range_tilted_F}
\end{equation}
Write $d_{uv}:=\dist(u,v)$.  The triangle inequality gives
\begin{align}
  \frac{F_{\eta,a}(d_{xz})F_{\eta,a}(d_{zy})}
       {F_{\eta,a}(d_{xy})}
  &=
  e^{-a(d_{xz}+d_{zy}-d_{xy})}
  \nonumber\\
  &\quad\times
  \frac{F_\eta(d_{xz})F_\eta(d_{zy})}
       {F_\eta(d_{xy})}.
  \label{eq:finite_range_tilt_convolution}
\end{align}
Consequently,
\begin{equation}
  \snorm{F_{\eta,a}}_1\le\snorm{F_\eta}_1,
  \qquad
  C_{\eta,a}:=C_{F_{\eta,a}}\le C_\eta.
  \label{eq:finite_range_tilt_constants}
\end{equation}
The term $z=x$ in the convolution sum also gives
$C_{\eta,a}\ge F_{\eta,a}(0)=1$.

Every nonzero term of $\Phi_s^{(R,R')}$ has diameter at most $R'$.
Hence, using Eqs.~\eqref{eq:lr_pair_norm_general} and
\eqref{eq:lr_diameter_controls_pair},
\begin{equation}
  J_{F_{\eta,a}}\bigl(\Phi_s^{(R,R')}\bigr)
  \le
  e^{aR'}J_{F_\eta}\bigl(\Phi_s^{(R,R')}\bigr)
  \le eM.
  \label{eq:finite_range_tilted_pair_norm}
\end{equation}
Define
\begin{equation}
  \kappa:=2eC_\eta M.
  \label{eq:finite_range_kappa}
\end{equation}
For disjoint nonempty $X,Y$, the two boundary sums in
Eq.~\eqref{eq:lr_boundary_factor} satisfy
\begin{equation}
  D_{F_{\eta,a}}^{\Phi_s}(X,Y)
  \le
  \snorm{F_\eta}_1
  \min\{|X|,|Y|\}
  e^{-\dist(X,Y)/L}.
  \label{eq:finite_range_boundary_factor}
\end{equation}
Theorem~\ref{thm:appendix_lr}, together with
Eqs.~\eqref{eq:finite_range_tilt_constants}--\eqref{eq:finite_range_kappa},
therefore gives
\begin{align}
  \snorm{[\tau_t^s(A_X),B_Y]}
  &\le
  C_0\snorm{A_X}\snorm{B_Y}
  \min\{|X|,|Y|\}
  \nonumber\\
  &\quad\times
  e^{-\dist(X,Y)/L+\kappa|t|},
  \label{eq:finite_range_light_cone}
\end{align}
where $C_0$ depends only on $d$ and $\eta$.  Combining this with the
trivial commutator bound, and increasing $C_0$ if necessary, yields for
all $X,Y$, including overlapping supports,
\begin{align}
  \snorm{[\tau_t^s(A_X),B_Y]}
  &\le
  C_0\snorm{A_X}\snorm{B_Y}
  \min\{|X|,|Y|\}
  \nonumber\\
  &\quad\times
  \min\left\{
    e^{-\dist(X,Y)/L+\kappa|t|},1
  \right\}.
  \label{eq:finite_range_light_cone_min}
\end{align}

Introduce the filtered finite-range kernel
\begin{equation}
  \mathcal K_L(r)
  :=
  \int_{\mathbb R}|W_\gamma(t)|
  \min\left\{e^{-r/L+\kappa|t|},1\right\}\,\dd t.
  \label{eq:finite_range_filtered_kernel}
\end{equation}
Let $\overline\kappa:=\max\{1,\kappa\}$, fix an integer $N>d$, and
split the integral at
\begin{equation}
  T_r:=\frac{r}{2\overline\kappa L}.
\end{equation}
For $|t|\le T_r$, the exponential in
Eq.~\eqref{eq:finite_range_filtered_kernel} is at most
$e^{-r/(2L)}$.  For $|t|>T_r$, use the second entry in the minimum and
the moment bound \eqref{eq:finite_filter_moments}.  This gives
\begin{align}
  \mathcal K_L(r)
  &\le
  m_0(W_\gamma)e^{-r/(2L)}
  \nonumber\\
  &\quad+
  m_N(W_\gamma)
  \left(1+\frac{r}{2\overline\kappa L}\right)^{-N}.
  \label{eq:finite_range_kernel_decay}
\end{align}
Uniform lattice-shell counting on $\mathbb Z^d$, followed by the
rescaling $r=Lu$, now gives
\begin{equation}
  \sup_{x\in\Gamma}
  \sum_{y\in\Gamma}
  \mathcal K_L\bigl(\dist(x,y)\bigr)
  \le
  C_1L^d,
  \label{eq:finite_range_kernel_l1}
\end{equation}
where $C_1=C_1(d,\eta,M,\gamma)$; convergence of the polynomial part
uses $N>d$.

For the shell path,
\begin{equation}
  \dot\Phi_s^{(R,R')}=\Phi_{(R,R')},
  \qquad
  G_s=-\mathcal I_s\bigl(\dot H_s\bigr).
\end{equation}
Apply Eq.~\eqref{eq:finite_range_light_cone_min} with
$B_Y=\dot\Phi_s^{(R,R')}(Z)$ and use
$\min\{|X|,|Z|\}\le|Z|$.  Equations
\eqref{eq:finite_inverse_liouvillian} and
\eqref{eq:finite_range_filtered_kernel} imply
\begin{align}
  \snorm{[G_s,A_X]}
  &\le
  C_0\snorm{A_X}
  \sum_Z|Z|\snorm{\dot\Phi_s^{(R,R')}(Z)}
  \nonumber\\
  &\qquad\times
  \mathcal K_L\bigl(\dist(Z,X)\bigr).
  \label{eq:finite_range_generator_before_anchor}
\end{align}
For every nonempty $Z$, choose one nearest anchor $z_Z\in Z$ such that
$\dist(z_Z,X)=\dist(Z,X)$.  Since
$\mathcal K_L$ is nonincreasing, grouping each term at this single
anchor and then using Eq.~\eqref{eq:finite_range_kernel_l1} gives
\begin{align}
  &\sum_Z|Z|\snorm{\dot\Phi_s^{(R,R')}(Z)}
  \mathcal K_L\bigl(\dist(Z,X)\bigr)
  \nonumber\\
  &\quad=
  \sum_{z\in\Gamma}
  \mathcal K_L\bigl(\dist(z,X)\bigr)
  \sum_{Z:\,z_Z=z}|Z|
  \snorm{\dot\Phi_s^{(R,R')}(Z)}
  \nonumber\\
  &\quad\le
  \mT_{\dot\Phi_s^{(R,R')}}
  \sum_{x\in X}\sum_{z\in\Gamma}
  \mathcal K_L\bigl(\dist(z,x)\bigr)
  \nonumber\\
  &\quad\le
  C_1|X|L^d\tailmass(R).
  \label{eq:finite_range_anchor_sum}
\end{align}
Here we used
$\mT_{\dot\Phi_s^{(R,R')}}\le\tailmass(R)$.
Finally,
\begin{equation}
  \frac{\dd}{\dd s}\alpha_s^{(R,R')}(A_X)
  =
  \ii U_s^*[G_s,A_X]U_s.
\end{equation}
Integrating its norm over $s\in[0,1]$ and using
$L=1+R'$ proves Eq.~\eqref{eq:finite_range_shell_response}.
\end{proof}

The same automorphism estimate immediately controls the
compressed-sector expression in Eq.~\eqref{eq:main_shell_sector}, by
the argument of Appendix~\ref{app:finite_volume_sectors}.  In the
rank-one case it controls local distinguishability through
Eq.~\eqref{eq:state_from_automorphism}.

\subsection{Shell sequence and comparison with the sharp rate}

The preceding estimate reproduces the elementary Cauchy construction
for adjacent truncation scales, but with the assumptions and constants
made uniform.

\begin{corollary}[Dyadic finite-range comparison]
\label{cor:finite_range_shell_dyadic}
Let $R_k=2^kR$ with $R>1$, and for every finite volume $\Lambda$ under
consideration let $N_\Lambda$ be the first index for which
$R_{N_\Lambda}\ge\diam\Lambda$.  Assume that
$\mathcal S_{R,2}$ in Eq.~\eqref{eq:finite_shell_family}
satisfies Assumption~\ref{assumption:adiabaticity-shell-path}, with every
prescribed sector being the rank-one projection onto the unique ground
state.  Then
\begin{equation}
  d_X(\omega_\Lambda,\omega_{\Lambda,R})
  \le
  C|X|
  \sum_{k=0}^{N_\Lambda-1}
  (1+R_{k+1})^d\tailmass(R_k),
  \label{eq:finite_range_shell_dyadic_general}
\end{equation}
If
$\snorm{\Phi}_{F_\eta}<\infty$ with $\eta>d$, then
\begin{equation}
  d_X(\omega_\Lambda,\omega_{\Lambda,R})
  \le
  C|X|\snorm{\Phi}_{F_\eta}
  (1+R)^{-(\eta-d)}.
  \label{eq:finite_range_shell_dyadic_power}
\end{equation}
Here $C$ depends on $d,\eta,\gamma$, and an upper bound for
$\snorm{\Phi}_{F_\eta}$, but not on $\Lambda$ or $R$.
\end{corollary}

\begin{proof}
Apply Proposition~\ref{prop:finite_range_shell} to each pair
$(R_k,R_{k+1})$ and sum by the triangle inequality.  This proves
Eq.~\eqref{eq:finite_range_shell_dyadic_general}.  The generic tail
bound \eqref{eq:tail_mass_power} and
$1+R_{k+1}\le2(1+R_k)$ give
\begin{equation}
  (1+R_{k+1})^d\tailmass(R_k)
  \le
  2^d\snorm{\Phi}_{F_\eta}
  (1+R_k)^{-(\eta-d)}.
\end{equation}
The resulting geometric series converges because $\eta>d$, proving
Eq.~\eqref{eq:finite_range_shell_dyadic_power}.
\end{proof}

Equation~\eqref{eq:finite_range_shell_dyadic_power} loses $d$ powers
relative to the sharp bound \eqref{eq:finite_shell_power}.  For the
two-body interaction \eqref{eq:two_body_decay}, inserting the direct
shell count \eqref{eq:two_body_tail_mass} gives the alternative rate
$R^{-(p-2d)}$ for $p>2d$, rather than the sharp
$R^{-(p-d)}$ in Eq.~\eqref{eq:two_body_final}.

\subsection{Concrete consequences for rapidly decaying tails}

The conclusion of the alternative proof has one simple structure.  The
background $F_\eta$-norm controls the response constant, whereas the
dependence on the truncation scale is the product
\begin{equation}
  \begin{aligned}
    \text{local response}
    &\ \lesssim\ |X|\snorm{A_X}\\[-1mm]
    &\qquad{}\times
    \underbrace{(1+R')^d}_{\text{finite-range geometry}}
    \qquad
    \underbrace{\tailmass(R)}_{\text{discarded tail}}.
  \end{aligned}
  \label{eq:finite_range_geometry_times_tail}
\end{equation}
Thus, for adjacent scales $R'\asymp R$, the elementary finite-range
method multiplies the actual tail profile by $(1+R)^d$.  The background
$F_\eta$ used for propagation need not have the same functional form as
the tail profile.  The power-law consequence was stated in
Corollary~\ref{cor:finite_range_shell_dyadic}; the following result
records concrete nonalgebraic tails.

\begin{corollary}[Rapid tails under the finite-range comparison]
\label{cor:finite_range_rapid_tails}
Let $\eta>d$, let $\snorm{\Phi}_{F_\eta}\le M$, let $R_k=2^kR$ with
$R>1$, and for every finite volume $\Lambda$ under consideration let
$N_\Lambda$ be the first index for which
$R_{N_\Lambda}\ge\diam\Lambda$.  Assume that
$\mathcal S_{R,2}$ in Eq.~\eqref{eq:finite_shell_family}
satisfies Assumption~\ref{assumption:adiabaticity-shell-path}, with every
prescribed sector being the rank-one projection onto the unique ground
state.
\begin{enumerate}
\item Suppose that, for some $A_0,\mu>0$ and $0<\vartheta<1$,
\begin{equation}
  \tailmass(r)\le A_0e^{-\mu r^\vartheta},
  \qquad r\ge1.
  \label{eq:finite_range_stretched_tail_assumption}
\end{equation}
Then
\begin{align}
  d_X(\omega_\Lambda,\omega_{\Lambda,R})
  &\le
  C_{\mu,\vartheta}A_0|X|
  (1+R)^d e^{-\mu R^\vartheta}
  \nonumber\\
  &\le
  C_{\mu',\mu,\vartheta}A_0|X|
  e^{-\mu'R^\vartheta}
  \label{eq:finite_range_stretched_tail}
\end{align}
for every $0<\mu'<\mu$.

\item Suppose that, for some $A_0,\mu>0$,
\begin{equation}
  \tailmass(r)\le A_0e^{-\mu r},
  \qquad r\ge1.
  \label{eq:finite_range_exponential_tail_assumption}
\end{equation}
Then
\begin{align}
  d_X(\omega_\Lambda,\omega_{\Lambda,R})
  &\le
  C_\mu A_0|X|(1+R)^d e^{-\mu R}
  \nonumber\\
  &\le
  C_{\mu',\mu}A_0|X|e^{-\mu'R}
  \label{eq:finite_range_exponential_tail}
\end{align}
for every $0<\mu'<\mu$.

\item Suppose the local tail mass is superpolynomial: for every $q>0$
there is $A_q<\infty$ such that
\begin{equation}
  \tailmass(r)\le A_q(1+r)^{-q},
  \qquad r\ge1.
  \label{eq:finite_range_superpoly_tail_assumption}
\end{equation}
Then, for every $m>0$,
\begin{equation}
  d_X(\omega_\Lambda,\omega_{\Lambda,R})
  \le
  C_m|X|(1+R)^{-m}.
  \label{eq:finite_range_superpoly_tail}
\end{equation}
\end{enumerate}
All constants are independent of $\Lambda$ and $R$; besides the
displayed profile parameters, they depend only on
$d,\eta,M,\gamma$.
\end{corollary}

\begin{proof}
Insert the assumed tail profile into
Eq.~\eqref{eq:finite_range_shell_dyadic_general}.  For
$0<\vartheta\le1$, lattice-scale summation gives
\begin{align}
  &\sum_{k\ge0}
  (1+2^{k+1}R)^d
  e^{-\mu(2^kR)^\vartheta}
  \nonumber\\
  &\qquad\le
  C_{d,\mu,\vartheta}
  (1+R)^d e^{-\mu R^\vartheta}.
  \label{eq:finite_range_rapid_dyadic_sum}
\end{align}
Indeed, after extracting the $k=0$ profile, the remaining series is
bounded by
\begin{equation}
  2^d\sum_{k\ge0}
  2^{kd}e^{-\mu(2^{k\vartheta}-1)}<\infty.
\end{equation}
This proves the first inequalities in
Eqs.~\eqref{eq:finite_range_stretched_tail} and
\eqref{eq:finite_range_exponential_tail}.  Their second inequalities
follow from
\begin{equation}
  \sup_{R>1}
  (1+R)^d e^{-(\mu-\mu')R^\vartheta}<\infty.
\end{equation}

For the superpolynomial case, fix $m>0$ and use
Eq.~\eqref{eq:finite_range_superpoly_tail_assumption} with any
$q>m+d$.  The summand in
Eq.~\eqref{eq:finite_range_shell_dyadic_general} is then bounded by a
constant times
$(1+2^kR)^{-(q-d)}$.  The corresponding geometric series is
$O((1+R)^{-(q-d)})$, which proves
Eq.~\eqref{eq:finite_range_superpoly_tail}.
\end{proof}

\medskip
\noindent\emph{Limitation of the finite-range route.}
The factor $(1+R')^d$ is the volume of the response kernel generated by
the exponentially tilted finite-range Lieb--Robinson estimate.  It
preserves stretched-exponential, exponential, and superpolynomial
decay classes, but adds a polynomial prefactor and loses $d$ powers for
algebraic tails.  Within this tilted finite-range strategy, optimizing
the tilt changes constants but does not remove the volume factor.  The
sharp proof based on Theorem~\ref{thm:tw_improved_lr} avoids this loss
by retaining the long-range spatial decay through the filter
integration.

\section{Intrinsic infinite-volume paths and the shell construction}
\label{app:intrinsic_infinite_volume}

This appendix clarifies the different roles of the direct and shell
paths in infinite volume.  In the quasi-local algebra there is generally
no global Hamiltonian or spectral projection corresponding to a
finite-volume pair $(H_\Lambda(s),P_\Lambda(s))$.  An intrinsic
infinite-volume formulation instead starts with an interaction path
$s\mapsto\Phi_s$ and a path of states $s\mapsto\omega_s$.  The gap is a
property of the GNS Hamiltonian associated with $\omega_s$, and spectral
flow is a cocycle of automorphisms satisfying
\begin{equation}
  \omega_t=\omega_s\circ\alpha_{s,t}.
  \label{eq:intrinsic_state_transport}
\end{equation}
Because the GNS representations at different values of $s$ need not be
unitarily equivalent, Eq.~\eqref{eq:intrinsic_state_transport} does not
follow by transporting spectral projections on one common Hilbert
space.

For every fixed $1<R<R'<\infty$, the shell path
$\Phi_s^{(R,R')}$ and its derivative have range at most $R'$.  Thus a
single shell is not itself a genuinely long-range path.  Existing
infinite-volume spectral-flow results can be applied directly if one
assumes a path of infinite-volume states
$\omega_s^{(R,R')}$ that is locally unique and gapped, is suitably
differentiable in $s$, and has the compatibility needed to concatenate
successive shells
\cite{Bachmann_2011,becker2025automorphicequivalencegappedphases}.  Such
an intrinsic formulation is therefore possible, but these state-path
assumptions are not part of the shell hypotheses used in this work.

Assumption~\ref{assumption:adiabaticity-shell-path} is instead a
finite-volume hypothesis: each selected shell family has a prescribed
isolated sector and a gap uniform in the volume, the path parameter, and
the shell index.  For fixed $1<R<R'$, the
finite-range thermodynamic-limit spectral-flow theorem (Theorem 5.2 in
Ref.\ \onlinecite{Bachmann_2011}) then gives the
automorphism $\alpha_s^{(R,R')}$ used in
Lemma~\ref{lem:limiting_shell_flow}.  In the rank-one setting of
Theorem~\ref{thm:main_infinite_shell}, weak-$*$ limits of the finite-volume
ground states provide its endpoint states.  This route
constructs the relevant infinite-volume branch rather than presupposing
a differentiable branch $s\mapsto\omega_s^{(R,R')}$.

Uniformity in the shell scales is a second reason for using this route.
Although a fixed finite-range shell belongs to every polynomial
interaction class, its higher weighted norms need not be uniform in
$R'$.  For $q>\eta$, one has only the general estimate
\begin{equation}
  \sup_{s\in[0,1]}
  \snorm{\Phi_s^{(R,R')}}_{F_q}
  \le
  (1+R')^{q-\eta}\snorm{\Phi}_{F_\eta}.
  \label{eq:shell_high_norm_growth}
\end{equation}
Consequently, applying an intrinsic theorem separately to each shell
does not by itself produce constants that can be summed uniformly as
$R'\to\infty$.  The finite-volume response estimate used here depends
only on a fixed background norm and the local tail mass, with no factor
growing with $R'$.  Passing that estimate to the thermodynamic limit
preserves the bound needed for the shell Cauchy sequence.  The weaker
finite-range alternative in Appendix~\ref{app:finite_range_shell}
illustrates the possible loss through its additional $(R')^d$ factor.

The direct path presents a different issue.  If the full interaction is
only polynomially local, then $\Phi_s^{(R,\infty)}$ remains genuinely
long range.  Ref.~\cite{teufel2025liebrobinsonboundsautomorphicequivalence}
proves polynomial Lieb--Robinson bounds and the locality of the
inverse Liouvillian, and uses them to construct a polynomially local
spectral flow for finite-lattice Hamiltonians with uniform estimates.
This is the locality input used in the finite-volume part of the present
work.  It does not by itself provide the intrinsic state transport
\eqref{eq:intrinsic_state_transport} for a path specified only through
infinite-volume GNS gaps.

The intrinsic infinite-volume result of
Ref.~\cite{becker2025automorphicequivalencegappedphases} establishes
that transport for differentiable paths of locally unique gapped states
with superpolynomially decaying interactions.  Its proof must control
the infinite-volume dynamics and inverse Liouvillian on a quasi-local
interaction space, show that the resulting interaction generates an
automorphism cocycle, and prove the parallel-transport identity between
states in potentially inequivalent GNS representations.

The restriction to superpolynomial decay enters through the locality
losses in these steps.  For example, the cited proof controls an
observable seminorm of order $\nu$ for the automorphism cocycle using an
interaction seminorm of order $4\nu+9d+4$; the inverse-Liouvillian and
commutator estimates likewise require stronger seminorms than the one
being estimated.  The space of superpolynomial interactions contains
all polynomially weighted seminorms and is therefore stable under these
losses.  A fixed power-law class provides only a finite decay exponent,
and the available estimates do not show that sufficient decay remains
to construct the generator and cocycle in compatible interaction
spaces and then prove the state-transport identity.  Thus the missing
input is not a polynomial Lieb--Robinson bound, but closure of the full
infinite-volume argument within a fixed polynomial decay class.  Such
an extension may be possible with sharper decay bookkeeping and
additional state-path assumptions, but it is not supplied by the
results invoked here.

Thus the asymmetry in the main text is one of hypotheses and proof
strategy.  A fixed shell can be treated intrinsically after adding an
infinite-volume gapped state path, whereas our thermodynamic-limit
construction derives the required states from finite-volume data and
retains estimates uniform in the shell scale.  The intrinsic direct
construction is used only in the superpolynomial setting where a
intrinsic infinite-volume automorphic-equivalence theorem is available.

\section{Uniform inverse-Liouvillian estimate in the superpolynomial setting}
\label{app:inverse_liouvillian}

The infinite-volume direct theorem needs a quantitative form of
Lemma~B.6 of Ref.~\cite{becker2025automorphicequivalencegappedphases}.  That lemma
is formulated for a Fr\'echet space of superpolynomial interactions.
Here we keep track of the finitely many seminorms that enter its proof
and of the additional factor $|Z|$ in our diameter norm.
The proof first controls each site-centered inverse-Liouvillian
contribution in localized-observable seminorms, uniformly in the
truncation scale and path parameter.  A ball-shell reconstruction then
converts those estimates to the diameter interaction norms used in the
main text.

We first fix notation.  The reference uses the maximum metric on
$\mathbb Z^d$; write $\diam_\infty$ for the corresponding diameter and
set
\begin{equation}
  \snorm{\Theta}^{\mathrm{BTW}}_\ell
  :=
  \sup_{x\in\Gamma}
  \sum_{Z\ni x}
  (1+\diam_\infty Z)^\ell\snorm{\Theta(Z)}.
  \label{eq:btw_interaction_seminorm}
\end{equation}
Because
\begin{equation}
  \diam_\infty Z\le\diam Z\le d\,\diam_\infty Z
  \quad\text{and}\quad |Z|\ge1,
\end{equation}
the seminorms defined with the two metrics are equivalent up to constants
depending only on $d$ and $\ell$, and in particular
\begin{equation}
  \snorm{\Theta}^{\mathrm{BTW}}_\ell
  \le
  \snorm{\Theta}_{F_\ell}.
  \label{eq:btw_norm_dominated}
\end{equation}
Let
\begin{equation}
  B_k^\infty(x)
  =
  \{y\in\Gamma:\snorm{y-x}_\infty\le k\},
\end{equation}
and let $\mathbb E_{k,x}$ be the normalized-trace conditional
expectation onto $\mA_{B_k^\infty(x)}$.  For a quasi-local observable
$B$, define
\begin{equation}
  \snorm{B}^{\mathrm{loc}}_{\rho,x}
  :=
  \snorm{B}
  +
  \sup_{k\in\mathbb N_0}
  (1+k)^\rho
  \snorm{B-\mathbb E_{k,x}(B)}.
  \label{eq:localized_observable_seminorm}
\end{equation}
These are the localized-observable seminorms of
Ref.~\cite{becker2025automorphicequivalencegappedphases}.

\begin{lemma}[Uniform seminorm bound]
\label{lem:uniform_infinite_inverse}
For $\nu\in\mathbb N_0$, define
\begin{align}
  \rho_\nu&=\nu+2d+2,
  &
  q_\nu&=2\nu+5d+5,
  \nonumber\\
  a_\nu&=4\nu+17d+12.
  \label{eq:inverse_indices}
\end{align}
There is a constant
$C_\nu=C_\nu(d,\gamma,\snorm{\Phi}_{F_{a_\nu}})$ such that, along the
family of direct truncation paths,
\begin{equation}
  \sup_{R,s}
  \snorm{\mathcal I_s^{(R)}(\Theta)}_{F_\nu}
  \le
  C_\nu\snorm{\Theta}_{F_{q_\nu}}
  \label{eq:appendix_inverse_bound}
\end{equation}
for every superpolynomial interaction $\Theta$.  In particular, the
constant is independent of $R$, $s$, and of the continuity constants
of the state path.
\end{lemma}

\begin{proof}
Set
\begin{equation}
  g_*=\tfrac12\min\{\gamma,1\}
\end{equation}
and fix the inverse-Liouvillian filter $W=W_{g_*}$ of
Ref.~\cite{becker2025automorphicequivalencegappedphases}.  The same filter can be
used for every $R>1$ and $s$, and it satisfies
\begin{equation}
  \int_{\mathbb R}|W(t)|(1+|t|)^N\,\dd t<\infty
  \qquad\text{for every }N\in\mathbb N_0.
  \label{eq:filter_moments}
\end{equation}

Choose deterministically a center $c(Z)\in Z$ for every nonempty
finite $Z$ and write
\begin{equation}
  H_{s,x}^{(R)}
  =
  \sum_{Z:\,c(Z)=x}\Phi_s^{(R,\infty)}(Z).
  \label{eq:centered_background_observable}
\end{equation}
This is the site-centered observable in Definition~2.5 of the
reference.  Terms with $\diam_\infty Z\le k$ and center $x$ are
supported in $B_k^\infty(x)$.  Splitting the remaining sum into its
tail and using contractivity of the conditional expectation gives,
for every $r\in\mathbb N_0$,
\begin{equation}
  \sup_x
  \snorm{H_{s,x}^{(R)}}_{r,x}^{\mathrm{loc}}
  \le
  3\snorm{\Phi_s^{(R,\infty)}}^{\mathrm{BTW}}_r
  \le
  3\snorm{\Phi}_{F_r}.
  \label{eq:centered_background_bound}
\end{equation}
The last inequality is uniform in $R$ and $s$.

For an interaction $\Theta$, let
\begin{equation}
  Q_{s,x}^{(R)}(\Theta)
  =
  -\ii\int_{\mathbb R}W(t)\int_0^t
  e^{\ii u\mL_{\Phi_s^{(R,\infty)}}}
  \mL_\Theta\bigl(H_{s,x}^{(R)}\bigr)
  \,\dd u\,\dd t.
  \label{eq:inverse_local_observable}
\end{equation}
The integral is absolutely convergent in every localized seminorm, as
we now verify.  Lemma~2.7 of
Ref.~\cite{becker2025automorphicequivalencegappedphases} gives, for every
$\rho\in\mathbb N_0$,
\begin{equation}
  \snorm{
    \mL_\Theta\bigl(H_{s,x}^{(R)}\bigr)
  }_{\rho,x}^{\mathrm{loc}}
  \le
  c_\rho
  \snorm{\Theta}^{\mathrm{BTW}}_{d+1+2\rho}
  \snorm{H_{s,x}^{(R)}}_{d+3+2\rho,x}^{\mathrm{loc}}.
  \label{eq:inverse_commutator_bound}
\end{equation}
Lemma~C.4 of that reference gives an increasing, polynomially bounded
function $b_\rho$ such that
\begin{equation}
  \snorm{
    e^{\ii u\mL_{\Phi_s^{(R,\infty)}}}B
  }_{\rho,x}^{\mathrm{loc}}
  \le
  b_\rho(|u|)
  \snorm{B}_{\rho,x}^{\mathrm{loc}}.
  \label{eq:uniform_local_dynamics}
\end{equation}
The same $b_\rho$ works for every $R>1$ and $s$: according to the
uniformity statement in that lemma, it is enough to bound
$\snorm{\Phi_s^{(R,\infty)}}^{\mathrm{BTW}}_{4\rho+9d+4}$, and
Eqs.~\eqref{eq:uniform_background_superpoly} and
\eqref{eq:btw_norm_dominated} give
\begin{equation}
  \sup_{R,s}
  \snorm{\Phi_s^{(R,\infty)}}^{\mathrm{BTW}}_{4\rho+9d+4}
  \le
  \snorm{\Phi}_{F_{4\rho+9d+4}}.
  \label{eq:uniform_btw_background}
\end{equation}
Define
\begin{equation}
  J_\rho
  =
  \int_{\mathbb R}|W(t)|
  \int_0^{|t|}b_\rho(u)\,\dd u\,\dd t.
  \label{eq:filter_dynamics_integral}
\end{equation}
This number is finite by Eq.~\eqref{eq:filter_moments} and the
polynomial growth of $b_\rho$.  Equations
\eqref{eq:centered_background_bound}--\eqref{eq:filter_dynamics_integral}
therefore yield the explicit estimate
\begin{align}
  &\sup_{R,s,x}
  \snorm{Q_{s,x}^{(R)}(\Theta)}_{\rho,x}^{\mathrm{loc}}
  \nonumber\\
  &\quad\le
  3c_\rho J_\rho
  \snorm{\Phi}_{F_{d+3+2\rho}}
  \snorm{\Theta}_{F_{d+1+2\rho}}.
  \label{eq:inverse_local_bound}
\end{align}

It remains to reconstruct an interaction and to account for the factor
$|Z|$ in our norm.  Suppress $R,s,\Theta$ temporarily and set
\begin{align}
  Q_{x,0}&=\mathbb E_{0,x}(Q_x),
  \\
  Q_{x,k}&=
  \bigl(\mathbb E_{k,x}-\mathbb E_{k-1,x}\bigr)(Q_x),
  \qquad k\ge1.
  \label{eq:inverse_shell_decomposition}
\end{align}
The series $\sum_{k\ge0}Q_{x,k}$ converges in norm to $Q_x$.  Define
$\mathcal I_s^{(R)}(\Theta)$ by assigning $Q_{x,k}$ to
$B_k^\infty(x)$, exactly as in Lemma~B.6 of the reference.  From
Eq.~\eqref{eq:localized_observable_seminorm}, for $k\ge1$,
\begin{align}
  \snorm{Q_{x,k}}
  &\le
  \snorm{Q_x-\mathbb E_{k,x}(Q_x)}
  +
  \snorm{Q_x-\mathbb E_{k-1,x}(Q_x)}
  \nonumber\\
  &\le
  C_\rho(1+k)^{-\rho}
  \snorm{Q_x}_{\rho,x}^{\mathrm{loc}}.
  \label{eq:inverse_shell_decay}
\end{align}
Moreover,
\begin{equation}
  |B_k^\infty(x)|=(2k+1)^d,
  \qquad
  \diam B_k^\infty(x)\le2dk,
  \label{eq:max_ball_geometry}
\end{equation}
and, for fixed $y$, there are exactly $(2k+1)^d$ centers $x$ for which
$y\in B_k^\infty(x)$.  Consequently,
\begin{align}
  \snorm{\mathcal I_s^{(R)}(\Theta)}_{F_\nu}
  &\le
  C_{d,\nu,\rho}
  \left[
    1+
    \sum_{k\ge1}(1+k)^{2d+\nu-\rho}
  \right]
  \nonumber\\
  &\qquad\times
  \sup_x
  \snorm{Q_{s,x}^{(R)}(\Theta)}_{\rho,x}^{\mathrm{loc}}.
  \label{eq:inverse_interaction_from_local}
\end{align}
This is the step at which the additional factor $|Z|$ costs one extra
volume factor $(1+k)^d$ compared with the interaction norm used in the
reference.

Take $\rho=\rho_\nu=\nu+2d+2$.  The series in
Eq.~\eqref{eq:inverse_interaction_from_local} is then bounded by
$\sum_{k\ge1}(1+k)^{-2}$.  Furthermore,
\begin{align}
  d+1+2\rho_\nu&=q_\nu,
  \\
  4\rho_\nu+9d+4&=a_\nu,
\end{align}
and $d+3+2\rho_\nu\le a_\nu$.  Combining
Eqs.~\eqref{eq:inverse_local_bound} and
\eqref{eq:inverse_interaction_from_local}, and using monotonicity of
the weighted norms in their exponent, proves
Eq.~\eqref{eq:appendix_inverse_bound}.  The constants used in the
state-differentiability assumption never entered the estimate.
\end{proof}

\begin{lemma}[Identification of the reconstructed derivation]
\label{lem:inverse_reconstructed_derivation}
Let $\mathcal I_s^{(R)}(\Theta)$ be the interaction reconstructed in
the proof of Lemma~\ref{lem:uniform_infinite_inverse}.  For every local
observable $A$,
\begin{align}
  &\mL_{\mathcal I_s^{(R)}(\Theta)}(A)
  \nonumber\\
  &\quad=
  \int_{\mathbb R}W(t)
  e^{\ii t\mL_{\Phi_s^{(R,\infty)}}}
  \mL_\Theta\!\left(
    e^{-\ii t\mL_{\Phi_s^{(R,\infty)}}}A
  \right)\,\dd t.
  \label{eq:inverse_reconstructed_derivation}
\end{align}
Thus the reconstructed interaction generates the inverse-Liouvillian
derivation used in the automorphic-equivalence theorem of
Ref.~\cite{becker2025automorphicequivalencegappedphases}.
\end{lemma}

\begin{proof}
Use the same deterministic centers as in
Eq.~\eqref{eq:centered_background_observable} and set
\begin{equation}
  \Theta_y=\sum_{Z:\,c(Z)=y}\Theta(Z).
  \label{eq:centered_perturbation_observable}
\end{equation}
For a local observable, the corresponding derivations can be written
as convergent sums of commutators with the centered observables.  The
shell reconstruction in Eq.~\eqref{eq:inverse_shell_decomposition}
also gives
\begin{equation}
  \mL_{\mathcal I_s^{(R)}(\Theta)}(A)
  =
  \sum_x[Q_{s,x}^{(R)}(\Theta),A].
  \label{eq:inverse_centered_derivation}
\end{equation}
Absolute convergence follows from
Eq.~\eqref{eq:appendix_inverse_bound} with $\nu=0$ and from the
localized-observable estimate \eqref{eq:inverse_local_bound}.

Suppress $R$ and $s$ temporarily and write
$\mL_H=\mL_{\Phi_s^{(R,\infty)}}$.  Antisymmetry of the commutator and
the centered decompositions give, within the absolutely convergent
commutator sum for a local $A$,
\begin{align}
  &\sum_x
  [e^{\ii u\mL_H}\mL_\Theta(H_{s,x}^{(R)}),A]
  \nonumber\\
  &\quad=
  \sum_{x,y}
  [e^{\ii u\mL_H}[\Theta_y,H_{s,x}^{(R)}],A]
  \nonumber\\
  &\quad=-\sum_y
  [e^{\ii u\mL_H}\mL_H(\Theta_y),A].
  \label{eq:inverse_centered_jacobi}
\end{align}
Substituting this relation into
Eq.~\eqref{eq:inverse_local_observable} and using
\begin{equation}
  \frac{\dd}{\dd u}
  e^{\ii u\mL_H}(\Theta_y)
  =
  \ii e^{\ii u\mL_H}\mL_H(\Theta_y)
\end{equation}
gives
\begin{align}
  \mL_{\mathcal I_s^{(R)}(\Theta)}(A)
  ={}&
  \sum_y
  \left[
    \int_{\mathbb R}W(t)
    \left(e^{\ii t\mL_H}(\Theta_y)-\Theta_y\right)\,\dd t,
    A
  \right]
  \nonumber\\
  ={}&
  \sum_y
  \left[
    \int_{\mathbb R}W(t)
    e^{\ii t\mL_H}(\Theta_y)\,\dd t,
    A
  \right],
  \label{eq:inverse_centered_integral_identity}
\end{align}
where the last equality uses that $W$ is odd and integrable, so
$\int_{\mathbb R}W(t)\,\dd t=0$.  Finally, covariance of the
commutator under the dynamics yields
\begin{align}
  \sum_y
  [e^{\ii t\mL_H}(\Theta_y),A]
  &={}
  e^{\ii t\mL_H}
  \sum_y
  [\Theta_y,e^{-\ii t\mL_H}(A)]
  \nonumber\\
  &={}
  e^{\ii t\mL_H}
  \mL_\Theta(e^{-\ii t\mL_H}(A)).
  \label{eq:inverse_commutator_covariance}
\end{align}
Equations~\eqref{eq:inverse_centered_derivation},
\eqref{eq:inverse_centered_integral_identity}, and
\eqref{eq:inverse_commutator_covariance} prove
Eq.~\eqref{eq:inverse_reconstructed_derivation}.  The exchanges of
sums and integrals are justified by
Eqs.~\eqref{eq:inverse_commutator_bound}--\eqref{eq:filter_dynamics_integral};
these are the same absolute-convergence estimates used in the
``inverse Liouvillian on derivations'' lemma of the cited reference.
For even fermionic interactions, all centered observables are even and
the same computation holds in the CAR algebra.
\end{proof}

For the direct truncation path,
$\Theta=\dot\Phi_s^{(R,\infty)}=\Phi_{>R}$.  Hence
Lemma~\ref{lem:inverse_reconstructed_derivation} identifies
$\Psi_s^{(R)}=-\mathcal I_s^{(R)}(\Phi_{>R})$ with the generator
interaction of the automorphic transport in
Eq.~\eqref{eq:infinite_cocycle_generator}; the norm estimate in
Lemma~\ref{lem:uniform_infinite_inverse} therefore applies to that
generator, rather than merely to an interaction with the same local
decomposition.

Taking $\nu=0$ gives $q_0=5d+5$ and $a_0=17d+12$.  Thus
Eq.~\eqref{eq:appendix_inverse_bound}, with
$\Theta=\Phi_{>R}$, gives Eq.~\eqref{eq:uniform_inverse_liouvillian}
with
\begin{equation}
  \eta_0=5d+5.
\end{equation}
For every $m>0$,
\begin{equation}
  \snorm{\Phi_{>R}}_{F_{q_0}}
  \le
  (1+R)^{-m}\snorm{\Phi}_{F_{q_0+m}},
\end{equation}
which is the quantitative input used in the proof of
Theorem~\ref{thm:main_infinite_direct}.

\section{An exponent-sharp example without translation invariance}
\label{app:sharpness}

We give a simple non-translation-invariant model showing that no exponent
better than $p-d$ can hold uniformly over the full class covered by
Eq.~\eqref{eq:two_body_final}.  The example does not address optimality
under additional structural assumptions, such as translation invariance.
Place a distinguished spin at the origin and one spin at every $x\ne0$.
Set $\rho(x):=\dist(0,x)$ and define an interaction $\Phi^\sharp$ by
\begin{align}
  \Phi^\sharp(\{0\})&=-hX_0,
  &\Phi^\sharp(\{x\})&=-BZ_x \quad (x\ne0),
  \nonumber\\
  \Phi^\sharp(\{0,x\})&=-\lambda J_xZ_0Z_x,
  &&x\ne0,
  \nonumber\\
  J_x&=(1+\rho(x))^{-p},
  &&x\ne0,
  \label{eq:sharp_model}
\end{align}
with all other terms zero, where $h,\lambda>0$ and
$B>\lambda\sup_{x\ne0}J_x$.  Its range-$R$ truncation is
$\Phi^\sharp_{\le R}$ in the notation of
Section~\ref{subsec:setup}.

For $p>d$ and $R\in(1,\infty]$, set
$S_R:=\sum_{0<\rho(x)\le R}J_x$ and let $\omega_R$ be the product state
with every outer spin in the $Z_x=+1$ eigenstate and the central spin in
the ground state of $-hX_0-\lambda S_RZ_0$.  Set
$\omega:=\omega_\infty$.  For finite $R$, $\omega_R$ is a ground state
of $\Phi^\sharp_{\le R}$, while $\omega$ is a ground state of
$\Phi^\sharp$.

In every finite-volume restriction containing the origin, all outer
spins are polarized in their $Z_x=+1$ states.  Indeed, changing an outer
eigenvalue from $-1$ to $+1$ lowers the energy by at least
$2(B-\lambda J_x)$.  The ground state is therefore unique, with gap at
least $\min\{2h,2(B-\lambda\sup_{x\ne0}J_x)\}>0$.  This bound is
uniform in the volume and in $R$ and also holds along the direct and
shell paths, because every interpolated bond strength lies between $0$
and $\lambda J_x$.  Moreover,
\begin{equation}
  \omega_R(Z_0)
  =
  \frac{\lambda S_R}{\sqrt{h^2+\lambda^2S_R^2}}.
\end{equation}
The derivative of the displayed function with respect to $S_R$ is
strictly positive at $S_\infty$.  Lattice-shell counting therefore gives
\begin{equation}
  \abs{\omega(Z_0)-\omega_R(Z_0)}
  \asymp
  S_\infty-S_R
  \asymp
  R^{-(p-d)}.
\end{equation}
Furthermore, $\snorm{\Phi^\sharp}_{F_\eta}<\infty$ whenever
$\eta<p-d$.  Thus, if $p>2d$, one may choose $d<\eta<p-d$, so the
example lies in the class allowed by our hypotheses and attains the
exponent $p-d$.  Sharpness for translation-invariant or otherwise
restricted subclasses is not claimed.

\section{Explanation of fast decay of local observable error in massive fermion model}\label{app:massive_gaussian_rate}
The rapid convergence is not a consequence of Gaussianity alone; it
comes from the Fourier-space response of the occupied-band projector.
For the real, even hopping used here,
$t_{-r}=t_r\in\mathbb R$, so $q_R(k)$ is real.  Set
\[
  \begin{split}
  F_m(q) &:=
  \frac12\left(1-\frac{m}{\sqrt{m^2+q^2}}\right),
  \\
  \tau_R(k)
  &:=
  q_\infty(k)-q_R(k)
  =
  \sum_{|r|>R}t_r e^{\ii k\cdot r},
  \end{split}
\]
and
\[
  g_m(k)
  :=
  F_m'\bigl(q_\infty(k)\bigr)
  =
  \frac{m q_\infty(k)}
       {2[m^2+q_\infty(k)^2]^{3/2}}.
\]
Because $m\ne0$, $F_m''$ is bounded on the real line.  Taylor's
theorem and $q_R=q_\infty-\tau_R$ therefore give
\begin{equation}
\begin{split}
  F_m(q_R(k))-F_m(q_\infty(k))
  &=
  -g_m(k)\tau_R(k)+\mathcal E_R(k),
  \\
  |\mathcal E_R(k)|
  &\le C_m|\tau_R(k)|^2,
  \label{eq:massive_projector_taylor}
  \end{split}
\end{equation}
where $C_m$ is independent of $R$ and $k$.

Let \(Q\) be the Laurent operator
\[
(Q\psi)_x=\sum_{r}t_r\psi_{x-r},
\]
whose Bloch symbol is \(q_\infty(k)\), and define
\[
f(z)=\frac{mz}{2(m^2+z^2)^{3/2}}.
\]
Then \(g_m(k)=f(q_\infty(k))\) is the Bloch symbol of \(f(Q)\).
Since \(Q\) is bounded and self-adjoint and \(m\neq0\), one may choose
the branch of \((m^2+z^2)^{-3/2}\) that is holomorphic on a neighborhood
of \(\sigma(Q)\).
Lemma~\ref{lem:jaffard_functional_calculus} therefore gives
\[
|[f(Q)]_{xy}|
\le C(1+\operatorname{dist}(x,y))^{-p}.
\]
Because \([f(Q)]_{xy}=\widehat g_m(x-y)\), we conclude that
\begin{equation}
|\widehat g_m(r)|\le C(1+|r|)^{-p}.
\label{eq:massive_response_coefficient_decay}
\end{equation}
Integrating Eq.~\eqref{eq:massive_projector_taylor} over the
Brillouin zone gives
\begin{equation}
\begin{split}
  \Delta\langle n_{a,x}\rangle_R
  &=
  -\sum_{|r|>R}t_r\,\widehat g_m(-r)
  +\mathcal R_R,
  \\
  |\mathcal R_R|
  &\le
  C_m
  \int_{\mathrm{BZ}}
  \frac{\dd^d k}{(2\pi)^d}
  |\tau_R(k)|^2
  =
  C_m\sum_{|r|>R}|t_r|^2,
\end{split}
\label{eq:massive_density_response}
\end{equation}
where the last identity is Parseval's identity.  Consequently,
\begin{align}
  \left|
    \sum_{|r|>R}t_r\,\widehat g_m(-r)
  \right|
  &\le
  C\sum_{|r|>R}(1+|r|)^{-2p},
  \\
  |\mathcal R_R|
  &\le
  C\sum_{|r|>R}(1+|r|)^{-2p},
\end{align}
and hence
\begin{equation}
  \left|
    \Delta\langle n_{a,x}\rangle_R
  \right|
  \le
  C R^{-(2p-d)}.
  \label{eq:massive_density_upper_bound}
\end{equation}

Multiplying Eq.~\eqref{eq:massive_projector_taylor} by
$e^{\ii k\cdot\ell}$ and integrating over the Brillouin zone gives the
corresponding estimate for every fixed Fourier coefficient of
$[P_R(k)]_{aa}$.  Since $\abs{e^{\ii k\cdot\ell}}=1$, its remainder is
again bounded by
\[
  C_m\int_{\mathrm{BZ}}
  \frac{\dd^d k}{(2\pi)^d}|\tau_R(k)|^2
  =
  C_m\sum_{|r|>R}|t_r|^2.
\]
On the other hand,
\[
\begin{aligned}\left|\sum_{|r|>R}t_r\,\widehat g_m(-r-\ell)\right|&\le C_\ell\sum_{|r|>R}(1+|r|)^{-p}(1+|r+\ell|)^{-p}\\&\le C_\ell'\sum_{|r|>R}(1+|r|)^{-2p}\\&=O\!\left(R^{-(2p-d)}\right).\end{aligned}
\]
Thus, for every fixed
$\ell\ne0$, we have
\[
\begin{split}
  \left|
    c_R(\ell)-c_\infty(\ell)
  \right|
  &=
  O\!\left(R^{-(2p-d)}\right),
  \\
  c_R(\ell)
  &:=
  \int_{\mathrm{BZ}}
  \frac{\dd^d k}{(2\pi)^d}
  e^{\ii k\cdot\ell}[P_R(k)]_{aa}.
  \end{split}
\]
Combining this estimate with
Eq.~\eqref{eq:gaussian_density_pair_observable} yields
\begin{equation}
  \left|
    \Delta\langle
      n_{a,x}n_{a,x+\ell}
    \rangle_R
  \right|
  =
  O\!\left(R^{-(2p-d)}\right)
  \label{eq:massive_density_pair_upper_bound}
\end{equation}
for every fixed $\ell\ne0$, including nearest neighbors.

Both the linear term and the integrated Taylor remainder are bounded
at the scale $R^{-(2p-d)}$.  These estimates therefore establish an
upper bound, not an asymptotic equivalence, and do not identify the
linear term as the leading contribution.  For $d=1$ and $p=2.25$, the
upper-bound exponent is $2p-d=3.5$, compatible with the faster decay
observed over the numerical fitting window.

For the parameters used here,
\begin{equation}
  q_R(k)\ge
  q_*:=1-2\sum_{r\ge1}(1+r)^{-p}>0
  \label{eq:massive_symbol_lower_bound}
\end{equation}
uniformly in $R$ and along the direct and shell interpolations.  At
$m=0$, this implies
$P_R(k)=(\mathbf1-\sigma_x)/2$ for every $R$, so all truncation errors
vanish exactly.  With $q_*$ fixed, the response coefficients above are
$O(m)$ as $m\to0$.  Within this family, a slower crossover can arise
when the off-diagonal symbol is also tuned so that
$\min_k|q_\infty(k)|$ becomes small.  The constants in the
Fourier-decay estimates can then become large, and the projector may
develop sharp momentum dependence near the incipient band touching.
This is the regime targeted by the two near-critical benchmarks.

\begin{lemma}[Polynomial decay under holomorphic functional calculus]
\label{lem:jaffard_functional_calculus}
Let $p>d$, and let $Q$ be a bounded self-adjoint Laurent operator on
$\ell^2(\Gamma)$, where $\Gamma=\mathbb Z^d$.  Write
\[
Q_{xy}:=\langle\delta_x,Q\delta_y\rangle,
\]
where $\delta_x(z)=\mathbf1_{\{z=x\}}$ denotes the canonical basis
vector of $\ell^2(\Gamma)$ localized at $x$, and
suppose that
\begin{equation}
  \abs{Q_{xy}}
  \le
  C_Q\bigl(1+\dist(x,y)\bigr)^{-p},
  \qquad x,y\in\Gamma.
  \label{eq:jaffard_decay}
\end{equation}
If $f$ is holomorphic on an open neighborhood of the spectrum $\sigma(Q)$, then
$f(Q)$ is again a Laurent operator and there exists
$C_{Q,f}<\infty$ such that
\begin{equation}
  \abs{[f(Q)]_{xy}}
  \le
  C_{Q,f}\bigl(1+\dist(x,y)\bigr)^{-p},
  \qquad x,y\in\Gamma.
  \label{eq:jaffard_functional_decay}
\end{equation}
\end{lemma}

\begin{proof}
Let $\mathcal J_p$ be the Jaffard algebra with norm
\begin{equation}
  \snorm{A}_{\mathcal J_p}
  :=
  \sup_{x,y\in\Gamma}
  \bigl(1+\dist(x,y)\bigr)^p\abs{A_{xy}}.
\end{equation}
The hypothesis says that $Q\in\mathcal J_p$. 
For \(p>d\), this space is complete and has continuous multiplication; after an equivalent renorming, it is a Banach algebra. By Jaffard's inverse-closedness theorem~\cite{Jaffard_1990}, it is spectrally invariant in \(\mathcal B(\ell^2(\Gamma))\), and is therefore closed under holomorphic functional calculus.

Therefore,
\[
  (z\mathbf1-Q)^{-1}\in\mathcal J_p
  \qquad\text{for every }z\notin\sigma(Q).
\]
Choose a finite union $\mathcal C$ of admissible contours contained in
the holomorphy domain of $f$ and enclosing $\sigma(Q)$.  Continuity of
inversion in $\mathcal J_p$ and compactness of $\mathcal C$ imply that
the resolvents are uniformly bounded in $\mathcal J_p$ along
$\mathcal C$.  Hence the holomorphic functional-calculus formula
\begin{equation}
  f(Q)=
  \frac{1}{2\pi\ii}
  \int_{\mathcal C}
  f(z)(z\mathbf1-Q)^{-1}\,\dd z
\end{equation}
defines an element of $\mathcal J_p$, which proves
Eq.~\eqref{eq:jaffard_functional_decay}.  Since $Q$ commutes with
lattice translations, so does $f(Q)$; thus $f(Q)$ is again Laurent.
\end{proof}

\end{document}